\documentclass[10pt]{article}
\usepackage{hyperref}
\usepackage{amssymb}
\usepackage{amsmath}
 
 \usepackage{amsmath}
\usepackage{latexsym}
\usepackage{amssymb}

\usepackage{amsthm}

\font\tengoth=eufm10 at 10pt
\font\sevengoth=eufm7 at 6pt
\newfam\gothfam
\textfont\gothfam=\tengoth
\scriptfont\gothfam=\sevengoth

\newcommand{\mlabel}[1]{\marginpar{#1}\label{#1}}

\newcommand{\g}{{\mathfrak g}}

\newcommand{\h}{{\mathfrak h}}
\newcommand{\z}{{\mathfrak z}}

\newcommand{\fa}{{\mathfrak a}}

\newcommand{\fc}{{\mathfrak c}}

\newcommand{\fe}{{\mathfrak e}}

\newcommand{\fg}{{\mathfrak g}}
\newcommand{\fh}{{\mathfrak h}}

\newcommand{\fj}{{\mathfrak j}}
\newcommand{\fk}{{\mathfrak k}}

\newcommand{\fn}{{\mathfrak n}}

\newcommand{\fq}{{\mathfrak q}}
\newcommand{\fp}{{\mathfrak p}}

\newcommand{\fs}{{\mathfrak s}}
\newcommand{\ft}{{\mathfrak t}}
\newcommand{\fu}{{\mathfrak u}}

\newcommand{\fz}{{\mathfrak z}}

\renewcommand\sp{\mathfrak {sp}}

\renewcommand{\:}{\colon}
\newcommand{\1}{\mathbf{1}}

\newcommand{\cE}{\mathcal{E}}
\newcommand{\cF}{\mathcal{F}}

\newcommand{\cH}{\mathcal{H}}

\newcommand{\cM}{\mathcal{M}}

\newcommand{\cO}{\mathcal{O}}

\newcommand{\cR}{\mathcal{R}}
\newcommand{\cS}{\mathcal{S}}
\newcommand{\cT}{\mathcal{T}}

\newcommand{\cV}{\mathcal{V}}
\newcommand{\cW}{\mathcal{W}}

\newcommand\bx{{\bf{x}}}
\newcommand\by{{\bf{y}}}

\newcommand{\bO}{\mathbf{O}}

\newcommand{\eset}{\emptyset}

\newcommand{\trile}{\trianglelefteq}
\newcommand{\subeq}{\subseteq}
\newcommand{\supeq}{\supseteq}

\newcommand{\into}{\hookrightarrow}
\newcommand{\eps}{\varepsilon}

\newcommand{\N}{{\mathbb N}}
\newcommand{\Z}{{\mathbb Z}}
\newcommand{\R}{{\mathbb R}}
\newcommand{\C}{{\mathbb C}}

\newcommand{\K}{{\mathbb K}}

\renewcommand{\H}{{\mathbb H}}
\newcommand{\T}{{\mathbb T}}

\newcommand{\bS}{{\mathbb S}}

\renewcommand{\tilde}{\widetilde}

\renewcommand{\L}{\mathop{\bf L{}}\nolimits}

\newcommand{\GL}{\mathop{{\rm GL}}\nolimits}
\newcommand{\SL}{\mathop{{\rm SL}}\nolimits}

\newcommand{\AU}{\mathop{{\rm AU}}\nolimits}

\newcommand{\PSL}{\mathop{{\rm PSL}}\nolimits}
\newcommand{\SO}{\mathop{{\rm SO}}\nolimits}

\newcommand{\OO}{\mathop{\rm O{}}\nolimits}
\newcommand{\UU}{\mathop{\rm U{}}\nolimits}
\newcommand{\U}{\mathop{\rm U{}}\nolimits}

\newcommand{\Sym}{\mathop{{\rm Sym}}\nolimits}

\newcommand{\Skew}{\mathop{{\rm Skew}}\nolimits}

\newcommand{\gl}  {\mathop{{\mathfrak{gl} }}\nolimits}

\newcommand{\fsl} {\mathop{{\mathfrak{sl} }}\nolimits}

\newcommand{\su}  {\mathop{{\mathfrak{su} }}\nolimits}
\newcommand{\so}  {\mathop{{\mathfrak{so} }}\nolimits}

\newcommand{\Exp}{\mathop{{\rm Exp}}\nolimits}
\newcommand{\Fix}{\mathop{{\rm Fix}}\nolimits}

\newcommand{\ad}{\mathop{{\rm ad}}\nolimits}
\newcommand{\Ad}{\mathop{{\rm Ad}}\nolimits}

\renewcommand{\Im}{\mathop{{\rm Im}}\nolimits}

\newcommand{\tr}{\mathop{{\rm tr}}\nolimits}

\newcommand{\Herm}{\mathop{{\rm Herm}}\nolimits}
\newcommand{\Aherm}{\mathop{{\rm Aherm}}\nolimits}

\newcommand{\Aut}{\mathop{{\rm Aut}}\nolimits}

\newcommand{\diag}{\mathop{{\rm diag}}\nolimits}
\newcommand{\End}{\mathop{{\rm End}}\nolimits}
\newcommand{\id}{\mathop{{\rm id}}\nolimits}
\newcommand{\rk}{\mathop{{\rm rank}}\nolimits}

\renewcommand{\dim}{\mathop{{\rm dim}}\nolimits}

\newcommand{\im}{\mathop{{\rm im}}\nolimits}

\newcommand{\Inn}{\mathop{{\rm Inn}}\nolimits}
\newcommand{\Int}{\mathop{{\rm int}}\nolimits}

\newcommand{\cone}{\mathop{{\rm cone}}\nolimits}

\newcommand{\dS}{\mathop{{\rm dS}}\nolimits}

\renewcommand{\phi}{\varphi}

\newcommand{\Rarrow}{\Rightarrow}
\newcommand{\nin}{\noindent} 
\newcommand{\oline}{\overline}

\newcommand{\la}{\langle}
\newcommand{\ra}{\rangle}

\newcommand{\up}{\mathop{\uparrow}}

\newcommand{\res}{\vert}

\newcommand{\spann}{{\rm span}}

\newcommand{\Spec}{{\rm Spec}}

\newcommand{\ssssarr}{\hbox to 15pt{\rightarrowfill}}
\newcommand{\sssarr}{\hbox to 20pt{\rightarrowfill}}
\newcommand{\ssarr}{\hbox to 30pt{\rightarrowfill}}
\newcommand{\sarr}{\hbox to 40pt{\rightarrowfill}}
\newcommand{\arr}{\hbox to 60pt{\rightarrowfill}}
\newcommand{\larr}{\hbox to 60pt{\leftarrowfill}}
\newcommand{\Arr}{\hbox to 80pt{\rightarrowfill}}

\def\theoremname{Theorem}
\def\propositionname{Proposition}
\def\corollaryname{Corollary}
\def\lemmaname{Lemma}
\def\remarkname{Remark}
\def\conjecturename{Conjecture} 

\def\definitionname{Definition}
\def\exercisename{Exercise}
\def\examplename{Example}
\def\examplesname{Examples}
\def\problemname{Problem}
\def\problemsname{Problems}

\def\satzname{Satz} 
\def\koroname{Korollar}
\def\folgname{Folgerung}
\def\bemerkname{Bemerkung}
\def\aufgname{Aufgabe}

\def\beisname{Beispiel}
\def\beissname{Beispiele}
\def\bewname{Beweis}

\def\@thmcounter#1{\noexpand\arabic{#1}}
\def\@thmcountersep{}
\def\@begintheorem#1#2{\it \trivlist \item[\hskip 
\labelsep{\bf #1\ #2.\quad}]}
\def\@opargbegintheorem#1#2#3{\it \trivlist
      \item[\hskip \labelsep{\bf #1\ #2.\quad{\rm #3}}]}
\makeatother
\newtheorem{theor}{\theoremname}[section]
\newtheorem{propo}[theor]{\propositionname}
\newtheorem{coro}[theor]{\corollaryname}
\newtheorem{lemm}[theor]{\lemmaname}

\newenvironment{thm}{\begin{theor}\it}{\end{theor}}
\newenvironment{theorem}{\begin{theor}\it}{\end{theor}}

\newenvironment{prop}{\begin{propo}\it}{\end{propo}}

\newenvironment{cor}{\begin{coro}\it}{\end{coro}}

\newenvironment{lem}{\begin{lemm}\it}{\end{lemm}}
\newenvironment{lemma}{\begin{lemm}\it}{\end{lemm}}

\newtheorem{rema}[theor]{\remarkname}

\newenvironment{rem}{\begin{rema}\rm}{\end{rema}}

\newtheorem{stepnow}[theor]{}

\newtheorem{defin}[theor]{\definitionname} 

\newenvironment{definition}{\begin{defin}\rm}{\end{defin}}
\newenvironment{defn}{\begin{defin}\rm}{\end{defin}}

\newtheorem{exerc}{\exercisename}[section]

\newtheorem{exa}[theor]{\examplename}

\newenvironment{example}{\begin{exa}\rm}{\end{exa}}
\newenvironment{ex}{\begin{exa}\rm}{\end{exa}}

\newtheorem{exas}[theor]{\examplesname}

\newenvironment{exs}{\begin{exas}\rm}{\end{exas}}

\newtheorem{conj}[theor]{\conjecturename}

\newtheorem{pro}[theor]{\problemname}

\newenvironment{prob}{\begin{pro}\rm}{\end{pro}}

\newtheorem{prs}[theor]{\problemsname}

\newtheorem{aufg}{\aufgname}[section]

\newenvironment{prf}{\begin{proof}}{\end{proof}}
 
\newcommand{\pmat}[1]{\begin{pmatrix} #1 \end{pmatrix}}

{\hfill\qed\end{trivlist}}

\newenvironment{beweis*}{\begin{trivlist}\item[\hskip%
\labelsep{\bf\bewname.\quad}]}%
{\end{trivlist}}

\newtheorem{satzn}[theor]{\satzname}

\newtheorem{koro}[theor]{\koroname}

\newtheorem{folg}[theor]{\folgname}

\newtheorem{bem}[theor]{\bemerkname}

\newtheorem{aufgn}[theor]{\aufgname}

\newtheorem{beis}[theor]{\beisname}

\newtheorem{beiss}[theor]{\beissname}

\numberwithin{equation}{section}

\usepackage{color}
\usepackage{color}

\renewcommand{\rk}{{\mathop{\rm rk}}}

\newcommand{\sE}{{\sf E}}

\renewcommand\up{{\uparrow}}

\newcommand{\sH}{{\sf H}}

\newcommand{\sV}{{\tt V}}

\newcommand{\be}{{\bf{e}}}

\renewcommand{\bO}{\mathbb O}

\renewcommand{\phi}{\varphi}

\newcommand{\otau}{\overline{\tau}}

\renewcommand\mlabel{\label} 

\begin{document}
\title{Wedge domains in non-compactly causal symmetric spaces }
\author{Karl-Hermann Neeb, Gestur \'Olafsson
\thanks{The research of K.-H. Neeb was partially supported
  by DFG-grant NE 413/10-1. The research of G. \'Olafsson was partially supported by Simons grant 586106.}} 

\maketitle

\abstract{This article is part of an ongoing project
aiming at the connections between causal structures
on homogeneous spaces, Algebraic Quantum Field Theory (AQFT),
modular theory of operator algebras 
and unitary representations of Lie groups.
In this article we concentrate on non-compactly causal symmetric spaces $G/H$. This
class contains de Sitter space but also other spaces with invariant partial ordering. 
 The central ingredient 
is an Euler element $h$ in the Lie
algebra of $\fg$. We define three different kinds of wedge domains 
depending on $h$ and
the causal structure on $G/H$.  Our main result is that
the connected component  containing the base point $eH$ of these seemingly different 
domains all agree. Furthermore we discuss the connectedness of
those wedge domains. We show that each of these spaces has a natural extension to a non-compactly
causal symmetric space of the form $G_\C/G^c$ where
$G^c$ is certain real form of the complexification $G_\C$ of $G$.
As $G_\C/G^c$ is non-compactly causal, it also contains  three types of
wedge domains. Our results says that the intersection of these domains with $G/H$
agree with the wedge domains in $G/H$.}
\bigskip

\tableofcontents

\section{Introduction}
\mlabel{sec:1}

This article is part of an ongoing project
aiming at the connections between causal structures
on homogeneous spaces, Algebraic Quantum Field Theory (AQFT),
modular theory of operator algebras 
and unitary representations of Lie groups.

The causal homogeneous spaces we consider in this paper
are symmetric spaces $M = G/H$, where
$G$ is a connected reductive Lie group, endowed with an involutive automorphism
$\tau$, and $H$ is an open subgroup of the group $G^\tau$ of
$\tau$-fixed points. The causal structure on $M$ is represented
by a $G$-invariant field of pointed open convex
cones $V_+(m) \subeq T_m(M), m \in M$, so that
we have on the infinitesimal level a triple
$(\g,\tau,C)$, where $C = \oline{V_+(eH)} \subeq T_{eH}(M) \cong \g^{-\tau}$ is a closed
pointed generating $\Ad(H)$-invariant convex cone
(see \cite{HO97} for details).

There are three major types of causal symmetric spaces: 
\begin{itemize}
\item The flat spaces, where $M$ is a vector space and $G$ acts by affine maps; discussed
in \cite{NOO21}. 
\item The compactly causal  (cc)  spaces: Here the
elements in the interior $C^\circ$ are elliptic (have imaginary spectrum). Those
spaces were treated in \cite{NO21}.
\item The 
non-compactly causal (ncc) spaces: The elements $x \in C^\circ$ are 
hyperbolic ($\ad x$ is semisimple with real spectrum).
\end{itemize}
In this article we 
discuss the geometry of the ncc case.

Before discussing the geometry and the content of the article in more details, let
us briefly discuss the motivation of this work, the Algebraic
Quantum Field Theory, in short AQFT,   
in the sense of Haag--Kastler. Here one considers 
{\it nets} of von Neumann algebras $\cO\mapsto \cM(\cO)$
of operators on a fixed Hilbert space $\cH$ (\cite{Ha96}). 
The hermitian elements of the algebra $\cM(\cO)$ represent
observables  that can be measured in the ``laboratory'' $\cO$,
an open subset of the spacetime~$M$. 
Accordingly, one requires {\it isotony}, i.e.,
that $\cO_1 \subeq \cO_2$ implies 
$\cM(\cO_1) \subeq \cM(\cO_2)$. 
Causality enters by the {\it locality} assumption that 
$\cM(\cO_1)$ and $\cM(\cO_2)$ commute if 
$\cO_1$ and $\cO_2$ are space-like separated, i.e.,
cannot be connected by  causal curves. One further assumes the existence of a unitary representation 
$U \: G \to \U(\cH)$  of a Lie group~$G$, acting as a space-time symmetry 
group on $M$, such that
\[U(g) \cM(\cO) U(g)^* = \cM(g\cO)\quad\text{for}\quad
g \in G.\]
In addition, one assumes a $U(G)$-fixed 
unit vector $\Omega \in \cH$, representing typically 
a vacuum state of a quantum field. 

The domains $\cO \subeq M$ for which 
$\Omega$ is cyclic and separating for $\cM(\cO)$ are of particular 
relevance. For these domains $\cO$, the
Tomita--Takesaki Theorem (\cite[Thm.~2.5.14]{BR87})
yields for the von Neumann algebra 
$\cM(\cO)$ a conjugation (antiunitary involution) $J_\cO$ and a
positive selfadjoint operator $\Delta_\cO$ satisfying
\begin{equation} \mlabel{eq:j1}
J_\cO \cM(\cO) J_\cO = \cM(\cO)' 
\quad \mbox{ and } \quad
\Delta_\cO^{it} \cM(\cO)\Delta_\cO^{-it} = \cM(\cO) \quad \mbox{
  for } \quad t \in \R. 
\end{equation} 
We thus obtain the modular automorphism group of $\cM(\cO)$ defined by 
$\alpha_t(A) = \Delta_\cO^{-it/2\pi} A \Delta_\cO^{it/2\pi}$,
$A \in \cM(\cO)$. 
It is now an interesting question when this
modular group is ``geometric'' 
in the sense that it is implemented by a one-parameter subgroup of~$G$, hence 
corresponds to a one-parameter group of symmetries of~$M$. For the modular 
conjugation $J_\cO$, we may likewise ask for the existence of an involutive 
automorphism $\tau_G$ of $G$ and an involution $\tau_M$ on $M$ 
reversing the causal structure, such that  
\begin{equation} 
  \mlabel{eq:j2} J_\cO \cM(\cO) J_\cO = \cM(\tau_M(\cO)), \qquad 
J_\cO U(g) J_\cO = U(\tau_G(g))\quad 
\mbox{ for } \quad g \in G, \cO \subeq M,
\end{equation}
and that $\tau_M$ and $\tau_G$ are compatible in the sense that 
\[ \tau_M \circ g = \tau_G(g) \circ \tau_M  \quad \mbox{ for } \quad g \in G. \]

It is often natural to simplify the structures
 by considering instead of the pair $(\cM,\Omega)$ the
  corresponding real subspace $\sV := \sV_{(\cM,\Omega)}
  := \oline{\cM_h\Omega}$, where $\cM_h =\{  M \in \cM \: M^* =M\}$.
  This subspace is called
  \begin{itemize}
  \item {\it cyclic} if $\sV + i \sV$ is dense in $\cH$,
    which means that $\Omega$ is cyclic for $\cM$.
  \item {\it separating} if $\sV \cap i \sV = \{0\}$, 
    which means that $\Omega$ is separating for $\cM$.
  \item {\it standard} if it is cyclic and separating, i.e.,
if $\sV \cap i \sV = \{0\}$ and $\oline{\sV + i \sV} = \cH$.
\end{itemize}
These three properties of real subspaces make sense without
any reference to operator algebras, but they still reflect
an important part of the underlying structures that can be studied
in the much simpler context of real subspaces.  In particular,
every standard subspace~$\sV$ leads to a densely defined closed conjugation
$S_{\sV}(u+iv)= u -iv $, $u, v \in \sV$, and then by polarization to an anti-unitary
involution $J_{\sV}$ and positive densely defined operator $\Delta_\sV$ such that
$J_\sV\Delta_\sV J_\sV=\Delta^{-1}$ and $S_\sV = J_\sV\Delta^{1/2}$.

Here is where the representations enter the picture.
Start with an antiunitary representation
$U \:  G \rtimes \{\id_G,\tau\} \to \AU(\cH)$, i.e.,
  $U(G) \subeq \U(\cH)$ and $J := U(\tau)$ is a conjugation.
Recall that the subspace 
$\cH^\infty \subeq \cH$ of vectors $v \in \cH$ for which the orbit map 
$U^v \: G \to \cH, g \mapsto U(g)v$, is smooth 
(the {\it smooth vectors}) is dense
and carries a natural Fr\'echet topology for which the action of 
$G$ on this space is smooth (\cite{Ne10}). 
The space $\cH^{-\infty}$ of continuous antilinear functionals $\eta \: \cH^\infty \to \C$ 
(the {\it distribution vectors}) 
contains in particular Dirac's kets 
$\la \cdot, v \ra$, $v \in \cH$, so that 
we obtain a {\it rigged Hilbert space} 
\[ \cH^\infty \into \cH \into \cH^{-\infty},\] 
where $G$ acts on all three spaces 
by representations denoted $U^\infty, U$ and $U^{-\infty}$, respectively.
To any real subspace $\sE \subeq \cH^{-\infty}$ and every open subset 
$\cO \subeq G$, we associate the closed real subspace 
\begin{equation}
  \label{eq:HE}
  \sH_\sE(\cO) := \oline{\spann_\R U^{-\infty}(C^\infty_c(\cO,\R))\sE},
\quad \mbox{ where } \quad 
 U^{-\infty}(\phi) = \int_G \phi(g)U^{-\infty}(g)\, dg, \ \ 
\phi \in C^\infty_c(G) 
\end{equation}
denotes the integrated representation of the convolution algebra 
$C^\infty_c(G)$ of test functions on $G$ on the space~$\cH^{-\infty}$.
On a homogeneous space $M = G/H$ with the projection map
$q \: G \to M$ we now obtain a ``push-forward net''
\begin{equation}
  \label{eq:pushforward}
  \sH^M_\sE(\cO) := \sH_\sE(q^{-1}(\cO)).
\end{equation}
This assignment is $G$-covariant and monotone.
One can now use the functorial process provided by Second Quantization 
(\cite{Si74}) to associate to any real  subspace $\sH_E(\cO) \subeq \cH$   
a  von Neumann algebra $\cM(\cO) := \cR_\pm(\sH_E(\cO))$ 
on the bosonic/fermionic Fock space $\cF_\pm(\cH)$.
This method has been developed by Araki and Woods  in the context of 
free bosonic quantum fields (\cite{AW63}); 
some of the corresponding fermionic results are more recent 
(cf.\ \cite{BJL02}). Other statistics (anyons) 
are developed in \cite{Schr97} and more recent deformations of 
this procedure are discussed in \cite[\S 3]{Le15}.
Although this construction ignores field theoretic interactions,
it displays
already some crucial features of quantum
field theories. This was explored for the flat case in \cite{NOO21}.

In \cite{NO21}, this construction has been
carried out for the class of {\it compactly causal symmetric spaces}.
The key point is to find suitable subspaces $\sE$ of distribution vectors such
that the net $\sH_E^M$ has nice properties,
such as the Bisognano--Wichmann property. 
 It asserts that there exists an element $h \in \g$ satisfying 
  $\tau(h) = h$ and an open subset $W \subeq M$ (called a wedge region),
  such that $\sV := \sH^M_\sE(W)$ is a standard subspace
  with modular conjugation $J_\sV = U(\tau)$ and
  modular operator $\Delta_\sV = e^{2\pi i \cdot \partial U(h)}$.
  
In the physics literature no uniform definition of a wedge region
in a spacetime exists, but there are several approaches which
share many aspects with our wedge domains in ncc symmetric spaces.
In \cite{DLM11}, wedge domains in $1+3$-dimensional spacetimes
are specified as connected components of the spacelike
complement of $2$-dimensional spacelike subspaces. In our context
this corresponds to starting with
the fixed point space $M^\alpha$ of the modular flow
$\alpha_t(m) = \exp(th).m$ and specifying
the wedge domain $W_M(h)$ in terms of the polar decomposition.
On de Sitter space, \cite{BB99} specifies wedge domains as the
causal completion of the world line
of a uniformly accelerated observer, and
we know from \cite{MNO22b}
that this picture fits for general ncc symmetric spaces.
This approach even works for anti-de Sitter space
(\cite{BS04}; see also \cite{LR08})
which is studied from the symmetric space perspective in \cite{NO21}. 
Further, \cite{Bo09} defines wedge regions as the union of a family
of order intervals between elements on two fixed lightrays,
and it is not so clear how this approach should be transcribed to symmetric
spaces, where the sets $\Exp_{eH}(C_\pm)$ are analogs of lightrays,
but in general of higher dimension. 

Causal symmetric spaces provide natural geometric objects for 
this theory and the most natural generalization of homogenous
time-oriented Lorentzian manifolds. So
a first step in our program is to understand how nets of real subspaces
can be constructed on causal symmetric spaces
$M = G/H$ corresponding to the data $(\g,\tau,C)$.

To connect with the modular theory of operator algebras,
we add to $(\g,\tau,C)$ additional structure represented
by an {\it Euler element} $h \in \fh$, i.e.,
$\ad h$ defines a $3$-grading
\[ \g = \g_1(h) \oplus \g_0(h) \oplus \g_{-1}(h),\quad \fg_j (h) = \ker (\ad (h)-  j\1), \, j=1,0,-1.\]
Our structural data is therefore given by the quadruple $(\g,\tau,C,h)$,
called a {\it modular   causal symmetric Lie algebra}.

With a view towards the connection with AQFT,
the set of pairs $(h,\tau) \in \g \times \Aut(\g)$,
consisting of an Euler element
$h \in \g$ and an involution $\tau$ with $\tau(h) = h$ have
been studied abstractly in \cite{MN21}, and we shall use some of
the obtained results along the way. Here we shall focus on the
relevant global geometric aspects of $M$ and in particular on
  the candidates for wedge regions.

On the global level $h$ generates the 
{\it modular flow} 
\[\alpha_t(gH) = \exp(th) g H,\quad t\in \R.\]
As we are particularly interested in domains
$\cO \subeq M$ invariant under the modular flow $\alpha$
for which there may exist nets of von Neumann algebras
such that $U(\exp t h) = \Delta_{\cM(\cO)}^{-it/2\pi}$ is the
modular group of $\cM(\cO)$, we are confronted with the 
problem to formulate sufficient conditions on~$\cO$.
As the specific examples in AQFT suggest, the modular flow on $\cO$
should be timelike future-oriented
(\cite{TW97, BB99, BMS01, BS04, Bo09, LR08}, \cite[\S 3]{CLRR22}),
which in our context means that
the {\it modular vector field}
\begin{equation}
  \label{eq:xhdef}
 X_h^M(m) 
:= \frac{d}{dt}\Big|_{t = 0} \alpha_t(m) 
\end{equation}
should satisfy
\[ X^M_h(m) \in V_+(m) \quad \mbox{ for all } \quad m \in \cO.\]
To understand how to find regions with this property, we therefore
have to study the  {\it positivity region}
\[ W_M^+(h) := \{ m \in M \: X^M_h(m) \in V_+(m)  \} \]
of the modular flow. Its connected components are called {\it wedge domains}.
In \cite{NO21} this has been carried out for compactly causal
symmetric spaces (cf.\ Definition~\ref{def:ssp}),
which turned out to be rather easy because
the wedge domains are orbits of certain subsemigroups of $G$ with a
rather simple structure. In \cite{NO21} we even constructed
second quantization nets of operator algebras for
unitary highest weight representations of $G$.

In this paper we turn to the geometric aspects of wedge domains in
non-compactly causal symmetric spaces (cf.\ Definition~\ref{def:ssp})
and we leave the construction
of nets on these spaces for the future. This is substantially harder than
in the compactly causal space because some of the groups $G$ we are dealing
with have no unitary representations in which any Lie algebra element
has semibounded spectrum. One therefore has to understand first from which
class of representations such nets may be constructed. We are actually
optimistic and think that this is always possible.
For the case where $G = \SO_{1,2}(\R)_e$ and $M = \dS^2$ is $2$-dimensional
de Sitter space. We refer to \cite{BM96} for results in this direction.\\

First we introduce several types of ``wedge domains''  
in reductive non-compactly causal (ncc) symmetric spaces $M = G/H$, 
specified by a modular causal symmetric Lie algebra 
$(\g,\tau,C,h)$,  a connected Lie group $G$
and an open subgroup $H \subeq G^\tau$ with $\Ad(H)C = C$. 
To strip off artificial difficulties arising from coverings,
we  make {\bf two structural assumptions} (GP) and (Eff) 
specified in Subsection~\ref{subsec:globass} which imply that
\[ G \cong \fz(\g) \times \Inn([\g,\g]), \]
where $\Inn$ denotes the group of inner automorphisms and $\fz(\g)$ is
the center of $\g$. Further,
$h$ is an Euler element contained in $\fh := \g^\tau$
which induces a one-parameter group 
$\alpha_t := e^{t \ad h}$ of automorphisms of $\g$ and also
automorphisms of $G$ and $M$.
Then 
\[ \tau_h := e^{\pi i \ad h}  \]
is an involutive automorphism of $\g$
and 
\[ \kappa_h := e^{-\frac{\pi i}{2} \ad h} \in \Aut(\g_\C) \] 
is an  automorphism of order~$4$ of $\g_\C$
with $\kappa_h^2 = \tau_h$.
Let $V_+(gH) := g.C^\circ \subeq T_{gH}(M)$
denote the open cones
\begin{footnote}
{We write $C^\circ$ for the relative interior of the cone $C$ in its span.}  
\end{footnote}
defining the causal structure on $M$ and recall 
the {\it modular vector field} $X_h^M \in \cV(M)$ defined by \eqref{eq:xhdef}.

In addition to $W_M^+(h)$ we introduce the following domains associated to this data: 
\begin{itemize}  
\item The {\it tube domain of $M$} (here we assume for simplicity that 
 $M \subeq M_\C := G_\C/G^c$, see the beginning of Section \ref{sect:Complex} for definitions.)
\[ \cT_M := G.\Exp_{eH}(i C^\pi), \qquad 
C^\pi := \{ x \in C^\circ \: s(x) < \pi \},\] 
where $s \: C^\circ \to (0,\infty)$ is an $\Ad(H_e)$-invariant function 
specified in terms of roots as follows. We fix a  maximal
abelian subspace $\fa$ consisting of hyperbolic elements. Then 
$C^\circ = \Ad(H_e).(C^\circ\cap \fa)$ and
for
$y = \Ad(g) x \in C^\circ$ with $x\in \fa \cap 
C^\circ$, $g \in H$, we have 
\[ s(y) = \max \{ |\alpha(x)|, 2|\beta(x)| \: \alpha \in \Delta_p, 
\beta \in \Delta_k \}.\] 
Here $\Delta_k$ is the set of compact roots and $\Delta_p$ 
the set of non-compact roots (\cite{KN96,HO97} and Definition~\ref{def:3.1}). 
\item The {\it  KMS wedge domain} 
\[ W_M^{\rm KMS}(h) := \{ m \in M \: 
(\forall z \in \cS_\pi = \{ z \in \C \: 0 < \Im z < \pi\})\ \alpha_z(m) \in \cT_M\},\] 
where $(\alpha_z)_{z \in \C}$ denotes the holomorphic extension 
of the modular flow to a holomorphic $\C$-flow on the complex symmetric space~$M_\C$. 
\item The {\it wedge domains of polar type} are the domains 
  \begin{equation}
    \label{eq:wmh}
 W_M (h) 
= (G^{\tau_h})_e.\Exp_{eH}( (C_+ + C_-)^\pi) 
= (G^h)_e.\Exp_{eH}( (C_+ + C_-)^\pi),
  \end{equation}
where 
$C_\pm := \pm C \cap \fq_{\pm 1}(h)$,
$\fq_{\pm 1}(h) = \ker(\ad h \mp \1) \cap \fq$, and 
\[ (C_+ + C_-)^\pi := \{ x = x_+ + x_- \in C_+^\circ+ C_-^\circ) 
\: s(x_+ - x_-) < \pi\}. \]
 \end{itemize}
By definition the domains $\cT_M$  and $W_M(h)$ are connected and it follows from
Theorem \ref{thm:6.1} that  $W_M^{\rm KMS}(h)$ likewise is. But
in general $W_M^+(h)$ is not connected.
One main result is that: 
\begin{equation}
  \label{eq:e1}
 W_M^+(h)_{eH} = W_M^{\rm KMS}(h) = \kappa_h(\cT_M^{\oline\tau_h})
= W_M(h),  
\end{equation}
where $W_M^+(h)_{eH}$ is the connected component of $W_M^+(h)$ containing 
$W_M(h)$ and $\oline\tau_h$ is the antilinear extension of $\tau_h$ to
$\g_\C$, so that $\cT_M^{\oline\tau_h}$ denotes the fixed point set of
an antiholomorphic involution on the complex manifold $\cT_M$. 
The proofs are given in  Theorem~\ref{thm:6.1}, 
Corollary~\ref{cor:4.12}, and Theorem~\ref{thm:7.4}. 
We also show that $\cT_M$ is open in $M_\C$ 
(see Theorem~\ref{thm:4.7} for details). 
 Furthermore we show in Corollary~\ref{cor:polwedgeopen}
that $W_M(h)$ is open in $M$, which also follows from \eqref{eq:e1} because 
$W_M^+(h)$ is obviously open. 


The structure of this paper is as follows:
In Section~\ref{sec:2} we introduce 
modular causal symmetric spaces and
fix the standard assumptions on the symmetric 
Lie algebras and groups considered in this article.
In particular, we assume throughout that
{\bf $\g$ is reductive}.   Our standard
references for causal symmetric spaces are \cite{HN93, KN96, HO97,Ol91}.
The last two references only deal with semisimple symmetric spaces, whereas the
first two treat general ones.\\

In Section~\ref{sec:3} we then turn to the fine structure
of causal symmetric Lie algebras,
which is encoded in the structure of restricted roots.
A key feature is that the set of roots decomposes nicely into the
so-called compact and non-compact roots and the
positive systems in the non-compact roots determine
minimal and maximal $\Ad(H)$-invariant cones
$C_\fq^{\rm min} \subeq C_\fq^{\rm max}$ in $\fq$.
This information goes back to 
\cite{OO91} but was systematized 
in \cite{Ol91} and developed further in \cite{HO97,KN96,HN93}.
As a key technical tool we use
in Section~\ref{subsec:olaf}  maximal $\tau$-invariant
sets of strongly orthogonal roots. 
In Subsection~\ref{sect:Complex} we introduce the important method
of embedding ncc symmetric spaces into spaces of complex type,
which are of the form $G_\C/G^c$, where $G^c$ is the integral subgroup
of $G_\C$ with Lie algebra $\g^c = \fg^\tau \oplus i\fg^{-\tau}$
(cf.~Lemma~\ref{lem:3.1}). \\

In Section~\ref{sec:4} we first  introduce 
a special class of Euler elements $h_c$, the 
{\it causal Euler elements} for $(\g,\tau , C)$,
which are those contained in $C^\circ \cap \fq_\fp$.
The associated {\it causal Riemannian element} is $x_r := \frac{\pi}{2} h_c
\in C^\pi \cap \fa$. It plays an important role in connecting the tube domain 
$\cT_M$ of the causal symmetric space $M$ 
with the crown domain of the Riemannian symmetric space $M^r := G/K$ 
(\cite{AG90, GK02, KS05}).
 Both domains have, up to coverings,
  natural realizations in the complex symmetric space
  $G_\C/K_\C \cong G_\C/H_\C$.
Concretely, we show in Theorem~\ref{thm:4.7} that,
for $m := \Exp_{eH}(i x_r) \in \cT_M$, the orbit
$G.m$ is isomorphic to $M^r$ and that we thus obtain an identification
of the tube domain $\cT_M$ with the crown domain
\[ \cT_{M^r} = G.\Exp_m(i\Omega_\fp), 
\quad \mbox{ where } \quad 
\Omega_\fp 
= \Big\{ x \in \fp \: \rho(\ad x)  < \frac{\pi}{2}\Big\}.\]
This result, proved in Section~\ref{sec:5}, is prepared in several steps.
In Subsection~\ref{subsec:4.1a}, we first show that
the polar decomposition of the crown domain
$\cT_{M^r}$ of the Riemannian symmetric space $M^r$
is a diffeomorphism (see \cite{AG90}). 
We also obtain a new
characterization of real crown domains as a submanifold of the real tube domain
$\cT_{C_\fq} = \fh + C^\circ \subeq \g$ of 
the cone $C$ (Theorem~\ref{thm:crownchar-gen}):
For a causal Euler element $h_c \in C \cap \fq_\fp$,
the connected component of $h_c$ in the intersection 
$\cO_{h_c} \cap \cT_{C_\fq}$ is the domain 
\[ \cT_{M_H} = \Ad(H) e^{\ad \Omega_{\fq_\fk}} h_c, \quad \mbox{ where } \quad 
\Omega_{\fq_\fk} = \Big\{ x \in \fq_\fk \: \rho(\ad x) < \frac{\pi}{2}\Big\}.\] 

Our main result is then proved in Section~\ref{sec:6}: 
Theorem~\ref{thm:6.1} asserts that 
\[W_M(h) = W_M^{\rm KMS}(h) = \kappa_h^{-1}((\cT_M)^{\oline\tau_h}) . \]

In Section~\ref{sec:7} we eventually connect this identity
to the positivity domain by showing that
its identity component $W^+_M(h)_{eH}$ (whose boundary contains the
base point $eH$) coincides with these three domains.
We do not address other connected components of $W_M^+(h)$ in
  this article. However,  we show in Proposition \ref{prop:7.8}
  that it is connected 
for Cayley type spaces of the form
$G/G^\tau$, where $G = \Inn(\g)$ and $\tau = \tau_h$.
For
$G = \Inn(\g)$ (the centerfree case) and non-compactly causal spaces,
it is shown in \cite[\S 7]{MNO22b} that $W^+_M(h)$ is always connected.
This contains the present result on Cayley type spaces,
  but as the argument is very short and direct in this case
  and provides a nice description of the wedge domains, we include
  this case in the present paper.

We conclude this paper with a discussion of several open problems
in Section~\ref{sec:8}. One important problem concerns
the connectedness of the positivity domains introduced here.
The domains $\cT_M$  and $W_M(h)$ are connected by
definition, but it is not at all clear if $W_M^+(h)$ and $W_M^{\rm KMS}(h)$ 
are connected or not. Theorem \ref{thm:6.1} shows that, under the conditions assumed
in this article, $W_M^{\rm KMS}(h)$ is connected. On the other hand,
if $\g$ is simple, for $W_M^+(h)$ to be connected, it is necessary that
$Z(G) = \{e\}$, i.e., $G \cong \Inn(\g)$ and that
$H = \{ g \in G^\tau \: \Ad(g)C = C \}$ is the maximal subgroup of $G^\tau$
preserving the cone~$C$ (Subsection~\ref{subsec:8.1}),
but it is not clear if this is sufficient in all cases. \\

We also add several appendices.
The first presents a table with the classification of irreducible
modular non-compactly causal symmetric spaces in a way most suitable
for our approach. Appendix~\ref{subsec:sl2b}
contains some calculations in $\fsl_2(\R)$.
Appendix~\ref{app:polmaps} contains some finer results on
polar maps in symmetric spaces and their singularities
and  Appendix~\ref{app:1} discusses
quadrics as symmetric spaces
and in particular de Sitter space $\dS^d$ in some detail.
For $d = 2$, this example is used in some proofs to 
deal with the case $\g = \fsl_2(\R) \cong \so_{1,2}(\R)$.

\vspace{5mm}

\nin {\bf Notation:} 
In this article we will use the following notation:
\begin{itemize}
\item  A pair $(\fg,\tau)$ of a Lie algebra $\g$ and an 
involutive automorphism $\tau$ is called a {\it symmetric Lie algebra}. 
For a symmetric Lie algebra we set
\[\fh =\fg^\tau =\{x\in\fg \: \tau (x) = x\}\quad\text{and}\quad
\fq = \g^{-\tau}= \{x\in \fg \: \tau(x) =-x\}.\]
\item  If $(\fg, \tau)$ is a symmetric Lie algebra with $\fg = \fh\oplus \fq$, 
 then the $c$-dual is $(\fg^c ,\tau^c)$ with $\fg^c= \fh+ i\fq$ and
$\tau^c(x + i y) = x- iy$ for $x \in \fg, y \in \fq$; 
see  Section \ref{sect:Complex} for more details. 
\item The {\it complex case} refers to the situation where 
$\g \cong \fh_\C$ and $\tau$ is complex conjugation. The
group case is the $c$-dual case $\fg = \fh \times \fh$ with
the flip involution $\tau (x,y) = (y,x)$. 
\item An element $x$ of a Lie algebra $\g$ is called {\it hyperbolic} if 
$\ad x$ is diagonalizable (over $\R$) and {\it elliptic} if 
$\ad x$ is semisimple with purely imaginary spectrum, i.e., 
if $\oline{e^{\R \ad x}}$ is a compact subgroup of~$\Aut(\g)$.
\item For a linear map $A \in \End(V)$, $\dim V < \infty$, we write 
\[ \rho_i(A) := \sup \{ |\Im \lambda| \: \lambda \in \Spec(A) \} 
\quad \mbox{ and }\quad 
\rho(A) := \sup \{ |\lambda| \: \lambda \in \Spec(A) \} \] 
for the {\it imaginary spectral radius} and the {\it spectral radius} 
 of $A$. 
\item For $x \in \g$, we consider the automorphisms 
$\sigma_x := e^{-2\ad x}$ and $\zeta_x := e^{-\ad x}$ of~$\g$. 
We write $\sigma_x^G$ and $\zeta_x^G$ for the corresponding automorphisms 
of~$G$ defined by conjugation with $\exp(-2x)$ and $\exp(-x)$, respectively.
\item For an Euler element $h \in \g$, i.e., $\ad h$ is diagonalizable 
with eigenvalues $-1,0,1$, we 
write 
$\tau_h :=  e^{-\pi i \ad h} = e^{\pi i \ad h}$ 
and $\kappa_h :=  e^{-\frac{\pi i}{2} \ad h}$. We write 
$\cE(\g)$ for the set of Euler elements in~$\g$.
\item We typically denote automorphism of $G$, its Lie algebra, 
and the induced automorphisms of homogeneous spaces of $G$ by the 
same letter. 
\item If $\fq$ is a subspace of $\g$ then $\z (\fq)=\{x\in \fq\: [x,\fq]=0\}$. Note that
$\z (\fq)$ is contained in the center of the Lie algebra generated by $\fq$.
\item If $x\in \fg$ is diagonalizable over $\R$   and $\lambda\in\R$, then
$\fg_\lambda (x)=\{y\in \g\: [h,x]= \lambda x\}$. 
For a symmetric Lie group $(G,\tau)$, we consider the action of 
$G$ on itself by 
\begin{equation}
  \label{eq:tauact}
  g.h := gh\tau(g)^{-1} = gh g^\sharp.
\end{equation} 
To an open subgroup $H \subeq G^\tau$, we associate the symmetric 
space $M = G/H$. Note that $H$ contains the identity component 
$G^\tau_e := (G^\tau)_e$. We sometimes also call the triple 
$(G,\tau,H)$ a {\it symmetric Lie group} because this triple specifies 
the symmetric space~$M$. Then the map 
$Q \: M = G/H \to G, Q(gH) := g \tau(g)^{-1} = g g^\sharp$ 
is called the {\it quadratic representation of $M$}. 
It is a $G$-equivariant covering 
of the submanifold $G^{-\tau}_e$, 
the identity component of 
\[ G^{-\tau} := \{ g \in G \: \tau(g) = g^{-1}\} \] 
(see the appendix in \cite{NO21}). 
\item  The tangent bundle $T(M)$ 
of $M$ is identified with $G\times_H\fq$ 
and the canonical 
action of $G$ on $TM$ is denoted by 
$G \times T(M) \to T(M), (g,v) \mapsto g.v $.
\item For a convex cone $C$ in a finite-dimensional vector space,
  we write $C^\circ := \Int_{C-C}(C)$ for the relative interior of
  $C$ in its span. 
\end{itemize}

\section{Causal symmetric spaces}
\mlabel{sec:2}

In this section  we introduce 
modular causal symmetric spaces and
fix the standard assumptions that we will make on the symmetric 
Lie algebras and groups considered in this article.
One of the standard assumptions is that the Lie algebra $\g$ is 
{\bf reductive}.   Our standard
references for causal symmetric spaces are \cite{HN93, KN96, HO97,Ol91}. The last
two references only deal with semisimple symmetric spaces, whereas the 
first two treat general spaces.

\subsection{Modular causal symmetric Lie algebras} 
\mlabel{subsec:2.1}

Let $\g$ be a Lie algebra and $\tau : \g\to \g$ an involution, i.e.,  
$(\g,\tau)$ is a {\it symmetric Lie algebra}. We write
\[ \fh := \g^\tau =\{x\in \g\: \tau (x)=x\}\quad 
\mbox{ and } \quad \fq := \g^{-\tau } 
= \{x\in\g\: \tau (x) = -x\}.\] 
The {\it Cartan dual} symmetric Lie algebra is $(\g^c, \tau^c)$ with 
\[ \g^c = \fh + i \fq \quad \mbox{ and } \quad 
\tau^c(x+ iy) = x - iy \quad \mbox{ for } \quad x \in \fh, y \in \fq.\] 
We have $\g^c =\fg_\C^{\overline{\tau}}$, 
where $\otau$ is the conjugate-linear extension of $\tau$ to $\g_\C$; 
in particular $\g^c$ also is a real form of $\g_\C$.

\begin{defn}\label{def:ssp}
A {\it causal symmetric Lie algebra} is a 
triple $(\g,\tau ,C)$, where  $(\g,\tau)$ is a symmetric Lie algebra and 
$C\subset \fq$ is a pointed generating closed convex cone
invariant under the group $\Inn_\g(\fh)$.  
We say that the causal symmetric Lie algebra is
\begin{itemize}
\item[\rm (cc)] {\it compactly causal}, or simply cc,  if the cone $C$ is {\it 
elliptic}, i.e., its interior consists of elliptic elements.  
\item[\rm (ncc)] {\it non-compactly causal}, or simply ncc,  if the cone $C$ is {\it 
hyperbolic}, i.e., its interior consists of hyperbolic elements.  
\item[\rm (ct)]  of Cayley type if
there exists an Euler element $h$ such that $\tau= \tau_h$. 
\end{itemize}
\end{defn}

\begin{rem}
Since $\ad x = \ad (x+z)$ for $x\in \g$ and $z\in \fz(\g)$, 
a cone $C$ in a reductive Lie algebra 
is hyperbolic or elliptic if and only 
if its projection $C^\prime = p_{[\g,\g]}(C)$ to  the 
semisimple commutator algebra 
$[\fg, \fg]$ along the center $\fz(\g)$ has this property. 
\end{rem}

\begin{exs} \mlabel{exs:3dim}
(Reducible Lorentzian causal symmetric Lie algebras  and 
non-adapted cones) \\
(a) If $(\fs,\theta)$ is a semisimple Lie algebra and 
$\theta$ a Cartan involution, then 
\[ \g := \R \oplus \fs \quad \mbox{ with } \quad 
\tau := (-\id_\R) \oplus \theta \]  
is a non-compactly causal symmetric Lie algebra with the cone 
\[ C := \{ (t,x) \in \fq = \R \oplus \fs^{-\theta} \: 0 \leq t, \kappa(x,x) 
\leq t^2 \},\] 
where $\kappa$ is the Cartan--Killing form of $\fs$ (which is positive definite 
on $\fs^{-\theta}$). 

\nin (b) The dual situation arises for a compact symmetric Lie algebra $(\fs,\sigma)$, 
$\g = \R\oplus \fs$ and $\tau= (-\id_\R) \oplus \sigma$. 
Then  we obtain a compactly causal symmetric Lie algebra by 
\[ C := \{ (t,x) \in \fq = \R \oplus \fs^{-\sigma} \: 0 \leq t, -\kappa(x,x) 
\leq t^2 \},\] 
where $\kappa$ is the negative definite Cartan--Killing form of $\fs$. 
\end{exs}

Motivated by the preceding examples, we shall use the following concept  
of a Cartan involution on reductive Lie algebras: 

\begin{defn} \mlabel{def:2.3} (Cartan involutions) 
We call an involutive automorphism $\theta$ of a reductive 
Lie algebra $\g$ a {\it Cartan involution} if 
$\fz(\g) \subeq \g^{-\theta}$ and the restriction to the commutator 
algebra is a Cartan involution. 
For $\fk = \g^\theta$ and $\fp = \g^{-\theta}$, this is equivalent to the requirement 
that all elements in $\fk$ are elliptic, all elements in $\fp$ are hyperbolic 
and $\fk \subeq [\g,\g]$.

Note that $\theta$ is a Cartan involution if 
$(\g,\theta)$ is an effective Riemannian symmetric Lie algebra of non-compact type. 
\end{defn}

\begin{rem}  If $(\fg,\tau_h,C)$ is of Cayley type,  
then $\fz(\g) \subset \fh = \ker(\ad h)$ 
and hence $\tau_h|_{\fz(\g)} = \id_{\fz(\g)}$.  
For any open subgroup $H \subeq G^\tau$, 
the central subgroup $Z(G)\cap H$ acts trivially
on $G/H$ so that one can reduce many question to the 
semisimple case. However, the non-semisimple Examples~\ref{exs:3dim}, \ref{ex:gln} and \ref{ex:gl2} show that 
the cones are typically not adapted to the splitting into 
center and commutator algebra. 
Later we shall assume that $\fz(\g) \subeq \fq$, 
which for Cayley type Lie algebras implies semisimplicity.

If $(\g,\tau)$ is irreducible, then 
 \cite[Thm.~5.6]{Ol91} implies that $(\fg,\tau, C)$ is
of Cayley type if and only if the space 
$\fq$ is not irreducible for the adjoint action of $\fh$. 
The $h$-eigenspace decomposition 
$\fq= \fq_{+1}(h)\oplus \fq_{-1}(h)$
further has the property that  
$\fq_{+1}(h)$ and $\fq_{-1}(h)$ are
isomorphic as vector spaces,  but not as $\fh$-modules 
because $h$ acts with different eigenvalues 
(cf.\ the proof of \cite[Thm.~5.6]{Ol91}). 
\end{rem}

We now introduce the main structure of
this paper on the infinitesimal level. 
To $(\g,\tau,C)$ we add an Euler element $h$ satisfying 
certain compatibility conditions.
\begin{defn} A {\it modular causal symmetric Lie algebras} 
is a  quadruple $(\g,\tau, C, h)$, where 
$(\g,\tau,C)$ is a causal symmetric Lie algebra  
and $h \in \g^\tau$ is an {\it Euler element} 
satisfying $\tau_h(C) = - C$. It is called
{\it compactly causal}, respectively
{\it non-compactly causal} if $(\fg,\tau ,C)$
is compactly causal, respectively, non-compactly causal. 
A  modular causal symmetric
Lie algebras is said to be of {\it Cayley type} if $\tau=\tau_h$. 
\end{defn}

\begin{example} (Cayley type spaces) 
Let $G= \SL_2(\R)$, $\fg = \fsl_2(\R)$, and  $\tau=\tau_h$ with  
 $h=\diag(1/2,-1/2)$, so that 
$(\fg , \tau)$ is of  Cayley type. 
Thus  
\begin{equation}\label{eq:tauSL}
\tau\begin{pmatrix} a & b \\ c & d\end{pmatrix}
= \begin{pmatrix} a & -b \\ -c & d\end{pmatrix} 
\quad \mbox{ and } \quad 
\fh = \R h\quad\text{and}\quad \fq = \left\{\begin{pmatrix} 0 & x \\ y & 0 \end{pmatrix}\: x,y\in \R\right\}.
\end{equation}
The cone 
\[C_h =\left\{ \begin{pmatrix} 0 & x\\ y & 0 \end{pmatrix}\: x,y \geq 0\right\} \]
is hyperbolic and $(\fg, \tau ,C_h)$ is non-compactly causal. 
In fact, every element in $C_h^\circ$ 
is $\Inn_\g(\fh)$-conjugate to a multiple of $h$ 
which is hyperbolic.
On the other hand, every element in the interior of 
\[C_e =\left\{ \begin{pmatrix} 0 & x\\ - y & 0 \end{pmatrix}\: 
x,y \geq 0\right\} \] 
is $\Inn_\g(\fh)$-conjugate to an element of $\so_2(\R)$, hence elliptic.  
Therefore $(\fg, \tau ,C_e)$ is compactly causal. 
 Both $(\fg, \tau,C_h,h)$ and $(\fg, \tau,C_e,h)$ 
are modular causal symmetric Lie algebras.
\end{example}
 
\begin{example} In the definition of a modular ncc 
symmetric Lie algebra $(\g,\tau,C,h)$, we
require that $-\tau_h(C) = C$.  We now show that there are examples
  where $-\tau_h (C) \not= C$. 
Note that 
$C' := -\tau_h(C) \subeq \fq$ is a pointed generating closed cone in 
$\fq$ invariant under $\Inn_\g(\fh)$ (note that $\tau_h$ preserves $\fh$).   
If $\theta(h) = -h$, which we may assume after conjugating with 
$\Inn_\g(\fh)$, then $\tau = \theta \tau_{h_c}$
for a causal Euler element $h_c \in C^\circ$
(\cite[Thm.~4.5]{MNO22a}) implies $\tau_{h_c}h = -h$, 
and thus \cite[Thm.~3.13]{MN21} yields 
that $\tau_h(h_c) = - h_c$. This shows that 
\[ h_c\in C \cap C'.\] 
As $h_c \in C^\circ$, it follows that $C \cap C'$ is pointed generating 
and $-\tau_h$-invariant. 
It follows in particular that 
$-\tau_h(C_\fq^{\rm max}) = C_\fq^{\rm max}$ 
and $-\tau_h(C_\fq^{\rm min}) = C_\fq^{\rm min}$.

To see an example, where $-\tau_h(C) \not= C$, we consider the following
Lie algebra of complex type: 
\[ \g = \fsl_{2r}(\C) \supeq \fh = \su_{r,r}(\C), \quad
  \fq = i \fh, \quad
h = \frac{1}{2} \pmat{ 0 & \1 \\ \1 & 0} \]
and the causal Euler element
\[ h_c = \frac{1}{2} \pmat{\1 & 0 \\ 0 & -\1}.\]
Then
\[ \fq_\fp = i \fh_\fk = \Big\{ \pmat{ a & 0 \\ 0 & d} \: a^* = a, d^* = d,
  \tr(a) + \tr(d) = 0\Big\}
  \quad \mbox{ and } \quad
\tau_h\pmat{ a & b \\ -b^* & d} = \pmat{d & - b^* \\ b & a}.\]
We may choose $\fa \subeq \fq_\fp$ as the subspace of diagonal matrices.
Then $\cW \cong S_{2r} \supeq \cW_k \cong S_r \times S_r$
(cf.\ Remark~\ref{rem:3.4}) act by permutation
of diagonal entries. Further
\[ \Sigma_1 = \{ \eps_j - \eps_k \: j \leq r < k \} \quad \mbox{ and } \quad 
  \Sigma_0 = \{ \eps_j - \eps_k \: j, k \leq r\ \mbox{ or } r  < j,k \}.\]
Here
\[ C_\fa^{\rm min} = \{ \diag(x,y) \in \fa \: x \geq 0, y \leq 0\}
  \subeq C_\fa^{\rm max} = \{ \diag(x,y) \in \fa \: \min(x) \geq \max(y)\}\]
(see \eqref{eq:minmaxconesina} below). 

For $r = 2$, we consider the element
\[ z := \diag(2,2,1,-5) \in C_\fa^{\rm max} \setminus C_\fa^{\rm min}
  \quad \mbox{ with } \quad
  -\tau_h(z) = \diag(-1,5,-2,-2).\]
The cone $C_\fa \subeq \fa$ generated by $C_\fa^{\rm min}$ and
$\cW_k z$ (cf.\ Remark~\ref{rem:3.4})
consists of elements  $x = \diag(x_1, \ldots, x_4)$ with $x_1, x_2 \geq 0$,
hence does not contain $-\tau_h(z)$. As $C_\fa$ extends to an
$\Inn_\g(\fh)$-invariant pointed generating invariant cone $C_\fq \subeq \fq$
(cf.\ Theorem~\ref{thm:3.6} below),
this cone is not invariant under $-\tau_h$.
\end{example}

\subsection{Global assumptions on the groups} 
\mlabel{subsec:globass}

Throughout this paper $G$ denotes a connected reductive Lie group.
We shall use the following assumptions which mainly concern the center of $G$: 
\begin{itemize}
\item[\rm(GP)] $G$ has a polar decomposition, i.e., 
there exists an involutive automorphism $\theta$ of $G$, 
such that the subgroup $K := G^\theta$ is connected, and  
for $\fp := \g^{-\theta}$, the polar map 
\[ K \times \fp \to G, \qquad 
(k,x) \mapsto k \exp x \] 
is a diffeomorphism. 
\item[\rm(Eff)] The action of $G$ on $G/K$ is effective, i.e., 
$K$ contains no non-trivial normal subgroup of~$G$.
\end{itemize}

\begin{rem} \mlabel{rem:red-polar} 
(a) If $\g$ is semisimple, then (GP) holds for any Cartan involution 
$\theta$, and $Z(G) \subeq K$ (\cite[Thm.~14.2.8]{HN12}). 
Therefore (Eff) is equivalent to $Z(G) = \{e\}$, which is equivalent to 
$G \cong \Aut(\g)_e = \Inn(\g)$. 

\nin (b) (GP) is also satisfied if $G$ is simply connected. 
Then we extend a Cartan involution $\theta$ from the commutator algebra 
$[\g,\g]$ to  a Cartan involution $\theta$ of $\g$ 
(Definition~\ref{def:2.3}).
As the simple connectedness implies 
$G \cong [G,G] \times Z(G)_e \cong [G,G] \times \fz(\g)$, 
property (GP) follows from the semisimple case. 
\end{rem}

\begin{defn} \mlabel{def:cxplss}
For a reductive symmetric Lie group $(G,\theta)$ 
satisfying (GP), we call 
\[ M^r := G/K \] 
the associated {\it Riemannian symmetric space} (of noncompact type). 
\end{defn}

\begin{ex} \mlabel{ex:gln} (Riemannian symmetric spaces which are also causal) 
Let $\K \in \{\R,\C,\H\}$. 
We consider the connected reductive Lie group 
\[ G := \GL_n(\K)_e/\Gamma, \qquad 
\Gamma := \U_n(\K)_e \cap Z(\GL_n(\K)) = 
\begin{cases}
  \{\pm \1 \} & \text{ for } \K = \R, n \mbox{ even} \\
  \{\1 \} & \text{ for } \K = \R, n \mbox{ odd} \\
  \T \1& \text{ for } \K = \C \\
  \{\pm \1 \} & \text{ for } \K = \H \\
\end{cases}
\] 
Then $\theta(g) := (g^*)^{-1}$ induces an involutive automorphism 
of $G$ with 
\[ K = G^\theta = \U_n(\K)_e/\Gamma\] 
and 
\[ M^r := G/K \cong \GL_n(\K)/\U_n(\K) 
= \{ g^*g \: g \in \GL_n(\K) \} 
\] 
is the open cone 
$\Herm_n(\K)_+$ of positive definite 
hermitian matrices. Here (GP) follows immediately from the polar 
decomposition of $\GL_n(\K)$. As 
$\Gamma$ 
is the largest proper normal subgroup of $\GL_n(\K)_e$ contained in 
$\U_n(\K)_e$, the action of $G$ on $M^r$ is effective. 

The space $G/K$ is Riemannian, but $\fp \cong \Herm_n(\K)$ contains  
the pointed generating invariant cone $\Herm_n(\K)_+$ which is invariant 
under $\Ad(K)$. Therefore $(\g,\theta,C)$ is also non-compactly causal. 
\end{ex}

The following lemma collects the implications of our assumptions. 

\begin{lem}
  \mlabel{lem:2.3.1} 
If {\rm(GP)} and {\rm(Eff)} are satisfied, then the following assertions hold:
\begin{itemize}
\item[\rm(a)] $\theta$ is a Cartan involution on $\g$.
\item[\rm(b)] Write $[G,G]$ for the commutator group of $G$. 
Then the  multiplication map \break $[G, G] \times Z(G)_e \to G$ 
is an isomorphism of Lie groups. 
\item[\rm(c)] $K = G^\theta \subeq [G,G]$. 
\item[\rm(d)] The universal complexification 
$\eta_G \: G \to G_\C$ is injective and 
$G_\C \cong [G,G]_\C \times \fz(\g)_\C$. 
\item[\rm(e)] $G \cong\Inn_\g([\g,\g]) \times \fz(\g)$. 
\item[\rm(f)] $Z(G)$ is a connected vector group.
\end{itemize}
\end{lem}

Condition {\rm(d)} implies in particular 
that the global type of $G$ is determined by {\rm(GP)} and {\rm(Eff)} 
and that the involutions $\tau$, $\theta$ and $\tau_h$ on $\g$ all integrate to the group~$G$.

\begin{prf} (a) $\fz(\g)^\theta = \{0\}$ follows from the fact that 
$\exp(\fz(\g)^\theta)$ is a normal subgroup contained in~$K$.

To see that $\theta$ restricts to a Cartan involution on $\g$, 
suppose that this is not the case. Then there exists a 
Cartan involution $\tilde\theta$ commuting with $\theta$. 
Accordingly $\fp_1 := [\g,\g]^{-\theta}$ decomposes 
under $\tilde\theta$ into eigenspaces $\fp_1^{\pm\tilde\theta}$. 
As the exponential function is supposed to be regular on $\fp$, 
this subspace contains no elliptic elements, and thus $\fp_1^{\tilde\theta} 
= \{0\}$. Therefore $\tilde\theta$ coincides with $\theta$ on 
the ideal $\fp_1 + [\fp_1,\fp_1]$ of~$\g$. As $\fk_1 := \fk \cap [\g,\g]$ 
is reductive and contains no non-zero ideal, we must have 
$\fk_1 = [\fp_1, \fp_1]$. This shows that $\tilde\theta = \theta\res_{[\g,\g]}$ 
is a Cartan involution. 

\nin (b) follows from (a) and the polar decomposition. 

\nin (c) The polar decomposition (GP) and 
$\theta(k\exp x) = k \exp(-x)$ 
for $k \in K$ and $x \in \fp$ imply that 
$G^\theta = K$. As $G$ is connected, (GP) further shows that 
$K$ is connected. The remaining assertion now 
follows from $\fz(\fg)^\theta = \{0\}$, which is (a). 

\nin (d) The first assertion follows from (b) and 
$Z(G)_e \cong (\fz(\g),+)$. For the second, 
we first observe that $\ker \eta_G = \ker \eta_{[G,G]}$ 
is a discrete normal subgroup 
of $G$, hence central. It is contained in $[G,G]$ because 
the inclusion $\fz(\g) \into \fz(\g)_\C$ is injective. 
As $\fk = \g^\theta$ is maximal compactly embedded 
in $[\g,\g]$ by (a), we have $Z([G,G]) \subeq K$ 
(\cite[Thm.~14.2.8]{HN12}), so that (Eff) implies 
that $Z([G,G])$ is trivial and thus $\eta_G$ is injective. 

\nin (e) We have already seen in (d) that $Z([G,G])$ is trivial. 
Therefore $\Ad \: [G,G] \to \Inn_\g([\g,\g])$ is injective, and this 
proves~(e).

\nin (f) As the center of $\Inn_\g([\g,\g) \cong \Inn([\g,\g])$ is trivial, 
this follows from (e).
\end{prf}

\subsection{Irreducible spaces} 
\mlabel{subsec:2.2}

In this section we recall that, for an irreducible causal symmetric 
Lie algebras $(\g,\tau)$, either 
\begin{itemize}
\item $\g$ is simple non-complex, 
\item isomorphic to $(\fh \oplus \fh, \tau_{\rm flip})$, 
$\fh$ simple non-complex, or to
\item  $(\fh_\C, \tau_{\rm conj})$, where $\fh$ is simple non-complex 
\end{itemize}
(\cite[Rem.~3.1.9]{HO97}).
 The latter two types are represented 
on the global level by the symmetric spaces 
$(H\times H)/\Delta_H$, where $\Delta_H = \{(h, h) \: h\in H\}$, and 
by $H_\C/H$. For general information about irreducible  causal symmetric spaces see \cite{HO97}. 

\begin{prop} \mlabel{prop:types}
For an  irreducible ncc symmetric spaces 
$(\g,\tau,C)$, the Lie algebra $\g$ is simple and 
we have the following possibilities: 
\begin{itemize}
\item[\rm(C)] {\rm(Complex type)}
If $\g$ is a complex Lie algebra and $\tau$ is antilinear,
then $\g \cong \fh_\C$, where $\fh$ is simple hermitian  
and $i C$ is a pointed generating invariant cone in $\fh$. 
An Euler element $h \in \fh$ exists if and only if $\fh$ is of tube type. 
\item[\rm(S)] {\rm(Simple type)} If $\g$ 
is a real simple Lie algebra which is not complex, then 
$\fq$ contains an Euler element $h_c$ and there exists a Cartan involution 
$\theta$  with $\tau = \theta \tau_{h_c}$.   
An Euler element $h \in \fh$ exists if and only if $\g^c$ is of tube type. 
\end{itemize}
\end{prop} 

\begin{prf} To see that $\g$ is simple, suppose that this is not the case. 
Then $(\g,\tau)$ is isomorphic to 
$(\fh \oplus \fh, \tau_{\rm flip})$ and $C \subeq \fq \cong \fh$ corresponds 
to  an invariant cone in the Lie algebra $\fh$, hence must be elliptic, 
contradicting the assumption that $(\g,\tau)$ is ncc.

\nin (C) If $(\g,\tau) \cong (\fh_\C, \tau_{\rm conj})$ is of complex type, 
where $\fh$ is a simple real Lie algebra, then $C \subeq \fq \cong i \fh$ 
corresponds to an invariant cone in $\fh$ and thus $\fh$ is simple hermitian. 
We recall that $\fh$  contains an Euler element if and only if it is 
 of tube type (cf.\ \cite[Prop.~3.11(b)]{MN21}). 

\nin (S) The existence of an Euler element $h_c \in \fq$ and a 
Cartan involution with $\tau = \theta \tau_{h_c}$ and $\theta(h)= -h$ 
follows from \cite[Thm.~4.4]{MNO22a}. 
The existence of an Euler element in $\fh$ 
is equivalent to the existence of an Euler elements 
in hermitian Lie algebra $\g^c$ which is contained in $\fh$, 
and by \cite[Prop.~2.7]{NO21} this is equivalent to $\g^c$ being of tube type. 
\end{prf}

\begin{prop} \mlabel{prop:decomp} 
Suppose that $(\g , \tau , C)$ is a reductive 
ncc symmetric Lie algebra. Assume
furthermore   that $\fh$ contains no non-zero ideal 
and that $\theta$ is a Cartan involution commuting with $\tau$.  
Then the symmetric Lie algebra $(\g,\tau)$ decomposes as 
\[ (\g,\tau) 
\cong (\g_0, \tau_0) \oplus \bigoplus_{j = 1}^N (\g_j, \tau_j), \]
where $\tau_0$ is a Cartan involution and each  simple symmetric Lie algebra 
$(\g_j,\tau_j)$, $j \geq 1$,  is irreducible ncc, hence in particular 
of complex or simple type. 
\end{prop}

\begin{prf} As $\g$ is reductive, it decomposes into $\tau$-invariant ideals. Let
$\fg_0$ be a maximal  $\tau$-stable ideal in $\fg$ such that 
$\tau\res_{\g_0}$ is a Cartan involution. 
Then 
\[\fg = \fg_0 \oplus \bigoplus_{j=1}^N \fg_j \]
where each $\fg_j$ is a non-compact ideal, 
so that $(\fg_j,\tau)$ is irreducible, i.e., $\fg_j$ does not
contain any non-trivial $\tau$-invariant ideals, and 
$\tau\res_{\fg_j}$ is not Cartan. 
As $\fh$ contains no non-trivial ideal, the center $\fz(\g)$ 
is contained in $\fq$, hence in $\g_0$.

For $k > 0$, let  $p_k :\fg\to \fg_k$ be
the projection with kernel $\bigoplus_{j\not= k} \fg_j$ 
and recall from Proposition~\ref{prop:types}  that $\g_k$ is simple.
As $C$ is generating, $C_k:=\oline{p_k (C)}\not=\{0\}$. 
Let $H_k = \la \exp \fh_j \ra \subeq G_j$ be the integral  subgroup 
with Lie algebra $\fh_j$. Then $C_k$ and $W_k = C_k\cap -C_k$ are $\Ad (H_k)$
invariant. The Jacobi identity then implies that $[W_k,W_k]\oplus W_k$ is
a $\tau$ invariant ideal in $\fg_k$ and hence equals either $\fg_k$ or~$\{0\}$. 
As $C$ is hyperbolic, the first case can only occur 
if $\tau\res_{\fg_k}$ is a Cartan involution,
contradicting our assumption on $\fg_k$.  Hence
$C_k$ is an $\Ad (H_j)$-invariant pointed generating hyperbolic cone and it
follows that $(\fg_k,\tau_k)$ is ncc, hence in particular simple by 
Proposition~\ref{prop:types}. 
\end{prf}
 
\subsection{Compatible Cartan involutions}
\mlabel{subsec:3.5}

In this section $\fg$ is always assumed to be reductive and $h$ will always stands for
an Euler element in $\fg$. We introduce the notion of a 
compatible Cartan involution and prove that
a compatible involution always exists. We start with some simple observations
about invariant cones in $\fq$ and compatible Cartan involutions.

\begin{defn}
For a reductive Lie algebra $\g$, a 
Cartan involution $\theta$ of $\g$  
is said to be {\it compatible} with $(\g,\tau,C,h)$ if 
\[ \theta\tau = \tau \theta\quad \mbox{ and }\quad \theta(h) = -h.\]
We then put 
\begin{equation}
  \label{eq:theta-tau-dec}
 \fk = \g^\theta, \quad \fp = \g^{-\theta}, \quad 
\fh_\fk := \fh \cap\fk, \quad 
\fh_\fp := \fh \cap\fp, \quad 
\fq_\fk := \fq \cap\fk, \quad 
\fq_\fp := \fq \cap\fp.
\end{equation}
\end{defn}

\begin{lem} \mlabel{lem:compcartan} 
Suppose that  $(\g,\tau, C,h)$ is a reductive 
modular non-compactly causal symmetric 
Lie algebra. Then compatible Cartan involutions for $(\g,\tau,C,h)$ exist. 
\end{lem}

\begin{prf}  Applying \cite[Prop.~I.5(iii)]{KN96} 
to the semisimple commutator algebra 
and extending by $-\id_{\fz(\g)}$, 
we obtain a Cartan involution $\theta$ of $\g$ commuting with $\tau$. 
Then $h \in \fh$ is a hyperbolic element of $\g$, hence 
conjugate under $\Inn_\g(\fh)$ to an element of $\fh_\fp = \fh^{-\theta}$ 
(apply \cite[Cor.~II.9]{KN96} to $(\fh \oplus \fh, \tau_{\rm flip})$). 
Conjugating $\theta$ accordingly, we thus obtain a compatible 
Cartan involution. 
\end{prf}

From now on $\theta$ always denotes a compatible Cartan involution. 
We will need the following Extension and Restriction Theorems 
(\cite[Thm.~X.7]{KN96}, \cite[Thm. 4.5.8]{HO97}): 

\begin{thm}{\rm(Cone Extension Theorem)} 
\mlabel{thm:extend} 
Suppose that  $\fh$ contains no non-zero ideals 
of $\g$. Then every 
$\Inn_\g(\fh)$-invariant closed convex pointed generating 
{\bf elliptic} invariant cone $C \subeq \fq$ 
can be extended to an $\Ad(G)$-invariant, 
$-\tau$-invariant 
pointed generating closed convex cone $C_\fg \subeq \g$. It satisfies 
\[ p_\fq (C_\fg) =   C_\fg \cap \fq = C. \] 
\end{thm}

For the non-compactly causal spaces, we obtain by duality: 
\begin{cor} \mlabel{cor:ext-ncc}
Let $(\g,\tau,C)$ be a reductive non-compactly causal 
symmetric Lie algebra for which $\fh$ contains no non-zero ideals  
of $\g$. Then there exists a closed convex pointed generating cone 
$C_{\g^c} \subeq \g^c$, invariant under $-\tau^c$ and $\Inn_\g(\g^c)$, 
such that 
\[ p_{i\fq} (C_{\fg^c}) =   C_{\fg^c} \cap i\fq = iC. \] 
\end{cor}

In this section we put $H_K :=\Inn_\fg (\fh_\fk)$. 
As this group is compact, 
averaging with the Haar measure on $H_K$ 
leads for every $x \in C^\circ$ to an $\Ad(H_K)$-fixed point 
\[x_0= \int_{H_K} \Ad(k)x\, d\mu(k)\in C^\circ \cap \fq^{H_K}\]  
(see \cite[Lem 1.3.5]{HO97} and the proof of Lemma 2.1.15 in \cite{HO97}). 

The following lemma recalls crucial structural information concerning 
properties of elements in $C^\circ \cap \fq^{H_K}$.

\begin{lem}\mlabel{lem:thetaCom}
 Let $(\fg,\tau ,C)$ be a reductive causal symmetric Lie algebra 
and $x_0 \in C^\circ \cap \fq^{H_K}$ be an $\Ad(H_K)$-fixed point.
Then the following assertions hold: 
\begin{itemize}
\item[\rm (i)] If $x_0 $ is hyperbolic, then $x_0 \in \fq_\fp$
  and $x_0 \in \fz(\fh_\fk + \fq_\fp)$. 
\item[\rm (ii)] If $x_0$ is elliptic, then $x_0 \in \fz(\fk) \cap \fq$.
\end{itemize}
\end{lem} 

\begin{proof} The group $H := \Inn_\g(\fh)$ is the identity component
  of the group $\Aut([\g,\g])^\tau$, hence closed. It has the polar decomposition
  $H = H_K e^{\ad \fh_\fp}$. 
  (\cite[Prop.~13.1.5]{HN12}).
  The non-compact ideals of $\fh$ all come from irreducible ncc symmetric Lie
  subalgebras of $(\g,\tau)$, hence are contained in $[\fq,\fq]$.
  Therefore the representation 
  $\Ad_\fq$ of $H$ on $\fq$ is faithful and its image in $\GL(\fq)$ is closed with
  maximal compact subgroup $\Ad_\fq(H_K)$.

  \nin   (i)  The compactness of the stabilizer of $x_0$ in $\Ad_\fq(H)$ 
  (\cite[Prop.~V.5.11]{Ne00}) implies that $H^{x_0} = H_K$.
  As $x_0 \in C^\circ$ is hyperbolic, there exists $g\in \Inn_\fg (\fh)$
such that $x := g.x_0\in  \fq_\fp$ (\cite[Cor.~II.9]{KN96}). 
Then $\theta(x) = - x$ implies that the stabilizer $H^x$ of $x$ in $H$
is $\theta$-invariant. 
So $g_1 = h_1 e^{\ad z} \in H^x$ with $h_1 \in H_K$
 and $z \in \fh_\fp$ implies $\theta(g_1) g_1 = e^{2 \ad z} \in H^x$. The compactness 
 of $H^x = g H^{x_0} g^{-1}  = g H_K g^{-1}$ now entails $z = 0$. Hence
  $H^x \subeq H_K$ and thus $H^x = H_K$ as $H^x$  is maximally compact.
  We conclude that $g H_K g^{-1} = H_K$ which further implies $g \in H_K$,
  the stabilizer group of the base point in the Riemannian symmetric space~$H/H_K$.
  This shows that $x = x_0 \in \fq_\fp$.

Further $[x_0,\fh_\fk + \fq_\fp]=\{0\}$ follows from \cite[Lem. 1.3.5]{HO97}.
For the sake of completeness, we recall the simple argument: 
Let $\kappa$ denote the Cartan--Killing form, which is definite 
on $\fh_\fk$ and $\fq_\fp$. Then 
\[  \kappa([x_0, \fq_\fp], \fh_\fk) = -\kappa(\fq_\fp, [x_0,\fh_\fk]) = \{0\} 
\quad \mbox{ and }\quad [x_0,\fq_\fp] \subeq \fh_\fk \] 
implies that $[x_0, \fq_\fp] =\{0\}$.

\nin (ii) If $x_0\in\fq$ is elliptic, then $ix_0\in i\fq \subset
\fg^c=\fh\oplus i\fq$ is hyperbolic. By (i) we have 
$ix_0 \in \fq^c_\fp = i \fq_\fk$, so that
$x_0 \in \fq_\fk$. Moreover, $ix_0$ is central in $\fh_\fk + \fq^c_\fp = \fh_\fk +
i \fq_\fk$, and therefore $x_0 \in \fz(\fk) = \fz(\fh_\fk + \fq_\fk)$.
\end{proof}

We refer to Theorem \ref{thm:3.6} for further discussion related to the following 
theorem.

\begin{thm} \mlabel{thm:2.20} {\rm (Cone Restriction Theorem)} 
Assume that $(\fg,\tau, C)$ is a 
reductive causal symmetric Lie algebra and that $\theta $ 
is a Cartan involution as in {\rm Lemma \ref{lem:thetaCom}}
with $\theta (x_0) = \mp x_0$ depending on $x_0$ being hyperbolic or elliptic.
\begin{itemize}
\item[\rm (i)] If $C$ is hyperbolic and $\fa\subset \fq_\fp$ 
is maximal abelian, then
$C^\circ = \Inn_\fg (\fh)(C^\circ\cap \fa)$ and $\theta C = -C$. 
\item[\rm (ii)] If $-\theta C =C$ and $(\g,\tau)$ is irreducible, then $C$ is hyperbolic.
\item[\rm (iii)] If $C$ is elliptic and $\ft\subset \fq_\fk$ 
is maximal abelian, then $C^\circ = \Inn_\fg (\fh) (C^\circ\cap \ft)$ and
$\theta C = C$.
\item[\rm (iv)] If $\theta C=C$ and $(\g,\tau)$ is irreducible,
 then $C$ is elliptic.
\end{itemize}
\end{thm}

\begin{proof} (i) The first part of (i) is \cite[Thm.~4.4.11]{HO97}. 
As $\theta|_{\fa} = -\id_\fa$ and $\theta(\fh) = \fh$, 
it follows that 
$\theta(C^\circ) 
= \Inn_\fg (\fh)(\theta(C^\circ\cap \fa)) 
= -\Inn_\fg (\fh)(C^\circ\cap \fa) = - C^\circ$, and by taking closures we get 
$\theta (C ) = -C$.

\nin
(ii)  Assume that $-\theta C = C$. Then 
$C^\circ \cap \fq^{\fh_\fk}$ is invariant under $-\theta$. 
Replacing $x_0$ by $y_0 := \frac{1}{2} (x_0 - \theta x_0) \in C^\circ\cap \fp$, 
we may therefore assume that $x_0 \in \fq_\fp$. 
The $\Inn_\fg(\fh)$-invariant closed convex cone $C_m$  
generated by the orbit $\Inn_\fg (\fh)x_0$ is contained in $C$. 
As $(\g,\tau)$ is irreducible, $C_m$ is
the minimal $\Inn_\g(\fh)$-invariant 
cone $C_{\rm min}(x_0)$ in $\fq$ containing $x_0$ (\cite[Prop.~3.1.3]{HO97},
\cite[\S 3.5.2]{MNO22a}) 
and it follows that
\[C_{\rm min} \subset C_m \subeq C\subset C_{\rm min}^\star = C_{\rm max}.\]
This implies that $C$ is hyperbolic because $C_{\rm max}$ is hyperbolic 
by \cite[Thms.~VI.6, IX.9]{KN96}.

\nin 
(iii) and (iv) follow 
from (i) and (ii) by replacing $\fg$ by $\fg^c = \fh \oplus i\fq$ and $C$ by $iC$,
\end{proof}

The following remark clarifies the necessity of the 
irreducibility in the preceding theorem. 

\begin{rem} The irreducibility assumption in 
Theorem~\ref{thm:2.20}(ii),(iv) is needed. In fact, consider 
a direct sum 
\[ (\g,\tau) = (\g_1, \tau_1) \oplus (\g_2, \tau_2),\] 
where $(\g_1, \tau_1)$ is non-compactly causal and 
$\g_2$ is a compact Lie algebra. So 
$\theta = \theta_1 \oplus \id_{\g_2}$ and 
\[ \fq = \fq_{1,\fp} \oplus \underbrace{\fq_{1,\fk} \oplus \fq_2}_{\fq_\fk}.\] 

Let $C_1 \subeq \fq_1$ be a pointed generating hyperbolic 
cone which is invariant under 
$\Inn_\g(\fh_1)$ and $-\theta_1$. 
Then there exists a closed convex 
subset $D \subeq C_1^\circ$ with interior points 
invariant under $\Inn_\g(\fh_1)$ and $-\theta_1$.\begin{footnote}
{As in the proof of Lemma~\ref{lem:closedorbit-gen} below, this 
follows from the $\Inn_\g(\fh_1)$-invariance of the characteristic 
function of $C_1$ by taking $D = \phi^{-1}([1,\infty))$.}\end{footnote}
Let  
$B \subeq \fq_2$ be a symmetric convex $0$-neighborhood invariant 
under the compact group $\Inn_\g(\fh_2)$. Then 
\[ C := \oline{\R_+(D + B)} \subeq \fq \] 
is a pointed generating $\Inn_\g(\fh)$-invariant cone but 
$C$ is not hyperbolic because the $\fq_2$-components of elements in $\fq$ 
are elliptic. We refer to \cite[Lemma~B.2]{MNO22a} for a more general 
result of this type.

By duality, the irreducibility assumption in (iv) is also needed.   
\end{rem}

We now turn to the structural implications of the existence of
an Euler element. Suppose that $C \subeq \fq$ is a closed convex $\Inn_\g(\fh)$-invariant 
cone also invariant under $-\tau_h$.  Since $\ad h$ preserves 
$\fq$, the space $\fq$  decomposes as 
\[ \fq = \fq_{-1}(h) \oplus \fq_0(h) \oplus \fq_{+1}(h).\]

Note that $\tau\tau_h = \tau_h\tau$ so that $\tau^{ct} :=\tau_h\tau$ is an involution and 
\begin{equation}
  \label{eq:p10}
\fg^{\tau^{ct}} = \fh^{\tau_h} \oplus \fq^{-\tau_h},\quad \fg^{-\tau^{ct}} = \fh^{-\tau_h}\oplus \fq^{\tau_h}
\quad\text{and}\quad \fq^{-\tau_h}= \fq_{+1}(h) \oplus \fq_{-1} (h).
\end{equation}

\begin{prop}  \mlabel{prop:2.18} 
Let $H := \Inn_\g(\fh)$. 
Then the cones 
\[ C_\pm := (\pm  C)\cap \fq_{\pm 1}(h) 
\quad \mbox{ and } \quad C^{\rm ct} := C_+-C_-\] 
are pointed, generating and invariant under 
the centralizer subgroup $H^h = Z_H(h)$ in 
$\fq_{\pm 1}(h)$ and~$\fq^{-\tau_h}$, respectively, 
and they can also be obtained as projections: 
\begin{equation}
  \label{eq:coneproj}
p_{\fq_{\pm 1}(h)}(C) = \pm C_\pm \quad \mbox{ and } \quad 
p_{\fq^{-\tau_h} } (C ) = C^{-\tau_h} = C_+-C_-.
\end{equation}
Furthermore the following holds:
 \begin{itemize}
\item[\rm (i)]   $(\fg^{ct}, \tau |_{\fg^{ct}}, C^{ct} ,h)$ is a modular
causal symmetric Lie algebra of
Cayley type. 
\item[\rm (ii)] If   $-\theta C = C$, then $\theta C_+ = C_-$.  
\item[\rm (iii)] If $C$ is elliptic and $\theta C = C$, then $\theta C_+ 
= -C_-$ and
$ C_+ + C_-$ is elliptic.
\end{itemize}
\end{prop}

\begin{proof} Clearly,  $C_\pm$ are closed convex cones 
and $C_\pm \subseteq p_{\fq^{-\tau_h}}(C)$. 
Let $p_\pm \: \fq \to \fq_{\pm 1}(h)$ 
denote the projections. Clearly $C_\pm  \subseteq
 p_{ \pm } (\pm C)$. To show equality, let $x\in C$ and write
$x=x_+ + x_0 + x_-\in C$ with $x_\pm \in \fq_{\pm 1}(h)$ and 
$x_0 \in \fq_0(h)$. Then, as $e^{\R \ad h}C \subseteq C$,
\[\lim_{t\to \infty}e^{-t}e^{\ad h}x = x_+\in C\cap \fq_{1}(h)\quad\text{and}\quad 
\lim_{t\to -\infty}e^{t}e^{\ad h}x = x_-\in C\cap \fq_{-1}(h).\]
It follows that $p_{ \pm }(\pm C) = C_\pm$ and hence 
$p_{\fq^{-\tau_h}}(C)= C^{-\tau_h} =C_+- C_-$. 

Since $C$ is generating, the cones $C_\pm$ are generating in 
$\fq_{\pm 1}(h)$ and $C_+- C_-$ is generating in $\fq^{-\tau_h}$. 
Further, the invariance of $C$ under $H = \Inn_\g(\fh)$ 
and the fact that $H^h$ commutes with $\ad h$, 
hence leaves $\fq_{\pm1}(h)$ invariant,  
show that the cones $C_\pm$ are invariant under $-\tau_h$ and $H^h$.  
The relation $\tau (h) = h$ entails $h\in \fh^{\tau_h}$ and, by definition, 
$\tau|_{\fg^{ct}} = \tau_h|_{\fg^{ct}}$. 

\nin (i) The preceding discussion implies in particular that 
$(\fg^{\rm ct},\tau,C_+-C_-,h)$ is of Cayley type. 

\nin
(ii) As $\theta (h) = -h$ it follows that $\theta \fq_{\pm 1}(h) = \fq_{\mp 1}(h)$. Hence, if $-\theta C = C$ then
$\theta C\cap \fq_{\pm 1}(h) = -C \cap \fq_{\mp 1}(h)$. Thus $\theta C_+ = C_-$. 

\nin
(iii) As above, $\theta(C) = C$ implies 
$\theta(C_+) =- C_-$. Let $x_\pm \in \fq_{\pm 1}(h)$. Then
\begin{equation}
  \label{eq:kapparelx}
e^{\frac{\pi i }{2}\ad h} x_\pm = e^{\pm \frac{\pi i}{2}}x_\pm = \pm i x_\pm 
\quad \mbox{ implies } \quad 
\kappa_h (x_+ - x_-) = i(x_+ +   x_-).   
\end{equation}
We conclude that $-i\kappa (C_+-C_-) = C_+ + C_-$. 
Since $C_+ - C_- \subeq C$ is hyperbolic, it follows that 
$C_+ + C_-$ is elliptic.   
\end{proof}

As some of the relevant structure of $(\g,\tau,C,h)$ is 
represented by the  Cayley type subalgebra generated by $C^{\rm ct}$, 
some arguments concerning wedge domains reduce to 
the case of Cayley type spaces. 

 \begin{cor}\mlabel{cor:Cayley type} 
If $(\fg,\tau_h , C)$ is of Cayley type, then 
\begin{equation}\label{eq:Cpm}
C= C_+ -C_-,\end{equation}
and the following assertions hold: 
\begin{itemize}
\item[\rm(i)] The automorphism $\kappa_h = e^{\frac{\pi i }{2}\ad h}$ of $\g_\C$ 
satisfies $-i \kappa_h(C) = C_+ + C_-.$  \item[\rm(ii)] The cone $C$ is hyperbolic if and only if 
$C^c := C_++ C_-$ is elliptic,
\end{itemize}
\end{cor}

\begin{proof}  First we observe that 
\eqref{eq:Cpm} follows from \eqref{eq:coneproj} in 
Proposition~\ref{prop:2.18} because 
$\fq_0(h) = \{0\}$.

\nin
(i) follows from \eqref{eq:kapparelx}.  

\nin (ii) Let $x=x_+-x_-\in C^\circ$. 
Then (i) implies that $x$ is hyperbolic if and 
only if $x_+ + x_-$ is elliptic. 
\end{proof} 

\begin{ex} \mlabel{ex:gl2} (A non-Cayley type space which is cc and ncc) 
We consider the reductive Lie algebra 
$\g = \gl_2(\R)$ and the Euler element 
$h =  \diag(1/2,-1/2)$. 
We define the involution $\tau$ on $\g$ 
by extending $\tau_h= e^{\pi i \ad h}$ on $\fsl_2(\R)$ by 
$\tau(\1) = - \1$, so that $\fq = \g^{-\tau}$ is $3$-dimensional. 
Concretely, we have 
\begin{equation}\label{eq:tauGL}
\tau\begin{pmatrix} a & b \\ c & d\end{pmatrix}
= \begin{pmatrix} -d & -b \\ -c & -a\end{pmatrix}, 
\qquad 
\fh = \R h\quad\text{and}\quad \fq 
= \left\{\begin{pmatrix} z & x \\ y & z \end{pmatrix}\: x,y,z\in \R\right\}.
\end{equation}
All three eigenspaces 
$\fq_{-1}(h)$, $\fq_{1}(h)$ and $\fq_0(h) = \R \1$ are non-zero. 
Let 
\[ \be_0 := \1, \quad \be_1 := \pmat{0 & 1 \\ 0 & 0} \quad \mbox{ and } \quad 
\be_{-1} := \pmat{0 & 0 \\ 1 & 0} \]  
be corresponding eigenvectors. Then 
the hyperbolic element $h_c = \frac{1}{2}(\be_1 + \be_{-1})$ 
generates the cone 
\[ C_+ - C_- \subeq \fq_1(h) + \fq_{-1}(h) 
\quad \mbox{ with } \quad C_\pm = \pm \R_+ \be_{\pm 1}.\] 
The involution $-\tau_h$ acts on $\fq$ by the hyperplane reflection 
\[ -\tau_h(\be_{\pm 1}) =  \be_{\pm 1}, \quad  -\tau_h(\be_0) = -\be_0.\] 
It is now easy to describe all pointed generating closed convex cones 
$C \subeq \fq$ which are invariant under $-\tau_h$ and 
$\Inn_\g(\fh) = e^{\R \ad h}$. 
According to \cite[Ex.~3.1(c)]{NOO21}, 
all these cones are Lorentzian of the form 
\[ C^m = \{ x_0 \be_0 + x_1 \be_1 + x_{-1} \be_{-1} \: 
x_1 x_{-1} - m x_0^2 \geq 0, x_{\pm 1} \geq 0\}
\quad \mbox{ for some } \quad m > 0. \] 

We conclude that the symmetric spaces 
$M := \GL_2(\R)/H$ for $H = \exp(\R h)$ is 
$3$-dimensional Lorentzian but not of Cayley type. 
However, replacing $C_-$ by $-C_-$, we obtain not only non-compactly causal 
structures but also compactly causal ones. 
\end{ex}

\begin{ex} (A modular non-Cayley type space) 
We consider the Lie algebra $\g = \fsl_n(\R)$ with the 
Cartan involution $\theta(x) = - x^\top$ 
and write $n = p + q$ with $p,q> 0$. Then 
\begin{equation}
  \label{eq:euler-pq}
 h_p := \frac{1}{n}\pmat{ q \1_p & 0 \\ 0 & -p\1_q} 
\end{equation}
is an Euler element and 
$\tau := \tau_h\theta$ leads to a non-compactly causal symmetric 
Lie algebra $(\g,\tau,C)$, where 
\[ \fh = \so_{p,q}(\R) \quad \mbox{ and } \quad 
\fq = \Big\{\pmat{a & b \\ -b^\top & d} \: a^\top = a, d^\top = d, 
\tr(a) + \tr(d) =0\Big\}\] 

Now let $h \in\fh = \so_{p,q}(\R)$ be an Euler element of $\g$. 
Then $h$ has two eigenspaces
because it is conjugate to an element of the form 
\eqref{eq:euler-pq} for a possibly different partition of $n = r+s$. 
These eigenspaces must be isotropic and in duality,
hence of the same dimension. This is only possible 
if $p = q$ and if $h$ is conjugate to $h_p$. 
We conclude that $\fh$ contains an Euler element if and only if 
$p = q$, so that $n$ must be even. In this case 
\[ h := \frac{1}{2} \pmat{ 0 & \1 \\ \1 & 0} \in \so_{p,p}(\R) \] 
is an Euler element which is not central in $\so_{p,p}(\R)$ for $p > 1$. 
We have
\[ \fq_{\pm 1}(h) = \Big\{\pmat{a & \mp a \\ \pm a & -a} \: a \in \Sym_p(\R)\Big\}
\cong \Sym_p(\R), \quad 
\fq_0(h) = \Big\{\pmat{0& b \\ b & 0} \: b \in \gl_p(\R)\Big\}.\] 
It is easy to see that the Lie algebra $\g^{\rm ct}$ generated by 
$\fq_{\pm 1}(h)$ is isomorphic to $\sp_{2p}(\R)$ 
(see Example~\ref{ex:2.24}). 
\end{ex}

\begin{ex} \mlabel{ex:2.24} (An example with many modular structures) 
We consider the Cayley type symmetric Lie algebra 
$\g = \sp_{2n}(\R)$ with $\fh = \gl_n(\R)$ and 
\[ h = \frac{1}{2} \pmat{ \1_n & 0 \\ 0 & -\1_n}.\] 
Then $\fh$ contains several Euler elements 
\[ h_p := \frac{1}{2} \pmat{ 
\1_p & && \\ 
& -\1_q & && \\ 
& & -\1_p & & \\ 
& && \1_q }, \quad 1 \leq p \leq n, p+q= n\]
and all these commute with $h = h_n$. Although 
all Euler elements in $\g$ are conjugate under $\Inn(\g)$ 
(cf.\ \cite[Thm.~3.10]{MN21}), 
the Euler elements $h_p$ represent different $\Inn(\fh)$-orbits 
of Euler elements in $\fh\cong \gl_n(\R)$, as the dimensions of their
 eigenspaces are different. We have 
\[ \fq = \g_1(h) \oplus \g_{-1}(h) \cong \Sym_n(\R) \oplus \Sym_n(\R),\] 
and for $p < n$ both summands split into $h_p$-eigenspaces:  
\[ \fq_1(h_p) \cong\Sym_p(\R)^{\oplus 2}, \quad 
\fq_{-1}(h_p) \cong\Sym_q(\R)^{\oplus 2}, \quad 
\fq_0(h_p) \cong M_{p,q}(\R)^{\oplus 2}.\] 
\end{ex}

\begin{lem} \mlabel{lem:Cayley} 
Suppose that  $(\g,\tau, C,h)$ is a reductive 
modular non-compactly causal symmetric 
Lie algebra. Then the subalgebra generated by $\fq^{-\tau_h}=\fq_{+1}(h)\oplus \fq_{-1}(h)$ 
is a direct sum of simple symmetric Lie algebras of Cayley type.  
\end{lem} 

\begin{prf} The Jacobi identity implies that 
\[[h, [\fq^{-\tau_h},\fq^{-\tau_h}]] = [h, [\fq_{+1}(h),\fq_{-1}(h)]]=\{0\}.\]
Therefore 
\[ \fg^{\rm ct} = [\fq^{-\tau_h}, \fq^{-\tau_h}]+ (\fq_{+1}(h)+\fq_{-1}(h))\]
is a symmetric Lie subalgebra, 
generated by $\fq^{-\tau_h}$. With  $\tau^{\rm ct} =\tau|_{ \fg^{\rm ct}}$ 
we  obtain a semisimple causal symmetric Lie algebra 
 $(\fg^{\rm ct},\tau^{\rm ct},   C_+-C_-)$ of Cayley type. Write 
\[(\fg^{\rm ct},\tau^{\rm ct} )= \bigoplus_{j = 1}^N (\fg_j,\tau_j)\]
with each $ (\fg_j,\tau_j)$ irreducible.
We have $\tau^{\rm ct} = \tau_{h^\prime}$ where $h^\prime$ is the projection of
$h$ onto $\fg^{\rm ct}$. Write $h^\prime = \sum_j h_j$ with $h_j\in \fg_j$.
As $\ad h^\prime$ does act injectively on $\fq^{\rm ct} =\fq_{+1}(h)+\fq_{-1}(h)$
and $\fg^{\rm ct}$ is generated by $\fq_{+1}(h)+\fq_{-1}(h)$ it follows that there is no simple ideal with
trivial intersection with $\fq^{\rm ct}$. It follows that $h_j\not= 0$ for all $j$ and
$\tau_j = \tau_{h_j}$. Thus $(\fg_j,\tau_j)$ is of Cayley type. 
\end{prf}
 
\section{Root decomposition for 
non-compactly causal spaces} 
\mlabel{sec:3}

In this section we turn to the fine structure
of causal symmetric Lie algebras,
which is encoded in the structure of restricted roots
of a non-compactly causal reductive symmetric Lie algebra. It turns out that
the set of roots decomposes nicely into so-called 
compact and non-compact roots, and positive systems of
non-compact roots determine 
minimal and maximal $\Ad(H)$-invariant cones
$C_\fq^{\rm min} \subeq C_\fq^{\rm max}$ in $\fq$. 
This information goes back to \cite{OO91} but was systematized 
in \cite{Ol91} and developed further in \cite{HO97,KN96,HN93}.
In view of Proposition \ref{prop:decomp}, our
discussion includes also the root space decomposition of 
Riemannian and  semisimple ncc symmetric
spaces. The semisimple case was first discussed in \cite{Ol91} and the 
non-reductive case was studied in \cite{KN96}.
As a key technical tool in the structure theory, we use
in Section~\ref{subsec:olaf}  maximal $\tau$-invariant
sets of long strongly orthogonal roots. 
In Subsection~\ref{sect:Complex} we introduce the important method
of embedding ncc symmetric spaces into spaces of complex type,
which are of the form $G_\C/G$.

\nin {\bf Setting:} Let $(\g,\tau, C)$ be a reductive non-compactly causal 
symmetric Lie algebra with $\fz(\g) \subeq \fq$ 
and fix a Cartan involution $\theta$ commuting with $\tau$, 
so that $\fz(\g) \subeq \fq_\fp$.  

For further references we state here also the following simple facts that are well known for the semisimple case 
(\cite[Cor.~II.9]{KN96}, \cite[Lemma~1.2]{Ol91}, and  \cite[Lemma~1.3.5]{HO97}):

\begin{lem} \mlabel{lem:3.1}
Let $(\fg,C,\tau)$ be a 
reductive non-compactly causal symmetric Lie algebra. 
Then the following holds: 
\begin{itemize}
\item[\rm  (a)] $(\g,\tau)$ is quasihermitian, 
i.e., 
the center 
\[ \fz := \fz(\fq_\fp) = \{ x \in \fq_\fp \: [x,\fq_\fp] = \{0\}\} \] 
of $\fq_\fp$ 
satisfies $\fz_\fq(\fz) := \{ x \in \fq \: [x,\fz] = \{0\}\} = \fq_\fp$ . 
\item[\rm (b)] If $\fa \subeq \fq$ is maximal abelian hyperbolic 
subspace,\begin{footnote}{A subspace is called {\it hyperbolic} if it consists 
of hyperbolic elements.}\end{footnote}
then any other maximal abelian hyperbolic subspace of $\fq$ 
is conjugate to $\fa$ under $\Inn_\g(\fh)$. In particular, 
every $\Inn_\g(\fh)$-orbit in $C^\circ$ intersects~$\fa$, i.e., $C^\circ = \Inn_\fg (\fh).(C^\circ\cap \fa)$.
\end{itemize}
\end{lem} 
\subsection{Compact and non-compact roots}
\mlabel{subsec:3.1} 

We now discuss the set of {\it compact} and {\it non-compact} roots and their impact on the structure of ncc spaces.

\begin{defn}
 We pick a maximal abelian subspace $\fa \subeq \fq$ consisting of hyperbolic 
elements and write 
\[ \g^\alpha = \{ y \in \g \: (\forall x \in \fa)\ 
[x,y] = \alpha(x) y \} \] 
for the $\fa$-weight spaces in $\g$ (also called root spaces) and 
\[ \Delta := \Delta(\g,\fa) := \{ \alpha \in \fa^* \setminus \{0\} \: 
\g^\alpha \not=\{0\} \} \] 
for the corresponding set of {\it restricted roots}.  
\end{defn}

\begin{defn} \mlabel{def:3.1} 
(a) A root $\alpha \in \Delta(\g,\fa)$ is called {\it compact} 
if $\g^\alpha \subeq \fh_\fk + \fq_\fp = \g^{\tau\theta}$, and 
{\it non-compact} otherwise. We write $\Delta_k$ for the subset 
of compact roots and $\Delta_p$ for the subset of non-compact roots. 

\nin (b) The (compact) 
{\it Weyl group} $\cW _k \subeq \GL(\fa)$ is the subgroup generated by 
the reflections 
\[ s_\alpha(x) := x - \alpha(x) \alpha^\vee, \quad \alpha \in \Delta_k,\]
where $\alpha^\vee \in \fa$ is the corresponding coroot, i.e., 
$\alpha(\alpha^\vee) = 2$ and 
$\alpha^\vee \in \R [x_\alpha, \theta(x_\alpha)]$ for 
some $x_\alpha \in \g^\alpha$. 
\end{defn}

\begin{rem} \mlabel{rem:3.4} (a) As $(\g,\tau)$ is non-compactly causal, 
$\fq_\fp$ contain an Euler element $h_c$ 
  and we may assume that $h_c\in\fa$ \cite[Thm.~4.4]{MNO22a}.   
Then either $\fg^\alpha \subset \fg^{\tau_h}=\fz_\fg (h_c) =
\fh_\fk\oplus \fq_\fp$ or $\fg^\alpha \subset \g^{-\tau_{h_c}}= \h_\fp\oplus 
\fq_\fk$ (\cite[Prop.~V.9]{KN96}). 
Furthermore,
if $\fa\subset \fq_\fp$ is maximal abelian in $\fq_\fp$ containing $h_c$,  
then  $\fa$ is maximal abelian in $\fq$ and $\fp$ because 
$\g^{\tau_{h_c}} \cap \fq = \g^{\tau_{h_c}} \cap \fp = \fq_\fp$. 

\nin (b) We note that $\Delta_k=\Delta (\fg^{\tau_h},\fa)$ is the 
system of  restricted roots of the 
Riemannian symmetric 
Lie algebra $(\fg^{\tau_{h_c}}, \theta)$. 
Thus facts about the restricted root systems of 
Riemannian symmetric spaces apply in particular to $\Delta_k$.

\nin (c)  The Weyl group $\cW_k$ is the Weyl group associated to the 
root system $\Delta_k$, so that \break $\cW_k\Delta_k = \Delta_k$. Let 
\[\Delta_p^\pm = \{\alpha \in \Delta \: \alpha (h_c)=\pm 1\}.\]
Then we also have $\cW_k\Delta_p^\pm = \Delta_p^\pm$.
Any positive system $\Delta^+$ with $\Delta^+ \cap \Delta_p = \Delta_p^+$ 
is {\it $\fp$-adapted} in the sense that 
$\Delta_p^+$ is $\cW _k$-invariant (\cite[Prop.~V.10]{KN96}). 
A root $\alpha$ is compact if and only if it vanishes on $\fz({\fq_\fp}) \subeq \fa$ 
(\cite[Prop.~V.9]{KN96}).
\end{rem}

The Weyl group orbits in $\fa$ classify the $\Inn_\g(\fh)$-orbits 
of hyperbolic elements in $\fq$: 
\begin{lem} \mlabel{lem:3.2}  
{\rm(\cite[Thm.~III.10, Prop.~V.2]{KN96})} 
Any $\Inn_\g(\fh)$-orbit of a hyperbolic element in $\fq$ intersects 
$\fa$, and, for every $x \in \fa$, we have 
\[ \Inn_\g(\fh)x \cap \fa = \cW_k x\] 
\end{lem}

\begin{ex} \mlabel{ex:cc1} (Compact roots in the complex case) 
For a symmetric space $(\g_\C,\tau)$ of complex type 
with $\tau(x + i y) = x- iy$, 
we have $\fq = i\fg$. If $\fa$ is maximal abelian in $\fq$, 
then $\fa = i \ft$, where $\ft \subeq \g$ is a 
compactly embedded Cartan subalgebra. 
Then $\Delta(\g_\C,\fa) = \Delta(\g_\C,\ft) \subeq \fa^* \cong i \ft^*$, and 
the root spaces for $\fa$ and $\ft$ in $\g_\C$ coincide. 
If $\fk \supeq \ft$ is the unique maximal compactly embedded subalgebra 
containing $\ft$, then $\fq_\fp = i\fk$ is a maximal 
hyperbolic Lie triple system in $\fq=i\fg$ 
(\cite[Cor.~III.8]{KN96}). This implies that a 
root $\alpha \in \Delta(\g_\C,\fa)$ is compact in the sense 
of Definition~\ref{def:3.1} if and only if 
$\g_\C^\alpha \subeq \fk_\C$.
\end{ex}

\subsection{Classification of invariant cones} 
\mlabel{subsec:cone-classif}

For a non-compactly causal symmetric Lie algebra $(\g,\tau,C)$, 
there may be many different pointed generating $\Inn_\g(\fh)$-invariant 
cones, but there is a rather explicit classification of all these 
cones in terms of intersections with $\fa$.
For the following classification theorem, we refer to 
\cite[Thm.~VI.6]{KN96},   \cite[Sec.~4.4]{HO97} and \cite[Sec.~7]{Ol91}.
If $\fa\subset \fq_\fp$ is maximal abelian and $C$ is a hyperbolic cone, then we set $C_\fa = C\cap \fa$.

\begin{thm} \mlabel{thm:3.6} 
{\rm(Classification of invariant cones)} 
Let $(\g,\tau, C)$ be a non-compactly causal symmetric 
Lie algebra and $\fa \subeq \fq_\fp$ be maximal abelian. 
Then there exists a $\fp$-adapted positive system 
$\Delta^+ \subeq \Delta(\g,\fa)$ such that 
\begin{equation}
  \label{eq:a-cones}
 C^{\rm min}_\fa \subeq C \cap \fa \subeq C^{\rm max}_\fa,
\end{equation}
where 
\begin{equation}
  \label{eq:minmaxconesina}
C^{\rm min}_\fa = \cone(\{\alpha^\vee \: \alpha \in \Delta_p^+\}) 
\quad \mbox{ and } \quad C^{\rm max}_\fa = (\Delta_p^+)^\star
= \{ X \in \fa \: (\forall \alpha \in \Delta_p^+)\ \alpha(X) \geq 0\}.
\end{equation}
Moreover, $C\cap \fa$ is $\cW _k$-invariant and 
$C^\circ \cap \fz({\fq_\fp}) \not=\eset$. 

Conversely, for every pointed generating $\cW _k$-invariant closed convex cone 
$C_\fa \subeq \fa$ satisfying 
\[ C^{\rm min}_\fa \subeq C_\fa\subeq C^{\rm max}_\fa, \]
there exists a uniquely determined $\Inn_\g(\fh)$-invariant 
pointed generating closed convex cone $C_\fq \subeq \fq$ 
with $C_\fq \cap \fa = C_\fa$. 
\end{thm}

In the following we write
  \[C_\fq^{\rm min} \subeq C_\fq^{\rm max} \]
  for the $\Inn_\g(\fh)$-invariant pointed generating
  cones in $\fq$ which are uniquely determined by the intersection 
  $C_\fq^{\rm min} \cap \fa = C_\fa^{\rm min}$ and 
  $C_\fq^{\rm max} \cap \fa = C_\fa^{\rm max}$.

\begin{rem}
If $(\g,\tau)$ is semisimple {\bf without Riemannian factors}, 
then the classification theorem (Theorem~\ref{thm:3.6}) implies for every 
Euler element $h_c \in \fz(\fq_\fp)$ and any $\fp$-adapted positive system 
$\Delta^+$ with 
\[ \Delta_p^+ = \{ \alpha \in \Delta \: \alpha(h_c) = 1\}  \] 
the existence of pointed generating invariant cones 
\[ C^{\rm min}_\fq(h_c)  \subeq  C^{\rm max}_\fq(h_c) \subeq \fq \] 
with 
\[ C^{\rm min}_\fq(h_c)  \cap \fa = C_\fa^{\rm min} 
\quad \mbox{ and } \quad 
C^{\rm max}_\fq(h_c)  \cap \fa = C_\fa^{\rm max}.\] 
Both are adapted to the decomposition of $(\g,\tau)$ 
into irreducible summands, resp., the decomposition of $\g$ into simple ideals. 
\end{rem}

We record the following interesting consequence which in some 
cases implies that the cone~$C$ is adapted to the decomposition
into irreducible subspaces. 

\begin{cor} \mlabel{cor:adapt}
Let $(\g,\tau)$ be a semisimple ncc symmetric Lie algebra 
and $(\g,\tau) = \oplus_{j = 1}^n (\g_j, \tau_j)$ its decomposition into 
irreducible summands. 
If the minimal and the maximal $\Inn_{\g_j}(\fh_j)$-invariant cones 
$C_{\fq_j}^{\rm min}$ and $C_{\fq_j}^{\rm max}$ coincide, then every 
$\Inn_\g(\fh)$-invariant cone $C \subeq \fq$ is adapted to the decomposition into 
simple ideals.   
\end{cor}

\begin{prf} As $C_\fq^{\rm min} \subeq C \subeq C_\fq^{\rm max}$ 
 and $C_\fq^{\rm min} = C_\fq^{\rm max}$ by assumption, 
$C = C_\fq^{\rm min} = C_\fq^{\rm max}$, and these cones are adapted to the decomposition 
into irreducible summands. 
\end{prf}

The following proposition has the interesting consequence that 
the cones $C_\pm$ remain the same when we replace $C$ by the 
minimal/maximal cone $C^{\rm min}_\fq$ or $C^{\rm max}_\fq$ determined 
by~$C$ via~\eqref{eq:a-cones}.

\begin{prop} \mlabel{prop:reductiontominmax}
Let $(\g,\tau,C,h)$ be a reductive modular non-compactly causal 
 symmetric Lie algebra and choose a $\fp$-adapted positive 
system of roots in $\Delta(\g,\fa)$ with 
\[ C^{\rm min}_\fq \subeq C \subeq C^{\rm max}_\fq.\] 
Then 
\begin{equation}
  \label{eq:skew-equal}
   C \cap \fq^{-\tau_h} 
= C^{\rm min}_\fq \cap \fq^{-\tau_h} 
= C^{\rm max}_\fq \cap \fq^{-\tau_h}.
\end{equation}
\end{prop}

\begin{prf} (a) If $\g$ is a simple hermitian Lie algebra, 
$h \in \cE(\g)$ is an Euler element and 
$C_\fg \subeq C$ a pointed generating invariant cone invariant under 
$-\tau_h$, then there are uniquely determined 
minimal, resp., maximal pointed generating cones $C_\g^{\rm min}$, 
resp., $C_\g^{\rm max}$ which are $\Ad(G)$-invariant and 
$-\tau_h$-invariant, such that 
\[ C_\g^{\rm min} \subeq C_\g \subeq C_\g^{\rm max}.\] 
Then \cite[Lemma~3.2]{Oeh20} implies that 
\[ C_\g^{\rm min} \cap \g_{\pm 1}(h) = C_\g \cap \g_{\pm 1}(h)
= C_\g^{\rm max} \cap \g_{\pm 1}(h).\] 
  
\nin (b ) 
Now let $\g$ be a reductive quasihermitian Lie algebra 
with an Euler element $h$ and a pointed generating invariant cone 
$C_\fg \subeq C$ invariant under $-\tau_h$. 
Write $\g = \g_c \oplus \g_{nc}$, where 
$\g_c$ is the direct sum of the center and the compact simple ideals and 
$\g_{nc}$ is the direct sum of the hermitian simple ideals. 
Then there exist uniquely determined pointed generating minimal and maximal 
invariant cones $C^{\rm min}_{\g_{nc}}$ and $C^{\rm max}_{\g_{nc}}$ in $\g_{nc}$ 
with 
\[ C_{\g_{nc}}^{\rm min} \subeq C_\g \cap \g_{nc} \subeq C_{\g_{nc}}^{\rm max}.\] 
Both cones $C^{\rm min}_{\g_{nc}}$ and $C^{\rm max}_{\g_{nc}}$ are adapted 
to the decomposition of $\g_{nc}$ into  simple ideals and 
invariant under $-\tau$, so that (a) implies that 
\[ C_{\g_{nc}}^{\rm min} \cap \g_{\pm 1}(h) =  C_{\g_{nc}}^{\rm max} \cap \g_{\pm 1}(h).\] 
We further have 
\[ C_\g^{\rm min} = C_{\g_{nc}}^{\rm min} \quad \mbox{ and } \quad 
 C_\g^{\rm max} = \g_c \oplus C_{\g_{nc}}^{\rm min}.\] 
As $[h,\g] = [h,\g_{nc}] \subeq \g_{nc}$, we 
thus obtain 
\[ C_{\g}^{\rm min} \cap \g_{\pm 1}(h) = C_\g \cap \g_{\pm 1}(h) = 
  C_{\g}^{\rm max} \cap \g_{\pm 1}(h).\] 

\nin (c) Now we turn to the proof of the proposition. 
We consider the $c$-dual modular compactly causal symmetric Lie algebra
$(\g^c, \tau^c, i C, h)$ and use the Extension Theorem~\ref{thm:extend}  
to extend $-i C \subeq i\fq \subeq \g^c$ to a pointed generating 
invariant cone $C_{\g^c} \subeq \g^c$. As 
\[ C_{\fq}^{\rm min} = -i C_{\g^c}^{\rm min} \cap \fq \quad \mbox{ and } \quad 
 C_{\fq}^{\rm max} = -i C_{\g^c}^{\rm max} \cap \fq, \] 
we then have 
\[ C_{\fq}^{\rm min} \subeq C \subeq  C_{\fq}^{\rm max},\] 
and (b) implies that 
\[ C_{\fq}^{\rm min} \cap \g_{\pm 1}(h) = C  \cap \g_{\pm 1}(h) 
=  C_{\fq}^{\rm max}  \cap \g_{\pm 1}(h).\] 
As any $\Inn_\g(\fh)$-invariant closed convex 
cone $D \subeq \fq$ satisfies 
\[ D \cap \fq^{-\tau_h} = D_+ - D_- 
\quad \mbox{ with } \quad 
D_\pm := D \cap \g_{\pm 1}(h)\] 
by Proposition~\ref{prop:2.18}, the assertion follows. 
\end{prf}

\subsection{Strongly orthogonal roots} 
\mlabel{subsec:olaf}

Let $(\g,\tau)$ be an irreducible non-compactly  
causal symmetric Lie algebra and recall that this implies that $\g$ is simple 
(\cite[Rem.~3.1.9]{HO97}, Section~\ref{subsec:2.2}). 
We fix a causal Euler element $h_c \in \fq$ 
(cf.\ Remark~\ref{thm:3.6}) 
and the corresponding 
Cartan involution $\theta = \tau\tau_{h_c}$ 
(cf.~\cite[Thm.~4.4]{MNO22a}, 
Lemma~\ref{lem:thetaCom}); then 
\[ \g^{h_c} = \ker(\ad h_c)= \fh_\fk + \fq_\fp \quad \mbox{ and } \quad 
\fz(\g^{h_c}) \cap \fq_\fp = \R h_c.\] 
In particular, any maximal abelian subspace $\fa \subeq \fq_\fp$ contains $h_c$ 
and is also maximal abelian in $\fq$ and $\fp$ 
(Remark~\ref{thm:3.6}(a)).  
Let $\fc \subeq \g$ be a Cartan subalgebra containing $\fa$. Then 
$\fc_\fk := \fc \cap \fk \subeq \fh_\fk$ 
and $\fc$ is invariant under $\tau$ and $\theta$, 
which coincide on $\fc$. 

For the root decomposition of $\g_\C$ with respect to~$\fc_\C$, we then have 
\[ \Delta =  \Delta(\g_\C,\fc_\C) \subeq i \fc_\fk^* \oplus \fa^* \] 
and $h$ induces a 
$3$-grading of the root system 
\[ \Delta = \Delta_p^- \dot\cup \Delta_k \dot\cup \Delta_p^+ \] 
with 
\[ \Delta_k = \{ \alpha \in \Delta \: \alpha(h_c) = 0\} \quad 
\mbox{ and } \quad 
 \Delta_p^\pm = \{ \alpha \in \Delta \: \alpha(h_c) = \pm 1 \}\] 
(Remark~\ref{thm:3.6}(c)).  
The $\tau$-invariance of $\fc$ implies that $\tau$ acts on 
$\Delta$, and since $\tau(h_c) =- h_c$, we have 
\[ \tau(\Delta_k) = \Delta_k \quad \mbox{ and }\quad 
\tau(\Delta_p^+) = \Delta_p^-.\] 
Note that $\Delta_k = \Delta_k(\fg^h, \fc)$  is the root system of 
the subalgebra~$\g^h$.

According to \cite[Thm.~3.4]{Ol91}, there exists a maximal 
subset $\Gamma = \{ \gamma_1, \ldots, \gamma_r \} 
\subeq \Delta_p^+$ of strongly orthogonal roots, i.e., 
$\gamma_j \pm \gamma_k \not\in\Delta$ for $j \not=k$, such that 
$-\tau(\Gamma) = \Gamma$. Then $\Gamma = \Gamma_0 \dot\cup \Gamma_1$, 
where $\Gamma_0 := \{ \gamma \in \Gamma \: - \tau\gamma = \gamma\}$, and 
$-\tau$ acts without fixed points on $\Gamma_1$, so that this set has 
an even number of elements. 
For $r_0 := |\Gamma_0|$ and $r_1 := |\Gamma_1|/2$, we obviously 
have 
\[ r = r_0 + 2r_1 \quad \mbox{  and put } \quad 
s := r_0 + r_1.\] 
By \cite[Lemma~4.3]{Ol91}, 
\begin{equation}
  \label{eq:ranks}
\rk_\R(\fh) = s.  
\end{equation}

The set $\Gamma$ of strongly orthogonal roots specifies 
a subalgebra 
\begin{equation}
  \label{eq:subalgroot}
\fs_\C := \sum_{\gamma \in \Gamma} \g_\C^\gamma + \g_\C^{-\gamma} + \C \gamma^\vee 
\cong \fsl_2(\C)^r 
\end{equation}
of $\g_\C$ invariant under~$\tau$, for which 
$\g^c \cap \fs_\C \cong \fsl_2(\R)^r$. 
The involution 
$\tau$ leaves all ideals in $\fs_\C$ corresponding to roots in $\Gamma_0$ 
invariant and induces flip involutions on the ideals corresponding 
to $(-\tau)$-orbits in $\Gamma_1$. 
We then have  
\begin{equation}
  \label{eq:fs-decomp}
\fs  \cong \fsl_2(\R)^{r_0} \oplus \fsl_2(\C)^{r_1}
\quad \mbox{ and } \quad 
 \fs^{\tau}  
\cong \so_{1,1}(\R)^{r_0} \oplus \su_{1,1}(\C)^{r_1}.
\end{equation}

The Cayley transform $\kappa_{h_c} := e^{\frac{\pi i}{2} \ad h_c}$  
induces a complex structure on $\fh_\fp + i \fq_\fk$ for which $\tau$ 
acts as an antilinear involution. Therefore $\kappa_h$ 
maps $\fh_\fp$ bijectively to $i \fq_\fk$, and thus 
\eqref{eq:ranks} implies that the maximal abelian subspaces 
of $\fq_\fk$ are also of dimension~$s$. Note also that 
$\ad h_c$ defines a bijection $\fh_\fp \to \fq_\fk$. 

In view of \eqref{eq:ranks}, 
$\fs \cap \fq_\fk$ contains a maximal abelian subspace $\ft_\fq$ of 
$\fq_\fk$. With respect to the decomposition of $\fs$ in \eqref{eq:fs-decomp}, 
we may choose 
\[ \ft_\fq = \so_2(\R)^{r_0 + r_1} = \so_2(\R)^s.\]

From \cite[Prop.~5.2]{MNO22a} and the subsequent remark, 
we obtain by specialization to the case where 
$\fh$ contains an Euler element 
(which is equivalent to $\g^c$ being of tube type):

\begin{prop} \mlabel{prop:testing}
Let $(\g,\tau,C)$ be a simple ncc symmetric Lie algebra 
for which $\fh = \g^\tau$ contains an Euler element. Pick a 
causal Euler element $h_c \in C^\circ$ and 
$\ft_\fq \subeq \fq_\fk$ maximal abelian. 
Then $s := \dim \ft_\fq$ can be decomposed as $s = r_0 + r_1$, 
such that the 
Lie algebra $\fs_0$ generated by $h_c$ and $\ft_\fq$ 
is isomorphic to $\fsl_2(\R)^s$. 
It is $\tau$-invariant and $\fs_0^\tau \cong \so_{1,1}(\R)^s$.  
\end{prop}

\subsection{Embedding into spaces of complex type}
\mlabel{sect:Complex}

We now discuss  the causal embedding of a non-compactly causal symmetric 
space $G/H$ into the symmetric space $G_\C/G^c$ of complex type, 
where $G_\C$ is the universal complexification of $G$ and 
 $G^c $ is a Lie subgroup with Lie algebra $\fg^c = \fh \oplus i\fq$. 

Let $(\g,\tau,C)$ be a reductive non-compactly causal symmetric Lie algebra, 
$G$ a connected Lie group with Lie algebra $\g$ 
and $\theta$ a Cartan involution on $G$ satisfying  (GP) and (Eff). 
By Lemma~\ref{lem:2.3.1}(d) 
the universal complexification $\eta_G \: G \to G_\C$ is injective, so that 
we may consider $G$ as a subgroup of $G_\C$. We write $\sigma_G$ for the 
unique antiholomorphic automorphism of $G_\C$ with 
$\sigma_G \circ \eta_G = \eta_G$ and note that 
$G \cong (G_\C)^{\sigma_G}_e$. 

By the universal property of the universal complexification 
$\eta_G : G \to G_\C$, there exist a unique holomorphic 
involutions, also denoted $\theta$ and $\tau$, on $G_\C$ satisfying 
$\tau  \circ \eta_G = \eta_G \circ \tau$ and 
$\theta \circ \eta_G = \eta_G \circ \theta$, respectively.  
The anti-holomorphic extension of $\tau $ is
given by $\oline \tau =\sigma_G\circ \tau$. 
Let 
\[ G^c \subseteq (G_\C)^{\otau}\]
be an open subgroup satisfying $\sigma_G(G^c) = G^c$ and put 
\begin{equation}
  \label{eq:Hintersec}
 H := G \cap G^c.
\end{equation}
Then $H_e$ is the identity component of $G \cap G^c$ and
we obtain a $G$-equivariant embedding 
\begin{equation}
  \label{eq:embcplx1}
 M :=G/H\hookrightarrow G_\C/G^c
\end{equation}
of symmetric spaces. As $\sigma_G(G^c)= G^c$, 
the involution $\sigma_G$ defines an involution 
\[ \sigma_M : G_\C/G^c\to G_\C/G^c, \quad gG^c \mapsto \sigma_G(g)G^c,\] 
 and
\begin{equation}\label{eq:FixM}
M \cong (G_\C/G^c)^{\sigma_M}_{eG^c}
\end{equation}
is the connected component of the base point $eG^c$ 
in the fixed point space of $\sigma_M$.

Let $K_\C \subeq G_\C$ be an open $\sigma_G$-invariant 
subgroup of $(G_\C)^\theta$. 
Then $K = G^\theta = \eta_G^{-1}(K_\C)$ follows from 
$K \subeq K_\C \subeq (G_\C)^\theta$. 
As $K$ is connected but $K_\C$ need not be 
connected, the inclusion $K \into K_\C$ may {\bf not} be the universal complexification 
of $K$. However, the inclusion $K \into (K_\C)_e$ is the universal complexification 
of~$K$.
\begin{footnote}
{We want to keep some flexibility in choosing $K_\C$ and hence 
the complexification $M^r_\C$ because the crown domain, 
which is simply connected, can be realized in many ``complexifications''.}
\end{footnote}

For the symmetric space $M^r$, we  define the complexification 
by 
\[ M^r_\C := G_\C/K_\C  \] 
(see also Section~\ref{subsec:4.1a}). 
Then $M_\C^r$ is a complex symmetric space 
and the $G$-orbit of the base point in 
$M^r_\C$ is $G$-equivariantly diffeomorphic to $G/K$. 
We thus obtain an embedding $M^r \into M^r_\C$. 

\begin{ex} Consider the group 
$G := \GL_n(\R)/\Gamma$ with 
$\Gamma = \{\pm \1\}$ from Example~\ref{ex:gln}. 
It acts faithfully on the space $\Sym_n(\R)$ of symmetric matrices 
by $g.A = gAg^\top$. Then 
\[ M^r = G.\1 = \{ g g^\top \: g \in G \} \] 
is the space of positive definite matrices. The 
$G$-action on $\Sym_n(\R)$ extends naturally to a holomorphic action 
of $G_\C = \GL_n(\C)/\Gamma$ on $\Sym_n(\C)$, given by the same formula. 
The Cartan involution on $G$ is given by 
$\theta(g\Gamma) = (g^\top)^{-1}\Gamma$ on $G$.
The holomorphic extension to $G_\C$ is given by the same
formula whereas the anti-holomorphic extension is $\oline\theta (g\Gamma) = (g^*)^{-1}$ which
is a Cartan involution on $G_\C$ with $G_\C^{\oline\theta} = \UU_n(\C) \Gamma$. 

An element $g\Gamma \in G$ is $\theta$-fixed if and only if 
$g.\1 = gg^\top \in \Gamma$, which by the positive definiteness of $gg^\top$ 
implies $gg^\top = \1$, i.e., $g \in \OO_n(\R)$. However, over the 
complex numbers, 
$gg^\top = - \1$ is possible, 
as the matrix $g = i \1$ shows. 
Therefore 
\[(G_\C)^\theta = \OO_n(\C)\Gamma \cup i\OO_n(\C)\Gamma \] 
is not connected, but $(G_\C)^\theta \cap G= G^\theta = K$. 
\end{ex}

By Corollary~\ref{cor:ext-ncc}, 
there
exists a $G^c$ invariant cone $C_{\fg^c}\subset \fg^c$ 
such that 
$C_{\fg^c}\cap i \fq = i C$. 
The following lemma is now clear:  

\begin{lemma}\mlabel{lem:CausaEmb} The symmetric 
Lie algebra $(\g_\C, \otau, i C_{\fg^c})$ 
is non-compactly causal and 
\[ (\fg,\tau, C)\hookrightarrow (\fg_\C, \otau, iC_{\fg^c}) \] 
is an embedding of non-compactly causal symmetric Lie algebras.
If $h$ is an Euler element of $\fg$ contained
in $\fh$, then $h$ is also an Euler element in $\fg_\C$, 
so that we even obtain an embedding of modular non-compactly causal 
symmetric Lie algebras.
\end{lemma}

\begin{lemma} Let $S_{G^c}= G^c\exp(i C_{\fg^c})\subset G_\C$ be the complex Olshanskii semigroup corresponding to the
cone $C_{\fg^c}$ and let $S_H = H\exp C$ be the real Olshanski group. Then $S_{G^c}$ is 
invariant under the involutive antiautomorphism 
$g^* :=\otau (g)^{-1}$  and 
\[(S_{G^c}^{\sigma_G})_e\subset S_H = S_{G^c}\cap G \subset S_{G^c}^{\sigma_G}\]
and
\[(S_{G^c}^\circ)^{\sigma_G}_e\subset S_H^\circ = S_{G^c}^\circ\cap G \subset (S_{G^c}^\circ)^{\sigma_G}\]
\end{lemma}

\begin{prf} For $s=g\exp x\in S_{G^c}$ we have 
$s^* = g^{-1} \exp(\Ad(g)x) \in S_{G^c}$, so that $S_{G^c}$ is $*$-invariant. 

The inclusion $S_H \subeq G \cap S_{G^c}$ follows from $H \subeq G^c$. 
For the converse, let $s = g \exp x\in S_{G^c} \cap G$. Then 
\[ g \exp x = s = \sigma_G (s) = \sigma_G(g) \exp (\sigma_\fg (x))\]
where $\sigma_\fg$ is the conjugation on $\fg_\C$ with respect to $\fg$. By the uniqueness of the factors in
the polar decomposition of $s$ it follows that $g=\sigma_G(g) $ and $\sigma_\fg (x) = x$. As $(i\fg^c)\cap \fg =\fq$, 
it follows that $x\in C$. 
We conclude that $g = s\exp (-x) \in G\cap G^c = H$.
\end{prf}

\subsection{The domain $C^\pi$} 
\mlabel{subsec:3.3}

To discuss the connection between the wedge domains in $M$ and $M_\C$ we need to 
make all choices compatible. We have already fixed $\fa\subset \fq_\fp$ maximal abelian. Recall that 
$\fz_\g(\fa ) =\fz_{\fh_\fk}(\fa) \oplus \fa$. We extend
$i\fa$ to a Cartan subalgebra $\ft^c = \ft_\fh \oplus i\fa$ of $\fk^c = \fh_\fk 
\oplus i\fq_\fp$ and $\fg^c$, where 
$\ft_\fh = \ft\cap \fh_\fk$. 
Then $\fa^c = i\ft_\fh$ is maximal abelian in $\fg_{\C,\fp}^{-\otau} 
= i\fh_\fk \oplus \fq_\fp$
and $\fg_{\C, p} = (\fh_\fp \oplus i\fq_\fk) \oplus ( i\fh_\fk \oplus \fq_\fp )$.
We then have the root system
\[\Delta (\fg_\C ,\ft_\C) =\Delta (\fg_\C,\fa^c)\]
and then
\[\Delta (\fg ,\fa) = \{\alpha|_\fa \:  \alpha\in \Delta (\fg_\C ,\ft_\C) \text{ and } \alpha |_\fa \not= 0\}.\]

\begin{defn} \mlabel{def:3.5} (The domain $C^\pi$) As the function 
  \begin{equation}
    \label{eq:sa}
s_\fa : \fa \to \R, \quad 
s_\fa(x) := \max \{ |\alpha(x)|, 2|\beta(x)| \: \alpha \in \Delta_p^+, 
\beta \in \Delta_k \}
  \end{equation}
is invariant under $\cW_k$ and every $\Inn_\g(\fh)$-orbit in $C^\circ$ 
intersects $\fa$ in a $\cW_k$-orbit 
(Lemma~\ref{lem:3.2}), the function $s_\fa$ extends 
to a uniquely determined $\Inn_\g(\fh)$-invariant function 
\begin{equation}
  \label{eq:s}
 s \: C^\circ \to \R
\end{equation}
on the $\cW_k$-invariant open subset $C^\circ$ of~$\fq$. 
We define the $\Inn_\g(\fh)$-invariant open subset 
\begin{equation}
  \label{eq:cpi}
C^\pi := \{ x \in C^\circ \: s(x) < \pi\}.
\end{equation}
Note that, although the function $s_\fa$ is convex, this is not the 
case for $s$ if there exist non-compact roots, i.e., 
if $\tau$ is not a Cartan involution (Remark~\ref{rem:3.4}). 
In particular, $C^\pi$ is not convex if not every 
root is compact, but the intersection $C_{\rm max}^\pi \cap \fa$ is 
the interior of a convex polyhedron. 
\end{defn} 

We have $\Delta (\fg,\fa)_p = \Delta (\fg_\C,\ft_\C)_p|_\fa$ and
$\Delta (\fg,\fa)_k =\Delta (\fg_\C,\ft_\C)_k|_{\fa}\setminus \{0\}$. 
Thus, for $x\in \fa$, we have
that $s(x)$ can also be calculated as $s_{\fg_\C}(x)$ via 
the root decomposition of $\fg_\C$. It follows that
\begin{equation}\label{eq:CgFix}
(i C_{\fg^c}^\pi)^{\sigma} = C^\pi.
\end{equation}

\section{Polar maps, the crown and tube domains} 
\mlabel{sec:4}

In this section we first introduce 
a special classes of Euler elements $h_c$, the 
{\it causal Euler elements} for $(\g,\tau , C)$,
which are those contained in $C^\circ \cap \fq_\fp$.
The associated {\it causal Riemannian element} is $x_r := \frac{\pi}{2} h_c
\in C^\pi \cap \fa$. It plays an important role in connecting the tube domain 
$\cT_M$ of the causal symmetric space $M$ 
with the crown domain of a the Riemannian symmetric space $M^r := G/K$ 
(\cite{AG90, GK02}).
We shall also use them in our analysis of the different types 
of wedge domains in $M = G/H$.

Concretely, we show in Theorem~\ref{thm:4.7} that,
for $m := \Exp_{eH}(i x_r) \in \cT_M$, the orbit
$G.m$ is isomorphic to $M^r$ and that we thus obtain an identification
of the tube domain $\cT_M$ with the crown domain
\[ \cT_{M^r} = G.\Exp_m(i\Omega_\fp), 
\quad \mbox{ where } \quad 
\Omega_\fp 
= \Big\{ x \in \fp \: \rho(\ad x)  < \frac{\pi}{2}\Big\}.\]
This result is prepared in several steps.
In Subsection~\ref{subsec:4.1a}, we first show that
the polar decomposition of the crown domain
$\cT_{M^r}$ of the Riemannian symmetric space $M^r$
is a diffeomorphism (cf.\ \cite{AG90}). 
We also obtain a new
characterization of real crown domains as a submanifold of the real tube domain
$\cT_{C_\fq} = \fh + C^\circ \subeq \g$ of 
the cone $C$ (Theorem~\ref{thm:crownchar-gen}):
For a causal Euler element $h_c \in C \cap \fq_\fp$,
the connected component of $h_c$ in the intersection 
$\cO_{h_c} \cap \cT_{C_\fq}$ is the domain 
\[ \cT_{M_H} = \Ad(H) e^{\ad \Omega_{\fq_\fk}} h_c, \quad \mbox{ where } \quad 
\Omega_{\fq_\fk} = \Big\{ x \in \fq_\fk \: \rho(\ad x) < \frac{\pi}{2}\Big\}.\]

\subsection{Causal Euler and Riemann elements} 
\mlabel{subsec:3.4} 

In this section we introduce {\it causal Riemann elements} 
as a tool to translate between crowns of Riemannian symmetric spaces and 
non-compactly causal spaces in their boundary.

\begin{defn} \mlabel{def:causal-euler}
An Euler element $h_c \in \fa\subeq \fq_\fp$ 
is called a {\it causal Euler element for $(\g,\tau,C)$} if 
\[ h_c \in C^\circ   \quad \mbox{ and } \quad 
\{ \alpha \in \Delta \: \alpha(h_c) = 0\} = \Delta_k. \]
As no non-compact root vanishes in $C^\circ$ 
(Theorem~\ref{thm:3.6}), the latter condition is equivalent 
to ${h_c \in \fz(\fq_\fp)}$. 
If $\Delta^+$ is a $\fp$-adapted positive system with 
$C \subeq (\Delta_p^+)^\star$, it follows from the Euler property of $h_c$ 
that 
\[ \Delta_p^+ = \{ \alpha \in \Delta \: \alpha(h_c) = 1\}.\] 
For more details on the classification of ncc symmetric spaces 
in terms of causal Euler elements, we refer to \cite[\S 4]{MNO22a}. 
\end{defn}

\begin{rem} If $h_c$ and $h_c' \in C^\circ$ are two causal Euler elements, 
then all roots vanish on $h_c - h_c'$, so that $h_c - h_c' \in 
\fa_\fz := \fz(\g) \cap \fa$. 
Conversely, if $h_c \in C^\circ$ is a causal Euler element and 
$z \in \fa_\fz$ with $h_c + z \in C^\circ$, then $h_c + z$ 
is a causal Euler element. 

If $\fa_\fz = \{0\}$, then causal Euler elements 
are uniquely determined by the cone~$C$, 
resp., the set $\Delta_p^+$ of positive non-compact 
roots with $C_\fa \subeq (\Delta_p^+)^\star$, see \cite[Thm. 3.1.5]{HO97} 
and \cite[Thm.~4.4]{MNO22a}. 
\end{rem}

\begin{defn} \mlabel{def:riemtype}
 An element $x_r \in C_\fa^\pi$ is called a {\it causal 
Riemann element} if $h_c := \frac{2}{\pi} x_r$ is a causal Euler element. 
This means that all compact roots vanish on $x_r$ and that 
the positive non-compact roots $\beta \in \Delta_p^+$ all satisfy 
$\beta(x_r) = \frac{\pi}{2}$. In particular, $x_r \in C_\fa^\pi$ 
(cf.\ Definition~\ref{def:3.5}). 
\end{defn}

\begin{lem} \mlabel{lem:4.8} {\rm(Causal Riemann elements)} 
Suppose that $(\g, \tau, C)$ is a reductive 
ncc symmetric Lie algebra with $\fz(\g) \subeq \fq$, that $\fa \subeq \fq_\fp$ 
is maximal abelian, 
and that $\Delta^+$ is a $\fp$-adapted positive system 
with $C \cap \fa \subeq (\Delta_p^+)^\star$. 
\begin{itemize}
\item[\rm(a)] Then $C_\fa^\pi = C^\pi \cap \fa$
  contains a causal Riemann element $x_r$.  \item[\rm(b)] For  $\tau_{h_c}
 = e^{\pi i\ad h_c}$, 
the automorphism $\theta := \tau \tau_{h_c}$ 
 is a Cartan involution  
\item[\rm(c)] If, in addition, $h \in \fh$ is an Euler element with 
$\tau_h(C) = - C$, then 
$z_r := \frac{1}{2}(x_r - \tau_h(x_r))$ 
is a causal Riemann element satisfying $\tau_h(z_r) = -z_r$. 
\end{itemize}
\end{lem}

\begin{prf} (a) 
Write $(\g,\tau) = (\g_0, \tau_0) \oplus (\g_1, \tau_1)$, 
where $(\g_0, \tau_0)$ is Riemannian and 
$\g_1$ is semisimple such that $\g_1$ is a sum of 
irreducible Cayley type spaces (Proposition~\ref{prop:decomp}). We then have 
\[ \fz(\fq_\fp) = \fz(\g) \oplus \fz(\fq_{1,\fp}), \] 
and the subspace generated by the 
pointed cone $C_\fa^{\rm min}$ contains $\fz(\fq_{1,\fp})$ 
(cf.~\cite[Thm.~VI.6]{KN96}). 
Any causal Riemann element $y_r \in C_\fa^{\rm min}$ is uniquely determined 
and contained in $C$. As 
\[ C_\fa^{\rm min} \subeq C_\fa := 
C \cap \fa \subeq C_\fa^{\rm max} = (\Delta_p^+)^\star \] 
and $C^\circ \cap \fz(\fq_\fp) \not=\eset$, we can add a suitable element 
$z \in \fz(\g)$ to obtain an element in $C_\fa^\circ$, hence 
in~$C_\fa^\pi$. 

To see that causal Riemann elements in $\g_1$ always exist, we use 
the decomposition of $(\g_1, \tau)$ into 
irreducible Cayley type space 
$(\g_j, \tau_j)$ (Proposition~\ref{prop:decomp}) 
Since $(\g_j, \tau_j)$ is irreducible, 
\cite[Prop.~V.6(ii)]{KN96} provides an Euler 
element $\tilde h_j \in \fz(\fq_\fp) \cap \g_j$ of $\g_j$, 
contained in the interior of $C_\fa^{\rm min}$, satisfying 
\[ \fz(\fq_\fp) \cap \g_j = \R \tilde h_j \quad \mbox{ and } \quad 
\fz_\fq(\tilde h_j) \cap \g_j = \fq_\fp \cap \g_j, \] 
and such that $\tau \tau_{\tilde h_j}$ is a Cartan involution 
(see \cite[Prop.~V.6(iii)]{KN96}). 
Then ${x_{r,j}} := \frac{\pi}{2} \tilde h_j$ is a causal Riemann element 
contained in $C_{\rm min}$. Since each simple ideal 
of $\g_1$ contains a causal Riemann element 
contained in $C_{\rm min}$, so does $\g_1$. 

\nin (b) follows from the construction in (a). 

\nin (c) Now we assume that $h \in \fh$ 
is an Euler element with $\tau_h(C) = - C$. 
Let $\theta$ be a compatible Cartan involution of $[\g,\g]$ 
(Lemma~\ref{lem:compcartan})  and extend it by $-\id$ on $\fz(\g)$. 
Then $\theta$ commutes with $\tau$ and $\tau_h$. 
We may then assume that $\fq_\fp = \fq^{-\theta} = \g^{-\tau,-\theta}$, 
and then $\fz(\fq_\fp)$ is also invariant under $\tau_h$. 
Now $\tau_h(C) = - C$ implies that, for  any 
causal Riemann element $y_r$ in $\fz(\fq_\fp)$,  
the element $-\tau_h(y_r)$ is also a causal Riemann  
element. The remaining assertions are clear. 
\end{prf}

 \begin{ex} A concrete example is $\g = \fsl_n(\R)$ with $n = p + q$, 
   \[ \fh = \so_{p,q}(\R), \qquad \tau(x) = I_{p,q} (-x^\top) I_{p,q}
     \quad \mbox{ and } \quad I_{p,q} = \diag(\1_p, - \1_q).\] 
For $p \not=q$, we then have $\cE(\g) \cap \fh = \eset$ and, for
$G := \PSL_n(\R) = \SL_n(\R)/\{\pm\1\}$, the space  
\[M = G/H \cong \{ g I_{p,q} g^\top \: g \in \SL_n(\R) \} \] 
can be identified with a space of quadratic forms 
of signature $(p,q)$ on $\R^n$. 
Here
\[ \fq = \{ x \in \fsl_n(\R) \: x^\top = I_{p,q} x I_{p,q} \} \] 
contains the one-dimensional subspace 
$\fq_\fp^{\fh_\fk} = \R h_c$  for 
\[ h_c := \frac{1}{p+q}\diag(q \1_p, - p \1_q)  \] 
as fixed points for the subgroup $\SO_p(\R) \times \SO_q(\R)$
with Lie algebra~$\fh_\fk$. Now $h_c \in \fq$ is a causal Euler element for the 
$\Ad(H)$-invariant cone $C$ it generates. 
\end{ex}

\subsection{The polar map of a crown domain} 
\mlabel{subsec:4.1a} 
We start this subsection by introducing the crown of a Riemannian symmetric
space $M^r = G/K$ which in our case is
allowed to be reductive. We
then state and prove Proposition \ref{prop:4.9} which for semisimple symmetric spaces
is due to Akhiezer and  Gindikin (\cite{AG90}). 
Here we derive those domains not
from the Riemannian symmetric space $M^r$ realized
inside the crown, but from the perspective of the
ncc symmetric space $M$. In fact, the ncc symmetric spaces can 
(up to coverings) 
all be realized as a $G$-orbit in the boundary of 
the crown (cf.~\cite{GK02}). 
 Suppose that $(G,\theta)$ satisfies (GP) and (Eff) 
from Subsection~\ref{subsec:globass} 
and recall that (GP) implies that  the Riemannian exponential function 
$\Exp \: \fp \to M^r$ is a diffeomorphism. 

We consider the 
{\it crown of the associated Riemannian symmetric space $M^r$}
(cf.\ Definition~\ref{def:cxplss}): 
\[ \cT_{M^r} = G.\Exp_{eK}(i\Omega_\fp) \subeq M^r_\C := G_\C/K_\C,
\quad \mbox{ where } \quad 
\Omega_\fp 
= \Big\{ x \in \fp \: \rho(\ad x)  < \frac{\pi}{2}\Big\}.\]

\begin{rem} \mlabel{rem:4.2.1} 
(a) Lemma~\ref{lem:2.3.1}(d) implies that 
the exponential function of $G_\C$ is injective on~$\fz(\g_\C)$. 

\nin (b) The product $\oline\theta := \theta \sigma_G$ of 
$\theta$ with the complex conjugation $\sigma_G$ 
of $G_\C$ with respect to $G$ 
is an antiholomorphic extension of $\theta$. As this involution 
preserves the subgroup 
$G_\C^\theta$, 
it induces an antiholomorphic involution, also denoted $\oline\theta$,  
on $M_\C^r = G_\C/K_\C$. 
\end{rem}

\begin{prop} \mlabel{prop:4.9} {\rm(\cite[Prop.~4]{AG90})} 
If the connected reductive group 
$G$ satisfies {\rm(GP)} and {\rm(Eff)}, then the polar map of the crown 
\[  \Phi \: G \times_K \Omega_\fp \to \cT_{M^r} \subeq M^r_\C, \quad 
[g,y] \mapsto g.\Exp_{eK}(iy) \] 
is a diffeomorphism. 
\end{prop}

\begin{prf} Lemma~\ref{lem:2.3}(b), applied to 
the complex conjugation $\sigma$ of $\g_\C$ with respect to $\g$ 
implies that all tangent maps of $\Phi$ are invertible. Therefore 
it remains to show that $\Phi$ is injective, i.e., that 
\begin{equation}
  \label{eq:injrel1}
g_1.\Exp_{eK}(iy_1) = g_2.\Exp_{eK}(iy_2)  \quad \Rarrow \quad 
g_2^{-1}g_1 \in K, \quad 
y_2 = \Ad(g_2^{-1}g_1)y_1.
\end{equation}

\nin {\bf Step 1:} First we show that $\Exp_{eK} \: i\Omega_\fp \to M^r_\C$ is injective. 
If $\Exp_{eK}(iy_1) = \Exp_{eK}(iy_2)$, then applying the quadratic representation 
implies that $\exp(2 iy_1) = \exp(2 iy_2)$. 
As $iy_1$ and $i y_2$ are both $\exp$-regular, 
\cite[Lemma~9.2.31]{HN12} implies that 
\[ [y_1, y_2]= 0 \quad \mbox{ and } \quad \exp(2iy_1 - 2 iy_2)= e.\]
We conclude that $e^{2 i\ad(y_1-y_2)} = \id_\g$, and since 
the spectral radius of $2i \ad(y_1 - y_2)$ is less than 
$2 \pi$, it follows that $\ad(y_1 - y_2) = 0$, i.e., $y_1 - y_2 \in \fz(\g)$. 
As the exponential function of $G_\C$ is injective on 
$\fz(\g_\C) \cap \fp_\C$ by Remark~\ref{rem:4.2.1}(a), 
we obtain $y_1 = y_2$. 

\nin {\bf Step 2:} $g.\Exp_{eK}(iy_1) = \Exp_{eK}(iy_2)$ with 
$g \in G$ and $y_1, y_2 \in \Omega_\fp$ implies $g \in K$. 
The antiholomorphic involution $\oline\theta$ on $M_\C$ preserves 
$\cT_{M^r}$ because it satisfies 
\[ \oline\theta(g.\Exp_m(iy)) = \theta(g).\Exp_m(iy) \quad \mbox{ for } \quad 
 g\in G, y \in \Omega_\fp.\] 
The relation 
$g.\Exp_{eK}(iy_1) = \Exp_{eK}(iy_2)$ implies that 
$g.\Exp_{eK}(iy_1)$ is a fixed point of $\oline\theta$, 
so that 
\[ g.\Exp_{eK}(iy_1) = \theta(g).\Exp_{eK}(iy_1) \] 
entails that $\theta(g)^{-1}g$ fixes $m_1 :=\Exp_{eK}(iy_1)$. 

We now write $g = k \exp z$ in terms of the polar decomposition of $G$ and 
obtain 
\[ \theta(g)^{-1} g = \exp(2z) \in G_{m_1}.\] 
Applying the quadratic representation of $M_\C$, we thus obtain 
\begin{equation}
  \label{eq:commrel1}
 \exp(2z) \exp(2i y_1) \exp(2z) = \exp(2i y_1), 
\end{equation}
which can be rewritten as 
\[ \exp(e^{2i \ad y_1} 2z) = \exp(-2z).\] 
Since $\ad z$ has real spectrum, so has $e^{2i \ad y_1} z$. 
Therefore the same arguments as in Step 1 above imply that 
\[ [z,e^{2i \ad y_1} z] = 0, \quad 
\exp(2 e^{2i \ad y_1} z + 2z) = e,\] 
and $e^{2i \ad y_1} z + z \in \fz(\g_\C)$. 
The imaginary part $\sin(2 \ad y_1) z$ 
of this element is central, but also contained in the commutator algebra, 
hence trivial. Therefore $\sin(2 \ad y_1) z= 0$, and since 
${\rho(2\ad y_1) < \pi}$, it follows that $[y_1, z] =0$. 
Now \eqref{eq:commrel1} leads to $\exp(4z) = e$, and further to 
$z = 0$, because the exponential function on $\fp$ is injective. 
This proves that $g = k \in K$. 

\nin {\bf Step 3:} From \eqref{eq:injrel1} we derive 
\[ g_2^{-1}g_1.\Exp_{eK}(iy_1) = \Exp_{eK}(iy_2),\] 
so that Step 2 shows that 
$k := g_2^{-1}g_1 \in K$. We thus obtain 
\[ \Exp_{eK}(iy_2) = k.\Exp_{eK}(iy_1) = \Exp_{eK}(i \Ad(k) y_1),\] 
and since $\Ad(k) y_1 \in \Omega_\fp$, we infer from Step 1 that 
$\Ad(k) y_1 = y_2$. This completes the proof. 
\end{prf}

\subsection{Properness of the $H$-action on the real tube} 
\mlabel{subsec:4.2} 

In this section we prove the following 
lemma which is needed in the proof of Theorem~\ref{thm:crownchar-gen}   
in Section~\ref{subsec:4.4}. 
Here we write
\begin{equation}
  \label{eq:tc}
 \cT_C := \fh + C^\circ \subeq \fh \oplus \fq = \g  
\end{equation}
for the open real tube domain in $\g$
corresponding to a generating convex cone $C \subeq \fq$.

\begin{lem} \mlabel{lem:closedorbit-gen}

  Let$H \subeq G^\tau$ be an open subgroup with $\Ad(H)C_\fq =~C_\fq$
  and let $D \subeq \cT_{C_\fq} = \fh + C_\fq^\circ$ be compact.
  Then 
$\Ad(H)D$ is closed in $\g$.   
Moreover, there exists a smooth $\Ad(H)$-invariant function 
\[ \psi \: \cT_{C_\fq}\to (0,\infty) \] 
 such that $z_n \to z_0 \in \partial \cT_{C_\fq}$ 
for $z_n \in \cT_{C_\fq}$ implies $\psi(z_n) \to \infty$.
\end{lem}

\begin{prf} As $H$ has only finitely many connected component by 
\cite[Thm.~IV.3.4]{Lo69}, 
it suffices to show that $\Ad(H_e)D$ is closed. We may therefore assume that 
$H$ is connected, hence equal to~$G^\tau_e$. 
As~$\g$ and the subalgebra $\fh = \g^\tau$ are reductive, 
we have $\det\big(\Ad_\fq(h)\big) = 1$ for $h \in H$. We put $C := C_\fq$. 
Then  the characteristic function 
\[ \phi \: C^\circ \to (0,\infty),\quad 
\phi(x) = \int_{C^\star} e^{-\alpha(x)}\, d\alpha \] 
is smooth and $\Ad(H)$-invariant with the property that 
$x_n \to x_0$ with $x_n \in C^\circ$ and $x_0 \in \partial C$ 
implies $\phi(x_n) \to \infty$ (\cite[Thm.~V.5.4]{Ne00}). 

We write elements $z \in \cT_C$ as $z = x + y$ with $x \in \fh$ and $y \in C^\circ$. 
Then $\psi(z) := \phi(y)$ is an $\Ad(H)$-invariant smooth function on 
$\cT_C$. In particular, it is bounded on the compact subset~$D$. 
Suppose that the sequence  
$\Ad(h_n)d_n \in \Ad(H)D$ with $h_n \in H, d_n \in D$, 
 converges to some element $w \in \g$. 
If $w \not\in \cT_C$, then $w = a + b$ with $a \in \fh$ and $b \in \partial C$, 
and then $\psi(d_n) = \psi(\Ad(h_n)d_n) \to \infty$ contradicts 
the boundedness of $\psi$ on $D$. This shows that 
$w \in \cT_C$. Write $d_n = a_n + b_n$ with $a_n \in \fh$ and $b_n \in C^\circ$. 
Then $\Ad(h_n)b_n$ converges to $b$.

Let $C_\fp := C \cap \fq_\fp$. Then \cite[Lemma~1.3]{Ne99}  
implies that the map 
\[ \Phi \: \fh_\fp \times C_\fp^\circ \to C_\fq^\circ, \quad 
(x,y) \mapsto e^{\ad x} y \] 
is a diffeomorphism. Replacing $G$ by $\Ad(G)$, 
we may assume that $Z(G) = \{e\}$, so that $G \cong \Ad(G)$. 
Writing $h_n = \exp(p_n) k_n$ with 
$p_n \in \fh_\fp$ and $k_n \in H_K$, the convergence of 
$\Ad(h_n) b_n =e^{\ad p_n} \Ad(k_n) b_n$ implies the existence of 
$p \in \fh_\fp$ with $p_n \to p$. Since $H_K$ is compact, 
we may further assume that the sequence $k_n$ converges to some 
$k \in H_K$. Then $h_n = \exp(p_n) k_n \to \exp(p)k =:h$  and 
thus $d_n = \Ad(h_n)^{-1}(\Ad(h_n)d_n) \to \Ad(h)^{-1}w \in D$ 
entails that $w \in \Ad(H)D$.   
\end{prf}

\subsection{Real crown domains and real tubes} 
\mlabel{subsec:4.4} 
 
 Crown domains of Riemannian symmetric spaces are very useful tools in 
 harmonic analysis and representation theory. 
In the following we relate those domains to our context. 
 
\begin{defn}
Let $h_c \in \fz(\fq_\fp) \cap C_\fq^\circ$ be a causal Euler element 
(Lemma~\ref{lem:4.8}) and suppose that 
 $(G^\tau)_e \subseteq H\subseteq (G^\tau)_eK^h =: H_{\rm max}$. 
Then\[\cO_{h_c}^\fq := \Ad(H)h_c  = e^{\ad \fh_\fp}h_c
  \cong H/H_K =: M_H \] 
realizes the non-compact Riemannian symmetric space $M_H$ 
associated to~$(\fh,\theta)$. Moreover, 
\[ \cO_{h_c} := \Ad(G)h_c  \] 
is a (parahermitian) symmetric space containing
$\cO_{h_c}^\fq$. 
\end{defn}

\begin{thm} \mlabel{thm:crownchar-gen}  
Let $(\g,\tau,C_\fq,h)$ is a modular ncc symmetric Lie algebra 
which is semisimple and for which all $\tau$-invariant ideals 
are non-Riemannian. Fix a causal Euler element 
$h_c \in \fz(\fq_\fp) \cap C_\fq^\circ$. 
Then the connected component of $h_c$ in the open subset 
$\cO_{h_c} \cap \cT_{C_\fq}$ of $\cO_{h_c}$ 
is the domain 
\[ \cT_{M_H} = \Ad(H) e^{\ad \Omega_{\fq_\fk}} h_c, \quad \mbox{ where } \quad 
\Omega_{\fq_\fk} = \Big\{ x \in \fq_\fk \: \rho(\ad x) < \frac{\pi}{2}\Big\}.\] 
\end{thm}

\begin{prf} If $(\g,\tau)$ is irreducible ncc, then there exist 
minimal and maximal $\Inn_\g(\fh)$-invariant pointed generating 
invariant cones $C^{\rm min}_\fq \subeq C_\fq \subeq C^{\rm max}_\fq$ 
(Theorem~\ref{thm:3.6} ). Our assumptions on $(\g,\tau)$ imply in particular  
that all simple $\tau$-invariant ideals are of this type, so that 
we can define pointed generating $\Inn_\g(\fh)$-invariant cones 
$C^{\rm min}_\fq, C^{\rm max}_\fq \subeq \fq$ in such a way that they 
are adapted to the decomposition into irreducible 
summands (Proposition~\ref{prop:decomp} and 
Subsection~\ref{subsec:cone-classif}). 
Then the classification of $\Inn_\g(\fh)$-invariant cones 
in $\fq$ (Theorem~\ref{thm:3.6}) implies that 
\begin{equation}
  \label{eq:sandwich}
C^{\rm min}_\fq \subeq C_\fq \subeq  C^{\rm max}_\fq.
\end{equation}

We will derive the theorem from the following three claims: 
\begin{itemize}
\item[\rm(a)] $\cT_{M_H}$ is connected and open in $\cO_{h_c}$. 
\item[\rm(b)] $\cT_{M_H} \subeq \cT_{C^{\rm min}_\fq}$, and 
\item[\rm(c)] $\cT_{M_H}$ is relatively closed in $\cO_{h_c}^\g \cap 
\cT_{C^{\rm max}_\fq}$.
\end{itemize}

In view of \eqref{eq:sandwich}, (a), (b) and (c) imply that 
$\cT_{M_H} \subeq \cT_{C_\fq}$, and that 
$\cT_{M_H}$ is relatively closed in $\cO_{h_c}^\g \cap \cT_{C_\fq}$. 
Therefore $\cT_{M_H}$ is the connected component of $h_c$ in 
the open subset $\cO_{h_c}^\g \cap \cT_{C_\fq}$. 

Since both sides of (a), (b) and (c) decompose according 
to the decomposition of $(\g,\tau)$ into irreducible 
summands (Proposition~\ref{prop:decomp}), 
it suffices to prove (a), (b)  and (c) in the irreducible case. 

\nin (a) To see that $\cT_{M_H}$ is connected, we use the polar decomposition 
$H = H_K \exp(\fh_\fp)$. Then $\Ad(H_K)\Omega_{\fq_\fk} = \Omega_{\fq_\fk}$ 
and $\Ad(H_K)h = \{h\}$ (cf.\  \cite[\S 3.5.2]{MNO22a}). 
This implies that 
\[ \cT_{M_H} 
= \Ad(H) e^{\ad \Omega_{\fq_\fk}} h 
= e^{\ad \fh_\fk} e^{\ad \Omega_{\fq_\fk}} \Ad(H_K) h 
= e^{\ad \fh_\fk} e^{\ad \Omega_{\fq_\fk}}h,\] 
which is obviously connected. 

Next we use Lemma~\ref{lem:2.3}(a) to 
see that the exponential map 
\[ \Exp \: \fq_\fk \to \cO_h, \quad x \mapsto e^{\ad x} h \] 
is regular in $x\in \Omega_{\fq_\fk}$ because $\rho(\ad x) < \pi/2 < \pi$. 
Now Lemma~\ref{lem:2.3}(b) implies that the map 
\[ \Phi \: H \times \Omega_{\fq_\fk} \to \cO_h, \quad 
(g,x) \mapsto \Ad(g) e^{\ad x} h \] 
is regular in $(g,x)$ because $\Spec(\ad x) \subeq (-\pi/2,\pi/2)i$ 
does not intersect $\big(\frac{\pi}{2} + \Z \pi\big)i$. 
This implies that the differential of $\Phi$ is surjective in each point 
of $H \times \Omega_{\fq_\fk}$; hence its image is open.

\nin (b) We first observe that 
both sides of (b) are $\Ad(H)$-invariant, and 
\[ \cT_{M_H} 
= \Ad(H) e^{\ad \Omega_{\fq_\fk}} h_c
= \Ad(H) e^{\ad \Omega_{\ft_\fq}} h_c\]  
for a maximal abelian subspace $\ft_\fq \subeq \fq_\fk$ and 
$\Omega_{\ft_\fq} := \Omega_{\fq_\fk} \cap \ft_\fq$. Therefore it suffices to verify 
$e^{\ad x} h_c \in \cT_{C^{\rm min}_\fq}$ for 
\[ x \in \Omega_{\ft_\fq} = \Big\{ y \in \ft_\fq \: 
\rho(\ad y)  < \frac{\pi}{2}\Big\}.\] 

From Proposition~\ref{prop:testing} we infer that 
the causal Euler element $h_c$ and $\ft_\fq$ generate a 
$\tau$-invariant subalgebra 
\[ \fs \cong 
\fsl_2(\R)^{r_0} \oplus \fsl_2(\C)^{r_1} \quad \mbox{ with  } \quad 
\fs^\tau \cong \so_{1,1}(\R)^{r_0} \oplus \su_{1,1}(\C)^{r_1}\]
in which $h_c$ is a causal Euler element, 
contained in the interior of the 
pointed generating cone 
\[ C_\fs := C^{\rm min}_\fq \cap \fs\subeq \fs_\fq := \fs^{-\tau}, \] 
which is invariant under $\Inn(\fs^\tau)$. 
Let $s := r_0 + r_1$. 

For elements of $\fsl_2(\R)$, we fix the notation 
\begin{equation}
  \label{eq:sl2a-sep}
h^0 = \frac{1}{2} \pmat{1 & 0 \\ 0 & -1},\quad 
e^0 =  \pmat{0 & 1 \\ 0 & 0}, \quad 
f^0 =  \pmat{0 & 0 \\ 1 & 0}. 
\end{equation}
Choosing a suitable isomorphism 
$\fs \to \fsl_2(\R)^{r_0} \oplus \fsl_2(\C)^{r_1}$, 
we obtain
\[ h_c =  \sum_{j = 1}^s h_j 
\quad \mbox{ and } \quad 
\ft_\fq = \so_2(\R)^s,\]
where we write $h_j$ for the Euler elements $h^0$ in the $j$th summand,  
For 
\[ x 
= \sum_{j = 1}^s x_j \frac{e_j -f_j}{2} \in \ft_\fq,\]  we then obtain with \eqref{eq:turn} 
in Appendix~\ref{subsec:sl2b}
\begin{equation}
  \label{eq:40b}
 p_\fq(e^{\ad x} h) 
= \cosh(\ad x)(h) 
= \sum_{j=1}^s \cos(x_j) h_j.
\end{equation} 

As the cone $C_\fs$ is adapted to the decomposition of $\fs$ into simple 
ideals by Corollary~\ref{cor:adapt}, 
this element is contained in the open cone 
$C_\fs^\circ$ if $|x_j| < \pi/2$ for each~$j$.
Since $\rho(\ad x) < \frac{\pi}{2}$ for $x \in \Omega_{\ft_\fq}$ 
and $i x_j \in \Spec(\ad x)$, this is the case. 
This proves~(b).

\nin (c) We have to show that 
$\cT_{M_H}$ is relatively closed in $\cO_{h_c}\cap \cT_{C^{\rm max}_\fq}$. 
First we observe that, 
by Lemma~\ref{lem:closedorbit-gen}, if $D \subeq \Omega_{\ft_\fq}$ is compact, 
then $\Ad(H) e^{\ad D} h_c$ is closed in $\g$, 
and by (b) it is contained in $\cT_{C^{\rm min}_\fq}\subeq \cT_{C^{\rm max}_\fq}$. 

Now suppose that the sequence 
$\Ad(g_n) e^{\ad x_n} h_c$, 
$g_n \in H$, $x_n \in \Omega_{\ft_\fq}$, converges to some 
element in $\cT_{C^{\rm max}_\fq}$ which is not contained in $\cT_{M_H}$. 
As we may assume that the bounded sequence $x_n \in \Omega_{\ft_\fq}$ converges 
in $\ft_\fq$, it converges by the preceding paragraph to a boundary point 
$y \in \partial \Omega_{\ft_\fq}$. Writing 
$y = \sum_{j = 1}^s \frac{y_j}{2}(e_j - f_j)$, 
we claim that there exists a $j$ with $|y_j| = \frac{\pi}{2}$. 
As 
\[ \rho(\ad y\res_{\fs}) = \max \{ |y_j| \: j =1,\ldots,s \},\] 
this follows from $\rho(\ad y\res_{\fs}) = \rho(\ad y)$ 
(\cite[Lemma~5.1]{MNO22a}).  
Now \eqref{eq:40b} shows that 
\[ e^{\ad y} h_c \in 
\partial \cT_{C_\fs} \subeq \partial \cT_{C_\fq^{\rm max}}. \]
For the $H$-invariant function 
$\psi$ from Lemma~\ref{lem:closedorbit-gen}, this leads to 
\[ \psi(\Ad(g_n) e^{\ad x_n} h_c) =\psi(e^{\ad x_n} h_c) \to \infty, \]
contradicting the convergences of the sequence 
$\Ad(g_n) e^{\ad x_n} h_c$ in $\cT_{C_\fq^{\rm max}}$. 
\end{prf}

\begin{rem} In the special case where 
$(\g,\tau)$ is Riemannian,  $\fq = \fq_\fp$ and 
$\fh = \fh_\fk$, the real crown domain $\cT_{M_H}$ reduces to a point. 
Hence there is no interesting analog of the preceding theorem 
in the Riemannian case, and we therefore assume that $(\g,\tau)$ contains 
no $\tau$-invariant Riemannian ideals.
\end{rem}

\section{The complex tube domain $\cT_M$} 
\mlabel{sec:5}

In this section we introduce the tube domain $\cT_M$, associated 
to a non-compactly causal symmetric space $M = G/H$, 
specified by the causal symmetric Lie algebra $(\g,\tau,C)$ 
with 
\begin{equation}
  \label{eq:z-ass-tau}
\Ad(H)C = C \quad \mbox{ and } \quad  \fz(\g) \subeq \fq, 
\end{equation}
and a connected symmetric Lie group $(G,\tau)$ 
satisfying (GP) and (Eff). 
Then $\tau$ induces a holomorphic involution on the universal 
complexification $G_\C$, also denoted~$\tau$. 

\nin {\bf Assumptions:} 
We fix an open $\sigma_G$-invariant subgroup $H_\C \subeq (G_\C)^\tau$ 
with $H_\C \cap G = H$ (here we use that $G \subeq G_\C$ by 
Lemma~\ref{lem:2.3.1}(d)) and 
define a {\it complexification of $M$} by 
\[ M = G/H  \into~M_\C := G_\C/H_\C.\]
We show in Theorem~\ref{thm:4.7}
  that, for $C = C_\fq^{\rm max}$, the tube domain
  $\cT_M$ is diffeomorphic to the crown of the Riemannian
   symmetric space $G/K$. 
\subsection{Algebraic preliminaries}

From Lemma~\ref{lem:4.8}
we know that there exists a causal 
Riemann element $x_r \in C_\fa^\pi$. 
Then $h_c := \frac{2}{\pi} x_r$ is a causal Euler element in $\g$ and 
$\tau_{h_c}$ 
is an involution. As $\tau \zeta_{x_r} \tau = \zeta_{x_r}^{-1}$ 
and $\zeta_{x_r}^2 = \tau_{h_c}$, 
we obtain the Cartan involution (see Lemma~\ref{lem:4.8}(b)) 
\begin{equation}
  \label{eq:thetafromtau}
\theta := 
\tau \tau_{h_c} = \tau \zeta_{x_r}^2 = \zeta_{x_r}^{-1} \tau \zeta_{x_r}, 
\end{equation}
so that the complex linear involutions $\theta$ and $\tau$ on $\g_\C$ 
are conjugate under $\Inn(\g_\C)$. 

\begin{lem} \mlabel{lem:eff} 
The subalgebra $\fh = \g^\tau$ contains no non-zero ideal of $\g$  
if and only if $\fk = \g^\theta$ contains no non-zero ideal of  $\g$, 
which follows from {\rm(Eff)}. 
\end{lem}

\begin{prf} Suppose that $\fj \trile\g$ is an ideal. If 
\begin{equation}
  \label{eq:jrel1}
\fj \subeq \g^\theta  \subeq \g^\theta_\C = \zeta_{x_r}^{-1}(\g_\C^\tau), 
\end{equation}
then 
$\fj_\C = \zeta_{x_r}(\fj_\C) \subeq \g_\C^\tau$ 
implies $\fj \subeq \g^\tau = \fh$. 
Conversely, $\fj \subeq \fh$ leads to 
\[\fj_\C = \zeta_{x_r}^{-1}(\fj_\C) \subeq \zeta_{x_r}^{-1}(\g_\C^\tau) = \g^\theta_\C,
\] 
and thus $\fj \subeq \g^\theta$. 
This shows that $\g^\theta$ contains no non-zero ideal of $\g$ 
is and only if $\g^\tau$ has this property.
\end{prf}

\begin{lem} \mlabel{lem:qbrackh} If {\rm(GP)} and {\rm(Eff)} are satisfied, 
then 
\begin{itemize}
\item[\rm(a)] $\fh = [\fq,\fq]$. 
\item[\rm(b)] If $\fh$ contains no non-zero ideal of $\g$, 
then $G$ acts effectively on $G/H$.
\end{itemize}
\end{lem}

\begin{prf} (a) From $\fz(\g) \subeq \fq$, (Eff), and the preceding lemma, 
we derive that $\fh$ contains no non-zero ideal of~$\g$. 
This implies the assertion because the ideal $[\fq,\fq]$ of 
$\fh$ is complemented by an ideal $\fj \trile \fh$ centralizing $\fq$, 
hence by an ideal of $\g$. Therefore $\fj$ is trivial, and this entails 
$\fh = [\fq,\fq]$. 

\nin (b) Let $N \trile G$ be the effectivity kernel of the 
$G$-action on $G/H$, i.e., the largest normal subgroup contained in~$H$. 
Its Lie algebra $\fn$ is
an ideal of $\g$ contained in $\fh$ and hence trivial. It follows that $N$ 
is discrete, hence central in $G$ because $G$ is connected. 
From Lemma~\ref{lem:2.3.1} we know that 
$\exp \: \fz(\g) \to Z(G)$ is a diffeomorphism, so that $\fz(\g) \subeq \fq$ 
further entails that $Z(G) \cap H = \{e\}$. This shows that $N = \{e\}$,
i.e., that $G$ acts effectively on $G/H$.
\end{prf}

\begin{rem} Let $G_1$ be a simple group and $\tau_1$ an involution commuting with the Cartan involution~$\theta_1$.
Let $H_1= (G_1^{\tau_1})_e$ and $K_1= G^{\theta_1}$. If the center
of $G_1$ is trivial, then 
(GP) and (Eff) are satisfied. Let $V$ b a finite dimensional real vector spaces and 
$G=G_1\times V$. Extend $\tau_1$ to an involution $\tau $ on $G$ by $\tau ((a,v)) =(\tau_1(a), v)$ and extend $\theta_1$ to
$G$ by $\theta (a,v) = (\theta_1(a),-v)$. Then $K= G^\theta = K_1\times \{0\}$ and $G/K
=G_1/K_1\times V$, so (GP) and (Eff) are still satisfied. 
But with $H = G^\tau_e = H_1\times V$ we have
$G/H \cong G_1/H_1$ and $G$ does not act effectively on~$G/H$.
\end{rem}

\subsection{The complex crown domain} 
\mlabel{subsec:5.1} 

Recall the set $C^{\pi}$ from Section~\ref{subsec:3.3}. 
In the following theorem we see how the causal Riemann elements 
can be used to connect between the tube domain and the crown domain.
In this subsection we  do not need the assumption that
  $(\g,\tau,C)$ is modular, i.e., that
  $\fh$ contains an Euler element.

\begin{thm} 
\mlabel{thm:4.7} {\rm(The tube domain as a crown domain)}
Suppose that $(\g,\tau, C)$ is a reductive ncc symmetric Lie algebra and
that {\rm(GP)} and {\rm(Eff)} are satisfied   for a corresponding connected
Lie group~$G$. 
If $x_r \in C_\fa^\pi$ is a causal Riemann element and 
\[ m := \Exp_{eH_\C}(ix_r) = \exp(ix_r)H_\C,\]  
then the following assertions hold: 
\begin{itemize}
\item[\rm(a)]  $M^r := G.m  \subeq M_\C$ is a Riemannian symmetric space. 
\item[\rm(b)] If $C = C_\fq^{\rm max}$ holds for a $\fp$-adapted positive system, 
then the crown domain $\cT_{M^r}$ coincides with the tube domain 
\[ \cT_M := G.\Exp_{eH}(i C^\pi) = G.\Exp_{eH}(i C^\pi_\fa) \subeq M_\C \] 
associated to $(\g,\tau,C)$ and~$M$.
\item[\rm(c)] $\cT_M$ is open in $M_\C$.
\end{itemize}
\end{thm}

\begin{prf} (a) In $G_\C$, we have $\zeta_{x_r}(G_{\C,m}) = H_\C 
\subeq (G_\C)^\tau$, 
where $\zeta_{x_r}$ denotes conjugation with $\exp(-ix_r)$ on~$G_\C$. 
This implies 
\begin{equation}
  \label{eq:conj-stab}
G_{\C,m} 
= \zeta_{x_r}^{-1}(H_\C) 
\subeq  (G_\C)^{\zeta_{x_r}^{-1}\tau \zeta_{x_r}}  
= (G_\C)^{\tau \zeta_{x_r}^2}  
= (G_\C)^{\tau\tau_{h_c}} = (G_\C)^\theta,
\end{equation}
where $\theta = \tau \tau_{h_c}$ 
is the complex linear extension of a
Cartan involution of $\g$ with $\fz(\g) \subeq \g^{-\theta}$ 
(Lemma~\ref{lem:4.8}(b)). 
As all such Cartan involutions of $\g$ are conjugate 
under inner automorphisms, (GP) implies that 
$G^\theta$ is connected, so that $G^\theta = G_m$. 
We conclude that $M^r = G.m \cong G/G_m = G/G^\theta$. 

\nin (b) Next we observe that 
\[ \g^\theta = \g^{\tau\tau_{h_c}} 
= \fh_0({x_r})  \oplus \fq^{-\tau_{h_c}} 
= \fh^0 
\oplus (\1 + \tau) \bigoplus_{\alpha \in \Delta_k} \g^\alpha 
\oplus (\1 - \tau) \bigoplus_{\alpha \in \Delta_p} \g^\alpha\] 
and 
\[\g^{-\theta} =  \g^{-\tau\tau_{h_c}}  
= \fa \oplus (\1 + \tau) \bigoplus_{\alpha \in \Delta_p} \g^\alpha 
\oplus (\1 - \tau) \bigoplus_{\alpha \in \Delta_k} \g^\alpha.\] 
In particular, $\fq$ and $\g^{-\theta}$ share the same 
maximal abelian subspace $\fa$.  
From Lemma~\ref{lem:2.7}(g), we further derive that 
\[ \exp(-ix_r).T_m(M^r) 
= \cos(\ad {x_r})\fq + i [{x_r},\fh]
= \fq^{\tau_{h_c}} + i \fq^{-\tau_{h_c}} \] 
is a real form of $\fq_\C$, so that $M^r$ is totally real in $M_\C$. 

For $y \in \fa$, we have 
\begin{equation}
  \label{eq:exp-shift}
\Exp_m(iy) = \Exp_{eH}(i({x_r} + y)).
\end{equation}
If $C$ is maximal in the sense that 
\[ C \cap \fa = C_\fa^{\rm max} = (\Delta_p^+)^\star, \] 
(cf.~Theorem~\ref{thm:3.6}), we obtain 
\begin{align*}
 C_\fa^\pi - {x_r} 
&= \{ y \in \fa  \: {x_r} + y \in C^\circ, s_\fa({x_r}+y) < \pi \} \\
&= \{ y \in \fa \: (\forall \alpha \in \Delta_p^+) \ 
0 < \alpha({x_r} + y) < \pi, (\forall \beta \in \Delta_k)\ 
2|\beta({x_r}+y)| < \pi \} \\
&= \{ y \in \fa \: (\forall \alpha \in \Delta_p^+) \ 
-\frac{\pi}{2} < \alpha(y) < \frac{\pi}{2}, (\forall \beta \in \Delta_k)\ 
2|\beta(y)| < \pi \} \\
&= \Big\{ y \in \fa \: (\forall \alpha \in \Delta)\ |\alpha(y)| < \frac{\pi}{2}\Big\} =: \Omega_\fa.
\end{align*}
Therefore the crown domain of the Riemannian symmetric space $M^r = G.m$ 
in $M_\C$ coincides with 
\[ G.\Exp_m(i \Omega_\fa) 
= G.\Exp_{eH}(i ({x_r}+ \Omega_\fa)) 
= G.\Exp_{eH}(i C_\fa^\pi) = \cT_M.\]

\nin (c) If $C$ is not maximal, then there 
exists a $\fp$-adapted positive system 
with $C \subeq C_\fq^{\rm max}$ and 
$C_\fa := C \cap \fa \subeq C_\fa^{\rm max}$. Then 
$\Omega_\fa^C := C_\fa^\pi - x_r \subeq \Omega_\fa$ 
is an open $\cW_k$-invariant subset. We thus obtain as in~(b) 
\[ \cT_M 
= G.\Exp_m(i\Omega_\fa^C) 
= G.\Exp_m(i\Ad(K)\Omega_\fa^C).\] 
To see that $\cT_M$ is open, in view of Proposition~\ref{prop:4.9}, 
it remains to show that 
$\Omega_\fp^C = \Ad(K)\Omega_\fa^C$ 
is an open subset. To this end, we first note that 
$\Omega_\fa' := \cW\Omega_\fa^C$ is a finite union of open subsets of $\fa$, 
hence open in $\fa$, and also $\cW$-invariant. 
This implies that 
$\Omega_\fp^C = \Ad(K)\Omega_\fa^C = \Ad(K) \Omega_\fa'$ 
is open in~$\fp$ (\cite[Lemma~1.4]{Ne99}).
\end{prf}

\begin{prop}
For every element $p \in \Exp_{eH}(iC^\pi)$,  the stabilizer Lie algebra 
$\g_p$ is compactly embedded in $\g$. 
\end{prop}

\begin{prf}
From Lemma~\ref{lem:2.7}(f), we obtain 
for $p = \Exp_{eH}(ix)$, $x \in C^\pi$, the stabilizer Lie algebra 
\[ \g_p
= \fh^{\sigma_x} \oplus \fq^{-\sigma_x}
= \fh_0(x)  \oplus \fq^{-\sigma_x}
\quad \mbox{ with } \quad 
\zeta_x(\g_p) = \fh_0(x) \oplus i\fh^{-\sigma_x}.\]  
This Lie algebra 
is compactly embedded because 
\begin{itemize}
\item $\fh_0(x) = \L(H_x)$ is compactly embedded since $x \in C^\circ$ 
and $\Ad(H)C = C$ (Lemma~\ref{lem:qbrackh}). 
\item $\fh_x = [\fq_x,\fq_x]$ satisfies $[\fh_0(x), \fh_x] \subeq \fh_x$, and 
\item $i\fh^{-\sigma_x} = [x, i\fq^{-\sigma_x}] \subeq \fh_x$, and $\fh_x$ is compactly 
embedded. \qedhere
\end{itemize}
\end{prf}

We conclude from the preceding proposition 
that the restrictions defining 
the open subset $C^\pi$ ensure that all stabilizer Lie algebras 
of points in $\Exp_{eH}(iC^\pi)$ are compactly embedded. 

For the adjoint group, this condition is equivalent to the compactness 
of the stabilizer group $G_p$ because $\Ad(G)$ is an algebraic group 
and any algebraic subgroup has at most finitely many connected components. 

If, conversely, $x \in C^\circ$ is an element for which there exists a compact 
root $\alpha$ with $\alpha(x) = \frac{\pi}{2}$, then 
$\fq^{-\sigma_x} \supeq (\1 -\tau)\g^\alpha$ consists of hyperbolic elements, 
so that $\g_p$ is not compactly embedded. 
If there exists a non-compact 
root $\alpha$ with $\alpha(x) = \pi$, then 
$\fh^{\sigma_x} \supeq (\1 +\tau)\g^\alpha$ contains hyperbolic elements, 
so that $\g_p$ is also not compactly embedded. 

In this sense the subset $C^\pi \subeq C^\circ$ is ``maximal'': One 
cannot enlarge the polyhedron $C^\pi_\fa$ within the open cone $C^\circ\cap \fa$, 
 without picking up points with non-compact stabilizers.

 \section{The connection between the wedge domains}
 \mlabel{sec:6}

The main result in the section is Theorem \ref{thm:6.1}, 
asserting that 
\[W_M(h) = W_M^{\rm KMS}(h) = \kappa_h^{-1}((\cT_M)^{\oline\tau_h}) . \]
The positivity domain $W_M^+(h)$ will be studied in
  the following section.

\subsection{The tube domain $\cT_M$}
\mlabel{subsec:5.2} 

In this subsection we consider the wedge domains as fixed point sets of involutions acting on related bigger domains. 
We start with the tube domain $\cT_M$.
On the cone $C^\circ$ we consider the function
$s \: C^\circ \to (0,\infty)$ from \eqref{eq:sa}.  Recall that 
\[ \cT_M = G.\Exp_{eH}(i C^\pi)\subeq M_\C, \qquad \text{where}\quad
C^\pi := \{ x \in C^\circ \: s(x) < \pi \}.\] 
We also consider the involution $\oline\tau_h =\sigma\tau_h : G_\C\to G_\C$, where
$\sigma$ is the conjugation on $G_\C$ with respect to $G$. 
\begin{thm} \mlabel{thm:4.5}
If {\rm(GP)} and {\rm(Eff)} are satisfied,  then 
\[ (\cT_M)^{\oline\tau_h} = (G^h)_e.\Exp_{eH}(i (C^\pi)^{-\tau_h}).
\]
In particular $ (\cT_M)^{\oline\tau_h}$ is connected.
\end{thm}

\begin{prf} {\bf Case 1:} We assume first that  
$C = C^{\rm max}_\fq$ for a $\fp$-adapted positive system. 
We use Lemma~\ref{lem:4.8} to find 
a causal Riemann element $x_r \in C_\fa^\pi$ with 
$\tau_h(x_r) = - x_r$. 
As $\cT_M = \cT_{M^r}$ by Theorem~\ref{thm:4.7}, 
we first identify the fixed points 
of $\oline\tau_h$ in the crown domain 
\[ \cT_{M^r} = G.\Exp_m(i\Omega_\fp) \quad \mbox{ for } \quad 
m = \Exp_{eH}(ix_r).\] 
The relation $\tau_h(x_r) = -x_r$ entails $\oline\tau_h(m) = m$, 
so that the action of $\oline\tau_h$ on $\cT_{M^r}$ 
is given in polar coordinates on $M^r$ by 
\begin{equation}
  \label{eq:polinv}
\oline\tau_h(g.\Exp_m(iy)) 
= \tau_h(g).\Exp_m(-i \tau_h(y)).
\end{equation} 

Let $g.\Exp_m(iy) \in \cT_{M^r}$ be fixed under $\oline\tau_h$. 
Then Proposition~\ref{prop:4.9} implies that 
$\tau_h (g) \in g K$. On $\fg$ the involution $\tau_h$ commutes with 
$\theta$ because $\theta(h) = - h$ and $\tau_h = \tau_{-h}$. In terms of the polar 
decomposition $g= (\exp z) k $, $k \in K$, $z \in \fp$, we thus have 
\[ \tau_h(g) = \exp(\tau_h(z)) \tau_h(k)  \in g K = \exp(z) K.\] 
The bijectivity of the polar decompositions of $G$ thus shows that 
$\tau_h(z) = z$. Hence $k.\Exp_m(iy) = \Exp_m(i \Ad(k)y)$ is also fixed 
under $\tau_h$, and since $\Exp_m$ is injective on $i\Omega_\fp$, 
we further obtain $\Ad(k)y \in \fp^{-\tau_h}$. This shows that 
\begin{equation}
  \label{eq:crownfixpo}
 (\cT_{M^r})^{\oline\tau_h} 
= \exp(\fp^{\tau_h}).\Exp_m(i\Omega_\fp^{-\tau_h}).
\end{equation}
It follows in particular that the submanifold of $\oline\tau_h$-fixed 
points is connected. 

Next we show that 
\begin{equation}
  \label{eq:conjk}
\Omega_\fp^{-\tau_h} 
\subeq \Inn(\fk^{\tau_h}).(\Omega_\fp^{-\tau_h} \cap \fa^{-\tau_h}) 
=  \Inn(\fk^{\tau_h}).(\Omega_\fa^{-\tau_h}). 
\end{equation}
This follows if $\fa^{-\tau_h}$, which is contained in 
$\fq^{-\tau_h}$, is also maximal abelian in $\fp^{-\tau_h}$. 
As ${x_r} \in \fa^{-\tau_h}$, and $\tau_{h_c}$ commutes with $\theta = \tau\tau_{h_c}$ 
and $\tau_h$, we have 
\[ \fp^{-\tau_h} 
= \underbrace{\fp^{-\tau_h, \tau_{h_c}}}_{\supeq \fa^{-\tau_h}} 
\oplus \fp^{-\tau_h, -\tau_{h_c}}
\quad \mbox{ and } \quad 
 \fz_{\fp^{-\tau_h}}(\fa^{-\tau_h}) \subeq \fz_{\fp^{-\tau_h}}({x_r}) \subeq 
\fp^{\tau_{h_c}}.\]
Therefore $\fa^{-\tau_h}$ is maximal abelian in $\fp^{-\tau_h}$. 
This proves \eqref{eq:conjk}, so that \eqref{eq:crownfixpo} yields 
\begin{align*}
(\cT_M)^{\oline\tau_h} 
&=  (\cT_{M^r})^{\oline\tau_h} 
= (G^h)_e.\Exp_m(i\Omega_\fa^{-\tau_h})
= (G^h)_e.\Exp_{eH}(i( {x_r}+ \Omega_\fa^{-\tau_h}))\\
&= (G^h)_e.\Exp_{eH}(i (C_\fa^\pi)^{-\tau_h})
= (G^h)_e.\Exp_{eH}(i (C^\pi)^{-\tau_h}),
\end{align*}
where we have used that $\fa^{-\tau_h}$ is maximal abelian 
in $\fq^{-\tau_h}$ (which is proved in the same way as for~$\fp$), 
to obtain 
\begin{equation}
  \label{eq:28}
(C^\pi)^{-\tau_h} = \Inn(\fh^{\tau_h}).(C_\fa^\pi)^{-\tau_h}.
\end{equation}

\nin {\bf Case 2:} In the general case 
we use the inclusion $C \subeq C^{\rm max}_\fq$
(Theorem~\ref{thm:3.6}) and 
the fact that 
\[ C^{-\tau_h} = (C^{\rm max}_\fq)^{-\tau_h}\] 
(Proposition~\ref{prop:reductiontominmax}), 
which implies in particular that 
\[ (C^\pi)^{-\tau_h} = ((C^{\rm max}_\fq)^{\pi})^{-\tau_h}.\]  
We write $\cT_M(C) \subeq \cT_M(C_\fq^{\rm max})$ for the 
two complex domains associated to $M$ and the two cones 
$C \subeq C_\fq^{\rm max}$. 
From Case 1 we now derive that 
\begin{align*}
  (G^h)_e.\Exp_{eH}(i (C^\pi)^{-\tau_h}) 
&\subeq \cT_{M}(C)^{\oline\tau_h} 
\subeq  \cT_{M}(C^{\rm max}_\fq)^{\oline\tau_h} 
= (G^h)_e.\Exp_{eH}(i ((C^{\rm max}_\fq)^\pi)^{-\tau_h}) \\
&=  (G^h)_e.\Exp_{eH}(i (C^\pi)^{-\tau_h}).
\end{align*}
This proves the desired equality. 
\end{prf}

\subsection{The wedge domain $W_M^{\rm KMS}(h)$ and $W_M(h)$}
We write  $(\alpha_z)_{z \in \C}$ for  the holomorphic extension 
of the modular flow to the complex symmetric space~$M_\C$.
 Recall the wedge domains
\[ W_M^{\rm KMS}(h) 
:= \{ m \in M \: (\forall z \in \cS_\pi)\ \alpha_z(m) \in \cT_M\} 
= \{ m \in M \: (\forall t\in (0,\pi))\ \alpha_{it}(m) \in \cT_M\},\] 
and
\[ W_M(h) 
:= (G^{\tau_h})_e.\Exp_{eH}( (C_+ + C_-)^\pi) 
= (G^h)_e.\Exp_{eH}( (C_+ + C_-)^\pi).\] 
 The domain $W_M^{\rm KMS} (h)$
is called the {\it KMS-wedge domain} and 
$W_M(h)$ is called the {\it wedge domain of polar type 
in $M$} corresponding to the Euler element $h$. 

\begin{rem}\mlabel{rem:Gh} In \eqref{eq:wmh} 
we used  that the identity component $  (G^h)_e$ of the centralizer 
$G^h$ of $h$,  which we will
from now on denote by $G^h_e$, coincides  with $(G^{\tau_h})_e$, the connected 
subgroup with Lie algebra $\g^{\tau_h} = \g_0(h)$. 
If $G$ is $1$-connected, then $G^{\tau_h}$ is connected 
(\cite[Thm.~IV.3.4]{Lo69}), hence 
coincides with the connected Lie group $G^h_e$ 
with Lie algebra $\g_0(h)$. 
We refer to Example~\ref{ex:b.1} in the appendix for more details. 
\end{rem}

\begin{lem} \mlabel{cor:4.12} 
$W_M^{\rm KMS}(h) \subeq \kappa_h^{-1}((\cT_M)^{\oline\tau_h}) =
 W_M(h).$ 
\end{lem}

\begin{prf} Let $p \in W_M^{\rm KMS}(h)$, i.e., 
for every $z \in \cS_\pi$, we have $\alpha_z(p) \in \cT_M$. 
As $h \in \fh$, the base point 
$eH$ in $M = G/H$ is $\alpha$-fixed. Then 
$\tau_h = e^{\pi i \ad h}$ implies that 
\[ \alpha_{\pi i}(p) = \tau_h(p) \quad \mbox{ for } \quad 
p \in M.\] 
Since $\tau_h$ commutes with $\alpha$, this implies that 
\[ \alpha^p(\pi i + t) = \tau_h(\alpha^p(t)) \quad \mbox{ for } \quad 
t\in \R,\] 
and by analytic extension to $\cS_\pi$: 
\[ \alpha^p(\pi i + \oline z) = \oline\tau_h(\alpha^p(z)) 
\quad \mbox{ for } \quad 
z\in \oline{\cS_\pi}.\] 
It follows in particular that 
\begin{equation}
  \label{eq:middlelinefix}
\kappa_h(p) = \alpha^p\Big(\frac{\pi i}{2}\Big) \in (\cT_M)^{\oline\tau_h},
\end{equation}
With Theorem~\ref{thm:4.5} this shows that 
\begin{align*}
W^{\rm KMS}_M(h) 
&\subeq  \kappa_h^{-1}\big((\cT_M)^{\oline\tau_h}\big)
\ {\buildrel \ref{thm:4.5} \over =}\ \kappa_h^{-1}\big((G^{\tau_h})_e.\Exp_{eH}(i (C^\pi)^{-\tau_h})\big) \\
&= (G^{\tau_h})_e.\Exp_{eH}(i \kappa_h^{-1}(C_+ - C_-)^\pi\big) \\
&= (G^{\tau_h})_e.\Exp_{eH}((C_+ + C_-)^\pi)  = W_M(h).
\qedhere\end{align*}
\end{prf}

The following lemma is not obvious because the differential 
of the polar map defining $W_M(h)$ is not everywhere surjective. 

\begin{lem} \mlabel{cor:polwedgeopen}
$W_M(h)$ is an open subset of $M$. 
\end{lem}

\begin{prf} On $M_\C$, we consider the antiholomorphic 
involution $\oline\tau_h$. We also note that the 
antiholomorphic involution $\sigma_G$ 
of $G_\C$ leaves $(G_\C)^\tau$ invariant, hence induces 
an antiholomorphic involution on $M_\C$ whose totally real fixed point 
manifold  contains $M$ as a connected component. 
Next we recall that, on $G_\C$, we have  
\[ \kappa_h^{-1} \oline\tau_h \kappa_h
= \kappa_h^{-1} \sigma_G \tau_h \kappa_h
=  \sigma_G \kappa_h\tau_h \kappa_h
=  \sigma_G \kappa_h^4 = \sigma_G.\] 
This implies that 
\begin{equation}
  \label{eq:zeta-fixedrel}
\kappa_h^{-1}(M_\C^{\oline\tau_h}) = M_\C^{\sigma_G}.
\end{equation}
As $M$ is a connected component of the totally real submanifold 
$M_\C^{\sigma_G}$ and 
$\kappa_h(W_M(h)) = (\cT_M)^{\oline\tau_h}$ 
by Theorem~\ref{thm:4.5}, the assertion follows from the fact that 
\[ (\cT_M)^{\oline\tau_h} 
= \cT_M \cap M_\C^{\tau_h} \] 
is open in $M_\C^{\tau_h}$ because $\cT_M$ is open in $M_\C$. 
\end{prf}

\begin{thm} \mlabel{thm:6.1}  
 Let  $(\g, \tau, C)$ be a ncc reductive symmetric
    Lie algebra and let  $G$ be a corresponding
    connected Lie group satisfying 
    {\rm(GP)} and {\rm(Eff)}, and $H = G \cap G^c$.
  For $M = G/ H$, we then have 
\[W_M(h) = W_M^{\rm KMS}(h) = \kappa_h^{-1}((\cT_M)^{\oline\tau_h}) . \]
\end{thm}

\begin{prf} In view of Lemma ~\ref{cor:4.12}, it remains to show that 
$W_M(h)  \subeq  W_M^{\rm KMS}(h)$.
Since both sides are invariant under $(G^h)_e$, 
it suffices to consider elements of the form 
$\Exp_{eH}(x)$, $x \in (C_+ + C_-)^\pi$. 
As $h, x \in \g^{\rm ct} = \fh^{\tau_h} \oplus \fq^{-\tau_h}$, 
this can be verified in the Cayley type subalgebra $\g^{\rm ct}$. 
We may therefore assume that $\tau = \tau_h$. 

The subspace 
$E := \g_1(h) = \fq_1(h)$ carries the structure of a euclidean Jordan algebra 
whose simple ideals correspond to the simple ideals of $\g$ 
(\cite[Ch.~X]{FK94}). Its 
rank $r$ coincides with the dimension of a maximal abelian subspace 
$\fa_\fq \subeq \fq_\fp$. If $E_\fa \subeq E$ is the span of a 
Jordan frame, i.e., a maximal associative subalgebra, then 
\[ \ft_\fq = \{ x + \theta(x) \: x \in E_\fa \} \] 
is a maximal abelian elliptic subspace of $\fq$, contained in $\fq_\fk$. 
Now every element in the elliptic cone 
$C_+^\circ + C_-^\circ \subeq \fq$ 
is conjugate under $\Inn(\fh_\fk)$ to an element of $\ft_\fq$. 
We may therefore assume, in addition, that $x \in \ft_\fq$.  
. 
Next we recall that $h$ and $\ft_\fq$ generate a subalgebra 
$\fs \cong \fsl_2(\R)^r \subeq \g$ 
with $\fs_1(h) = E_\fa$ in which $\ft_\fq$ is a compactly 
embedded Cartan subalgebra (Proposition~\ref{prop:testing}). In particular, we obtain by intersection a 
Cayley type subalgebra
\[ (\fs, \tau_h, C \cap \fs, h) \] 
which is a direct sum of Cayley type subalgebras isomorphic to 
$\fsl_2(\R)$. From the corresponding assertion 
for $\fsl_2(\R)$, which follows from
Proposition~\ref{prop:desit}, specialized to $\dS^2$, we now obtain that 
$\Exp_{eH}(x) \in W_M^{\rm KMS}(h)$. 
\end{prf}

\subsection{The relation to wedge domains in $G_\C/G^c$}

We now show that the domains from Theorem~\ref{thm:6.1} 
are a connected component 
of the fixed point set of the conjugation $\sigma$ on the corresponding
domain in the bigger ncc space $G_\C/G^c$. Here $\sigma \: G_\C\to G_\C$
is the conjugation with $G = (G_\C^\sigma)_e$.
As before, we set $G^c = (G_\C)^{\bar\tau}_{e}$ and $H = H_{\rm max} = G\cap G^c$. 
It is shown in \cite{MNO22a} that this is the maximal choice for $H$ and that
$H =H_{e}K^h$.
 
 \begin{theorem} \mlabel{thm:6.5}
Let $\sigma = \sigma_M$ be the conjugation on $G_\C/G^c$ 
fixing $M$ pointwise. Then 
\[
W_M(h) = (W_{G_\C/G^c}(h))^\sigma _{eH}=
 W_M^{\rm KMS}(h) = ( W_{G_\C/G^c}^{\rm KMS}(h))^\sigma_{eH} .\] 
 \end{theorem}

\begin{prf} By Theorem \ref{thm:6.1}, applied to $M$ and $G_\C/G^c$, it is enough to show
that 
\[W_M^{\rm KMS}(h) = ( W_{G_\C/G^c}^{\rm KMS}(h))^\sigma_{eH}.\]
It is clear by the definition that 
$W_M^{\rm KMS}(h) \subseteq ( W_{G_\C/G^c}^{\rm KMS}(h))^\sigma_{eG^c}$.
On the other hand, if 
\[ m\in ( W_{G_\C/G^c}^{\rm KMS}(h))^\sigma_{eG^c}
\subset (G_\C/G^c)^\sigma, \] it follows
from \eqref{eq:FixM} and Subsection~\ref{sect:Complex}
 that $m\in M$ and hence by the definition of the KMS-wedge that $m\in W_M^{\rm KMS}(h)$.
\end{prf}

The modular one-parameter group acts  on 
$M = G/H$  
by 
\[  \alpha_t(gH) 
= \exp(th) g H  = g.\exp(t \Ad(g)^{-1} h) H. \] 
It is generated by the vector field given by 
\begin{align}
  \label{eq:Xh}
X_h^M(gH) 
&= g.p_\fq(\Ad(g)^{-1}h).
\end{align}
We therefore have
\begin{equation}
  \label{eq:posdomchar}
 W_M^+(h) = \{ m \in G/H \: X_h^M(m) \in V_+(m)\} 
= \{ gH  \: \Ad(g)^{-1}h \in \fh + C_\fq^\circ\}.
\end{equation}
In the same way we define the positivity domain 
$W_{G_\C/G^c}^+(h)$ in $G_\C/G^c$.

Note that $p_\fq(\Ad(g)^{-1}h) \in C_\fq^\circ$ is equivalent to 
\begin{equation}
  \label{eq:tubepos}
\Ad(g)^{-1}h \in \fh + C_\fq^\circ, 
\end{equation}
so that we have to understand the intersection of the adjoint orbit 
$\cO_h = \Ad(G)h$ with the real tube domain $\fh + C_\fq^\circ$.

\begin{lemma} $W^+_M(h)_{eG^c} \subeq W^+_{G_\C/G^c}(h)^\sigma_{eG^c}
  \subeq W_M^+(h)$.
\end{lemma}

\begin{prf} The proof follows the proof of Theorem~\ref{thm:6.5} above. 
It is clear by $C^\circ \subeq i C_{\g^c}^\circ$ that 
\[ W^+_M(h)_{eG^c} \subset (W^+_{G_\C/G^c}(h)^\sigma)_{eG^c}.\] 
But if $m\in W^+_{G_\C/G^c}(h)^\sigma_{eG^c}$, 
then $m\in M$ and $m=gG^c$ for some $g\in G$ 
(see \eqref{eq:FixM} in Subsection~\ref{sect:Complex}). 
The definition of $W^+_{G_\C/G^c}(h)$ entails 
\[\Ad (g)^{-1}h \in (\fg^c + i C_{\fg^c}^\circ) \cap \g 
= \fh + C^\circ.\]
Hence $m\in W^+_M(h)_{eG^c}$.
\end{prf} 

\section{The positivity domain $W^+_M(h)$} 
\mlabel{sec:7} 

Recall the positivity domain 
\begin{equation}
  \label{eq:posdomcharb}
 W_M^+(h) = \{ m \in G/H \: X_h^M(m) \in V_+(m)\} 
= \{ gH  \: \Ad(g)^{-1}h \in \fh + C_\fq^\circ\}.
\end{equation}
This set is not necessarily connected. In the first part we show that in
general    $W_M^+(h)_{eH} = W_M(h)$. Then we discuss the case where $G$ is simple
and $H=G^h$, so that $M=G/H$
is of Cayley type. In this case
we show that $W_M^+(h)$ is connected.

\subsection{The identity component $W_M^+(h)_{eH}$}

Recall that $M = G/H$, where 
$H = G \cap G^c$ holds for 
an open subgroup $G^c \subseteq (G_\C)^{\otau}$ 
(see \eqref{eq:Hintersec}). 
The man result in this section is the following theorem,
showing that the identity components of the
different wedge domains coincide. That the domains $W_M(h)$ and
$W^{\rm KMS}_M(h)$ are connected has already been pointed out. We show
in \cite{MNO22b} that the positivity domain $W^+_M(h)$ is also connected
under our assumption, but not for coverings.

 \begin{thm} \mlabel{thm:7.4} 
   Let $(\g,\tau,C_\fq,h)$ be a modular ncc symmetric Lie algebra,
   $G$ a corresponding connected Lie group satisfying
   {\rm(GP)} and {\rm(Eff)} and    $H = G \cap G^c$.  Then 
\[ W_{M}^+(h)_{eH} = W_M(h) = (G^h)_e.\Exp_{eH}((C_+ + C_-)^\pi).\]
\end{thm}
 
We prove this in two steps. We start with the one inclusion: 

\begin{lem} \mlabel{thm:pos-wedge-incl} 
If $(\g,\tau,C,h)$ is a reductive modular ncc symmetric Lie algebra 
and   that {\rm(GP)} and {\rm(Eff)} are satisfied for the corresponding
connected Lie group~$G$. Then $W_M(h) \subeq W_M^+(h)$. 
\end{lem}

\begin{prf} Since both sets are invariant under 
$(G^h)_e$, and 
\[ W_M(h) = (G^{\tau_h})_e.\Exp_{eH}((C_+ + C_-)^\pi),\] 
it suffices to see that 
$p_\fq(e^{-\ad x}h) = -\sinh(\ad x) h \in C_\fq^\circ$ 
holds for every $x \in (C_+ + C_-)^\pi$. 

As $h \in \fh^{\tau_h}$ and $x \in \fq^{-\tau_h}$ 
are both contained in 
$\g^{\rm ct} = \fh^{\tau_h} + \fq^{-\tau_h}$, 
we may assume that $\tau = \tau_h$. 
In the same way as in the proof of Theorem~\ref{thm:6.1}, 
we further reduce to the case where 
$\g = \fsl_2(\R)^r$, endowed with the canonical diagonal 
Cayley type structure. This reduces the problem to 
$(\fsl_2(\R), \tau_h)$, which corresponds to 
$2$-dimensional de Sitter space. Hence the assertion follows from 
Subsection~\ref{subsec:a.3} 
or by the calculations in Section~\ref{subsec:sl2b}

\end{prf}

\begin{prf}[Proof of Theorem \ref{thm:7.4}:] In view of Lemma~\ref{thm:pos-wedge-incl}, it remains to show that 
  \begin{equation}
    \label{eq:goal72}
 W_{M}^+(h)_e 
\subeq (G^h)_e.\Exp_{eH}((C_+ + C_-)^\pi) 
= (G^h)_e\exp((C_+ + C_-)^\pi)H.
  \end{equation}

We write $\g = \g_r \oplus \g_{nr}$ for the decomposition of $\g$ 
into the maximal $\tau$-invariant Riemannian ideal $\g_r$ 
and its complement $\g_{nr}$ which is a direct sum of 
irreducible ncc symmetric Lie algebras  (Proposition~\ref{prop:decomp}).
Then $[h,\g_r] = \{0\}$ implies that $G_r := \la \exp \g_r \ra\subeq (G^h)_e$. 
We also note that 
\[ h = h_r + h_{nr}\in C^\circ \subeq C^{\rm max}_\fq 
= \fq_r \oplus C^{\rm max}_{\fq_{nr}} \] 
with $h_r \in \z(\g) \subeq \g_r$, so that 
the projection $h_{nr}$ of $h$ onto $\g_{nr}$ is contained in 
$C^{\rm max}_{\fq_{nr}}$. Further 
\[ C_\pm = \pm C \cap \fq_{\pm 1}(h) 
= \pm C^{\rm max}_\fq \cap \fq_{\pm 1}(h) \] 
by Proposition~\ref{prop:reductiontominmax}. 
We may therefore assume w.l.o.g.\ that $\g = \g_{nr}$ and 
that $C = C^{\rm max}_\fq$. 

Next we recall from \eqref{eq:posdomcharb} that
\[ W^+_{M}(h) = \{ g H \in M\: \Ad(g)^{-1} h \in 
\fh +  C^\circ\}.\] 
As we have seen in \eqref{eq:firstconj2} in Appendix~\ref{subsec:sl2b},
 the elements 
\begin{equation}
  \label{eq:whatisx1}
 x_0 = \frac{\pi}{4} \sum_{j=1}^s (e_j - f_j) \in \ft_\fq, \quad 
 h = \frac{1}{2} \sum_{j=1}^s h_j^0 
\quad \mbox{ and } \quad 
 h_c = \frac{1}{2} \sum_{j=1}^s (e_j + f_j) 
\end{equation} 
in the subalgebra $\fs_0 \cong \fsl_2(\R)^s$ generated by 
$h_c$ and $\ft_\fq$ satisfy 
\[ \Ad(g_0) h_c = h \quad \mbox{ for } \quad g_0 := \exp(x_0).\] 
Therefore $gH \in W_{M}^+(h)$ is equivalent to 
\[  \Ad(g^{-1}g) h_c =  \Ad(g)^{-1} h \in \cT_C = \fh + C^\circ.\] 
Theorem~\ref{thm:crownchar-gen} now shows that, 
if $gH$ is contained in the connected component 
of this set containing $g_0H$, then 
\[ g^{-1}g_0\in H \exp(\Omega_{\ft_\fq}) G^{h_c},\] 
which is equivalent to 
\[ g^{-1} \in 
H \exp(\Omega_{\ft_\fq}) G^{h_c}g_0^{-1} 
= H \exp(\Omega_{\ft_\fq}) g_0^{-1} G^{\Ad(g_0)h_c}
= H \exp(\Omega_{\ft_\fq}- x_0) G^h,\] 
resp., to 
\[ g \in G^h \exp(x_0 - \Omega_{\ft_\fq})H.\] 

In view of \eqref{eq:goal72}, it suffices to show that 
\[ x_0 - \Omega_{\ft_\fq} \subeq (C_+ + C_-)^\pi \cap \ft_\fq.\]
By \eqref{eq:coneslice}, the set on the right hand side is 
\[ (C_+ + C_-)^\pi \cap \ft_\fq 
= \Big\{ \sum_{j = 1}^s \frac{c_j}{2}(e_j - f_j) \: 0 < c_j < \pi\Big\}.\]
So it remains to observe that \eqref{eq:whatisx1} yields 
$x_0 - \Omega_{\ft_\fq} = \big\{ \sum_{j = 1}^r \frac{c_j}{2}(e_j -f_j) \: 0 < c_j < \pi\big\}.$ 
\end{prf}

The following proposition provides an interesting link between 
compactness of stabilizers and the positivity domain 
$W_M^+(h)$ of the vector field $X_h^M$ on $M$. 

\begin{prop} If $(\g,\tau,C,h)$ is reductive and 
$m \in W_M^+(h)$, then the stabilizer $G_m^h$ acts as a relatively 
compact group on $T_m(M)$. 
\end{prop}

\begin{prf} The stabilizer group $G_m^h$ acts on the pointed generating 
cone $V_+(m) \subeq T_m(M)$ and fixes the element 
$X_h^M(m)$ in its interior. Therefore the image of this representation 
is relatively compact. 
\end{prf}

\subsection{The Cayley type case} 
\mlabel{subsec:7.1} 

In this section we discuss the special case where $M$ is
irreducible of Cayley type. We prove our main assertion \eqref{eq:e1} on wedge domains
for Cayley type spaces and show that, in this case, the positivity domain  
$W_M(h)$ is connected.

We assume that $(\g,\tau)$ is simple of Cayley type 
with Cartan involution $\theta$ such that the Euler element 
$h \in \fh$ satisfies $\theta(h) = -h$ and $\tau = \tau_h$. 
We consider the corresponding centerfree group 
$G= \Inn(\g)$ which automatically satisfies {\rm(GP)} and {\rm (Eff)}.
Further $H := G^h$ is an open subgroup of $G^\tau$ 
because $G^h$ commutes with $\tau = \tau_h$.
Note that \cite[Prop.~7.9]{MNO22a} implies that
  \begin{equation}
    \label{eq:ggc}
    G \cap G^c = K^{h_c} \exp(\fh_\fp)
  \end{equation}
holds for any causal
  Euler element $h_c \in  \fq_\fp$, so that
  $\tau_h = \tau = \theta \tau_{h_c}$.

\begin{ex}
A typical example is $2$-dimensional
de Sitter space $\dS^2 \cong \SO_{2,1}(\R)_e/\SO_{1,1}(\R)_e$. In this case
$\fg_0(h) = \fh$, $\fg_{\pm 1}(h) = \fq_{\pm 1}(h)=\fq_\pm$, 
$\fq = \fq_{+1}(h)\oplus \fq_{-1}(h)$ and, up to sign, the cone 
$C= C_+ - C_-$ is the only   $H=G^h$-invariant
hyperbolic convex generating cone in
$\fq$ (\cite[Thm. 2.6.8]{HO97}). For a detailed discussion 
of these spaces see \cite[Sec. 2.6]{HO97} and \cite{NeO00}.
\end{ex}

The involution $\tau$ commutes with $\theta$ 
because $\theta(h)= -h$. 
If we normalize $z \in \fz(\fk)$
in such a way that $iz \in \g_\C$ is an Euler element,
then 
\[ \tau(z) = -z \quad \mbox{ and }\quad 
\theta = e^{\pi \ad z} \in G^\tau \setminus G^h.\] 
As $G^h$ preserves the $3$-grading of $\g$ defined by $h$,
the cones 
$C_\pm = \pm C_\g \cap \g_{\pm 1}(h)$ are both $G^h$-invariant, 
but the Cartan involution satisfies $\theta(C_+) = C_-$
(Proposition~\ref{prop:2.18}). 
Therefore the elliptic cone $C_+ + C_-$ is $\theta$-invariant, 
but the hyperbolic cone $C := C_+ - C_-$ is not. 
The following lemma actually shows that 
\begin{equation}
  \label{eq:2comp}
  G^{\tau} = G^h \cup \theta G^h,
\end{equation}
so that $G^h$ is the maximal subgroup of $G^\tau$ leaving the
hyperbolic cone
$C = C_+ - C_-$ invariant.
From $G^c = e^{\frac{\pi i}{2} \ad h} G e^{-\frac{\pi i}{2} \ad h}$
  it  follows that $G^h \subeq G \cap G^c$.
  Further $K^{h_c} \subeq G^{\theta \tau_{h_c}} = G^{\tau}$
  leaves the cone $C = C_+ - C_- = C^{\rm max}$ invariant (\cite[Rem.~4.20]{MNO22a}).
  We thus also derive from \eqref{eq:ggc} and \eqref{eq:2comp} 
  that  $G \cap G^c \subeq G^h$,
  hence the equality
  \[ G \cap G^c = G^h.\]

\begin{lemma}
Let $\g$ be a simple real Lie algebra 
and $h \in \cE(\g) \cap \fh$ an Euler element. 
In the centerfree group  $G = \Inn(\g)$, we then  have 
$|G^{\tau_h}/G^h| \le 2$.
\end{lemma}

If $h$ is symmetric, i.e., $-h \in \cO_h$, then 
any $g \in \Inn(\g)$ with $gh = -h$ is contained in 
$G^\tau \setminus G^h$.

\begin{prf} It is clear that $G^h \subset G^{\tau_h}$. 
Assume, conversely, that $g \in G^{\tau_h}$. 
Then $g$ normalizes $\fh = \g^{\tau_h}$ and leaves $\fq$ invariant. 
Next we observe that $\fq = \fq_1(h) \oplus \fq_{-1}(h)$ is the decomposition 
of $\fq$ into two simple pairwise inequivalent $\fh$-modules
(\cite{HO97}). 
Therefore either 
$g.\fq_{\pm 1}(h) = \fq_{\pm 1}(h)$ or 
$g.\fq_{\pm 1}(h) = \fq_{\mp 1}(h)$. 
In the first case $g.h = h$ and in the second case $g.h = -h$. 
So $G^{\tau_h}/G^h \cong G^{\tau_h}h \subeq \{\pm h\}$ has at most two elements.
\end{prf}

We put 
\[ P_\pm = \exp(\g_{\pm 1}(h)) G^h \quad \mbox{ and } \quad 
M_\pm := G/P_\mp.\] 
We then have an embedding 
\[ M = G/G^h \into G/P_- \times G/P_+,\quad  g G^h \mapsto (gP_-, gP_+) \] 
where the space on the right carries the order for which the above 
inclusion is a causal embedding. 
First we consider the causal homogeneous spaces 
$M_\pm$. The cones 
\[ \pm C_\pm \subeq \g_{\pm 1}(h) \cong \g/\fp_\mp \] 
are $P_{\mp}$-invariant because 
the group $\exp(\g_{\pm 1}(h))$ 
acts trivially on $\g/\fp_\pm \cong \g_{\mp 1}(h)$.
We therefore obtain $G$-invariant cone fields 
\[ V_+(gP_\mp) := g.(\pm C_\pm^\circ),
\quad \mbox{ on } \quad M_\pm = G/P_{\mp}. \]

\begin{prop} {\rm(The positivity domain in the flag manifold)} 
     In $M_+$ the positivity domain 
\[ W_{M_+}(h) := \{ m \in M_+ \: X_h^{M_+}(m) \in V_+(m) \} \] 
of the modular vector field 
$X_h^{M_+}$ coincides with the subset $\exp(C_+^\circ)P_-$, which 
corresponds to the positive cone $E_+$ under the open dense embedding 
of the euclidean Jordan algebra $E = \g_1(h)$ into~$M_+$. 
\end{prop}

\begin{prf} 
We identify the euclidean Jordan algebra $E = \g_1(h)$ 
with the corresponding open subset of~$M_+$. 
On this subset the vector field $X_h^M$ coincides with the 
Euler vector field $X_h^M(v) = v$. Therefore the positivity domain 
$W_{M_+}(h)$ intersects $E$ precisely in the open cone $E_+$. 
As $W_{M_+}(h)$ is open and $E$ is dense in $M_+$, 
it follows that $W_{M_+}(h) \cap E = E_+$ is dense in $W_{M_+}(h)$. 
We therefore have to analyze points in the closure of $E_+$.

Let $\theta$ be a Cartan involution on $\g$ commuting with 
$\tau = \tau_h$. We then have 
\[ \fk = \fh_\fk \oplus \fq_\fk 
\quad \mbox{ and } \quad 
\fh_\fk 
= \fk \cap (\g_0(h) + \g_1(h)) 
= (\g_0(h) + \g_1(h))^\theta.\] 
From a Jordan frame in $E$, we obtain an associative 
Jordan subalgebra $E_\fa \subeq E$, and 
\[ \ft_\fq = \{ x + \theta(x) \: x \in E_\fa\subeq E = \g_1(h)\} \] 
is a maximal abelian subspace of $\fq_\fk$ 
(see the proof of Theorem~\ref{thm:6.1}). 
Then $h$ and $\ft_\fq$ generate a subalgebra 
\[ \fs \cong \fsl_2(\R)^r \subeq \g \] 
with $\fs_1(h) = E_\fa$ in which $\ft_\fq$ is a compactly 
embedded Cartan subalgebra. 
This is a Cayley type subalgebra
\[ (\fs, \tau_h, C \cap \fs, h) \] 
which is a direct sum of Cayley type subalgebras isomorphic to $\fsl_2(\R)$. 

For the groups $K = \exp \fk$, $H_K := \exp \fh_\fk$, and the torus 
$T := \exp \ft_\fq$, we then have $K = H_K T H_K$. 
Since $K$ acts transitively on $M_+$ and $H_K \subeq G^h \subeq P_-$, 
it suffices to determine the positive elements in 
the torus $T_+ := T.eP_+ \subeq M_+$. This set can be determined in terms 
of the subalgebra $\fs$ and the Jordan algebra $E_\fa$. It follows that 
\[ T_+ \cong (\bS^1)^r \cong (\R_\infty)^r, \] 
and that the modular vector field is given on the dense open subset 
$\R^r = T_+ \cap E \subeq T_+$ by 
\[ X_h(x_1, \ldots, x_r) = (x_1, \ldots, x_r).\] 
We conclude that 
\[ W_{M_+}(h) \cap T_+ = E_+ = (\R_+)^r,\] 
and this implies that 
\[ W_{M_+}(h) = H_K (W_{M_+}(h) \cap T_+) = H_K(E_+ \cap E_\fa) = E_+.
\qedhere \] 
\end{prf}

Now we turn to the Cayley type spaces, 
where $C = C_+ - C_- \subeq \fq = \fq_1 \oplus \fq_{-1}$ 
is a hyperbolic cone and $G^h = P_+ \cap P_-$. 

\begin{prop} \mlabel{prop:7.8}
  For a Cayley type space $M = G/G^h$, we have 
  \[   W_M^+(h) = \Exp_{eH}((C_+ + C_-)^\pi) = W_M(h).\]
In particular, $W_M^+(h)$ is connected.
\end{prop}

\begin{prf}
We consider the open embedding 
\begin{equation}
  \label{eq:cayemb}
\eta \: G/G^h = M \into M_+ \times M_-, \quad 
gG^h \mapsto (gP_-, gP_+). 
\end{equation}
Accordingly, 
$\fq \cong T_{eH}(M)$ decomposes $\fq_1 \oplus \fq_{-1}$. 

The Cartan involution $\theta$ satisfies $\theta(\fq_1) = \fq_{-1}$ 
and $\theta\res_{\fa} = - \id_\fa$, so that 
\[ \theta(C^\circ) 
= \theta(\Ad(G^h)C_\fa^\circ)
= \Ad(G^h)\theta(C_\fa^\circ)
= \Ad(G^h)(-C_\fa^\circ) = - C^\circ  \]
and thus $\theta(C) = - C$. In particular, we have $\theta(C_+) = C_-$, 
so that 
\[ \tilde \theta \: M_+ \to M_-, \quad 
gP_- \mapsto \theta(g) P_+ = \theta(gP_-) \] 
is a diffeomorphism. It reverses the causal structure if 
$M_-$ carries the cone field defined in $\fq_{-1}$ by the cone $-C_- 
= - \theta(C_+)$. For this causal structure, 
\eqref{eq:cayemb} becomes an embedding of causal manifolds. 
The map $\tilde\theta$ intertwines the vector field 
$X_h^{M_+}$ with the vector field $-X_h^{M_-}$, so that 
\[ W_{M_-}(h) = \tilde\theta(W_{M_+}(h)) = \tilde \theta(E_+) = \exp(C_-^\circ) P_+.\] 
This shows that 
\begin{align*}
W_M^+(h) 
&= \eta^{-1}(W_{M_+}(h) \times W_{M_-}(h)) 
= \eta^{-1}(\exp(C_+^\circ)P_- \times \exp(C_-^\circ)P_+) \\
&= (\exp(C_+^\circ)P_- \cap \exp(C_-^\circ)P_+)/G^h.
\end{align*}

It remains to show that 
\begin{equation}
  \label{eq:2=1}
  W_M^+(h)\ {\buildrel ! \over =}\ \Exp_{eG^h}((C_+ + C_-)^\pi) 
= W_M(h). 
\end{equation}
In view of the $G^h$-invariance of both sides,
we can verify both inclusions  
on suitable representatives of $G^h$-orbits on both sides. 
As the cone $C_+ + C_-$ is elliptic 
and $\ft_\fq \subeq \fq_\fk$ is maximal abelian, every 
$(G^h)_e$-orbit in $(C_+ + C_-)^\circ$ meets $\ft_\fq$. 
On the Jordan algebra level, this corresponds to the fact that  
every pair in $E_+ \times E_+$ can 
be diagonalized simultaneously: 
\[ E_+ \times E_+ = \Ad((G^h)_e).(E_{\fa,+} \times E_{\fa,+}).\] 

It therefore suffices to verify \eqref{eq:2=1} for the 
Jordan algebra $E_\fa \cong \R^r$ with the conformal 
Lie algebra $\fs = \fsl_2(\R)^r$. 
Then $M \cong (\dS^2)^r \into (\T^2)^r$, and the product 
structure reduces the problem to 
the $2$-dimensional de Sitter space $M = \dS^2 
\cong \PSL_2(\R)/\SO_{1,1}(\R)_e$, for which we know from 
Proposition~\ref{prop:desit} that 
$W_{\dS^2}^+(h)$ coincides with $W_{\dS^2} (h)$. 
\end{prf}

\section{Perspectives and open problems} 
\mlabel{sec:8}

In this final section we collect some open problems related 
to wedge domains in non-compactly causal symmetric spaces.

\subsection{Connectedness of the positivity domain 
for ncc symmetric spaces} 
\mlabel{subsec:8.1} 

As we know from Subsection~\ref{subsec:a.3}, 
the positivity domain $W_{M}^+(h)$ is connected for de Sitter space 
and for the Cayley type 
spaces (Proposition~\ref{prop:7.8}).
One may therefore expect this to be the case in general, 
but this is not true. The reason for this are non-trivial coverings. 
If $M = G/H$ and 
$q_M \: \tilde M = \tilde G/\tilde G^\tau \to M$ 
denotes its simply connected covering, then 
\[ W_{\tilde M}^+(h)= q_M^{-1}(W_M^+(h)),\] 
and since  $W_M(h)$ is simply connected 
(even contractible) by \cite[Prop.~3.6]{MNO22b}, it lifts to~$\tilde M$, 
so that $W_{\tilde M}^+(h)$ is not connected if the covering  
is non-trivial. To avoid this source of non-connectedness, one has to consider
the space $M = G/H$ for
\begin{equation}
\label{eq:hmax}
  G = \Inn(\g) \quad \mbox{ and }\quad
  H := \{ g \in G^\tau \: \Ad(g)C = C \}.
\end{equation}

\begin{ex} The non-compactly causal space $M = \dS^2$ 
($2$-dimensional de Sitter space) is not simply connected 
with $\pi_1(\dS^2) \cong \Z$. The positivity domain 
$W_M^+(h)$ is contractible, hence simply connected. 
Therefore  its inverse image in $\tilde M$ 
has infinitely many connected components.
\end{ex}

In \cite[\S 7]{MNO22b} we show that: 
\begin{thm} Let $(\g,\tau,C)$ be an irreducible 
ncc symmetric Lie algebra, $G = \Inn(\g)$ 
and $H := H_{\rm max} = \Inn_\g(\fh) K^{h_c}   \subeq  G^\tau$
  for a causal Euler element~$h_c$.
Then the positivity domain $W_M^+(h_c)$ in
$M := G/H$ is connected. 
\end{thm}

\subsection{Fixed points of the modular flow} 
\mlabel{subsec:8.2} 

If $M = G/H$ is a symmetric space associated to the symmetric 
Lie algebra $(\g,\tau)$ and $h \in \cE(\g)$ is an Euler element, 
then its adjoint orbit $\cO_h = \Ad(G)h$ contains an element 
of $\fh$ if and only if the corresponding vector field 
$X_h^M$ on $M$ has a zero. In our axiomatic 
setup for modular causal symmetric Lie algebras 
we assume $h \in \fh$,  
which can be achieved by conjugation if $X_h^M$ has a zero.
However, there are cases where $\g$ contains an Euler element, 
but none is contained in $\fh$. In this subsection we discuss some 
of the related problems. \\

If $(\g,\tau,C)$ is irreducible and 
non-compactly causal, then $C$ always contains a causal 
Euler element $h_c$ which may not be conjugate to an element of $\fh$ 
(cf.~\cite[Thm.~4.4]{MNO22a}). 
However, we may use it to define a modular flow on $M$, 
and since $h_c \in C^\circ$, the corresponding positivity domain 
\[ W_M^+(h_c) = \{ m \in M \: X_{h_c}^M(m) \in V_+(m) \}
= \{ gH \: \Ad(g)^{-1}h_c \in \fh + C^\circ\} \]
is a non-empty open subset of~$M$. To determine the structure 
of this set, one needs a good description 
of the intersection $\cO_{h_c} \cap (\fh + C^\circ)$, 
which is an analog of a real crown domain of the Riemannian 
symmetric space $\Ad(H)h_c \cong H/H_K$ 
(see in particular \cite[Thm.~6.7]{MNO22a} for details).

In \cite{MNO22b} we show that it is 
also possible to develop a rich theory of wedge domains in 
non-compactly causal symmetric spaces 
for causal Euler elements $h_c$, even if $X^M_{h_c}$ has no zeros.
Its flavor is very different from the approach in the present paper. 

\begin{prob}
The tube domain 
\[ \cT_M := G.\Exp_{eH}(i C^\pi) \subeq M_\C \] 
does always make sense, but is it open in $M_\C$? 
As the polar map of $\cT_M$ has singularities, this is not clear. 
Since $\cT_M$ is defined, 
\[ W_M^{\rm KMS}(h) := \{ m \in M \: 
(\forall z \in \cS_\pi)\ \alpha_z(m) \in \cT_M\}\] 
also makes sense. 
Since $\tau_{h_c}(h_c) = h_c$, the two cones $C$ and $-\tau_{h_c}(C)$ 
will certainly be different, so that it is not clear what 
the correct analog of $C^{-\tau_h}$ is in this context. 
\end{prob}

In view of \cite[Thm.~5.4]{MNO22a}, in the irreducible case, 
the vector field $X^M_{h_c}$ associated to a causal 
Euler element $h_c$ has a zero if and only if 
$h_c$ is symmetric, i.e., $-h_c \in \cO_{h_c}$,
which in turn is equivalent to $\cE(\g) \cap \fh \not=\eset$.
If this is the case, then 
there exists an element $h \in \cE(\g) \cap \fh_\fp$, 
so that $\theta(h) = -h$ and $\tau(h) = h$ implies that 
$\tau_{h_c}(h) = \tau \theta(h) = - h$. 
Therefore \cite[Thm.~3.13]{MN21} implies that 
$h$ and $h_c$ generate a subalgebra isomorphic to $\fsl_2(\R)$ 
which is invariant under $\tau$ and $\theta$. In particular 
$\cO_h= \cO_{h_c}$. This shows that 
\begin{equation}
  \label{eq:hincl}
 \cE(\g) \cap \fh \subeq \cO_{h_c} \quad \mbox{ and }  \quad 
\cO_h = \cO_{h_c} = - \cO_{h_c}. 
\end{equation}

\begin{ex} (The complex case)
  If $\g$ is simple hermitian not of tube type, then 
$(\g_\C,\tau)$, with $\tau(z) = \oline z$, is irreducible ncc 
of complex type. The Lie algebra $\g_\C$ contains an Euler element, 
but none is contained in $\g$ if it is not of tube type 
(cf.~Proposition~\ref{prop:types}). 

Moreover, there may be many adjoint orbits of Euler elements 
in $\g_\C$, for instance 
$\fsl_n(\C)$ contains $n-1$ classes (cf.\ \cite[Thm.~3.10]{MN21}), 
but once a hermitian real form is fixed, the conjugacy class is 
determined by the requirement $h_c \in i\g$. 
In fact, if $h_c \in i \g$ is an Euler element, then 
$z := i h_c \in  \g$ is elliptic with 
$\Spec(\ad z)  = \{ 0, \pm i\}$ and 
$\fk := \ker(\ad z)$ is maximal compactly embedded in~$\g$. 
Hence there are precisely two $\Ad(G)$-orbits of 
Euler elements in $i\fg$. 
\end{ex}

Here is a prototypical example of an ncc symmetric space of Riemannian type: 

\begin{ex} (Modular flow without fixed points and degenerate wedge domain) \\
Consider the open cone $M := \Herm_n(\R)_+$ as a 
symmetric space of 
\[G = \GL_n(\R)_e/\Gamma \] 
 as in Example~\ref{ex:gln}. 
Then $\alpha_t(A) = e^t A$ defines a ``modular'' flow on 
$M$ without fixed points.  
Here $M$ is a Riemannian symmetric space and  
\[ \g = \gl_n(\R) \cong \R h \oplus \fsl_n(\R), \] 
so that $h$ is central. Then $\Ad(g)h = h$ for every $g \in G$ and 
$W_M^+(h) = M$. 
\end{ex}

\subsection{Different choices of $h \in \fh \cap \cE(\g)$} 
\mlabel{subsec:8.3} 

A classification of $\Inn_\g(\fh)$-orbits in 
$\fh \cap \cE(\g)$ has been carried out in \cite[\S 3.1]{NO21}.
It is the same for cc and ncc space because
$\g = \fh + \fq$ and $\g^c = \fh + i \fq$ share the same
subalgebra~$\fh$.
By \eqref{eq:hincl}, 
the different possibilities for $h \in \cE(\g) \cap \fh$ 
are symmetric Euler elements conjugate to the causal Euler element~$h_c$. 

    \appendix

\section{Irreducible modular ncc symmetric Lie algebras}
\mlabel{sec:app.4}

\hspace{-10mm}
\begin{tabular}{||l|l|l|l|l|l||}\hline
$\g$ &  $\g^c = \fh + i \fq$  & $\fh$ 
& $\Delta(\g,\fa)$  & $h$ &  $\g_1(h)$   \\ 
\hline
\hline 
Complex type \phantom{\Big)} &&& &&  \\
\hline 
 $\fsl_{2r}(\C)$ & $\su_{r,r}(\C)^{\oplus 2}$  \phantom{\Big)}& $\su_{r,r}(\C)$
& $A_{2r-1}$ & $h_r$ & $M_r(\C)$   \\
 $\sp_{2r}(\C)$  & $\sp_{2r}(\R)^{\oplus 2}$ &  $\sp_{2r}(\R)$ 
& $C_r$ & $h_r$ & $\Sym_r(\C)$     \\
$\so_{2k}(\C), k > 2$ &  $\so_{2,2k-2}(\R)^{\oplus 2}$ & $\so_{2,2k-2}(\R)$ 
& $D_k$  &$h_1$ & $\C^{2k-2}$    \\
$\so_{2k+1}(\C), k > 1$ &  $\so_{2,2k-1}(\R)^{\oplus 2}$ & $\so_{2,2k-1}(\R)$ 
& $B_k$  &$h_1$ & $\C^{2k-1}$    \\
$\so_{4r}(\C)$ & $\so^*(4r)^{\oplus 2}$ &  $\so^*(4r)$& 
$D_{2r}$ & $h_{2r-1}, h_{2r}$ & $\Skew_{2r}(\C)$      \\
 $\fe_7(\C)$ &  $(\fe_{7(-25)})^{\oplus 2}$ &$\fe_{7(-25)}$ &
$E_7$ & $h_7$ & $\Herm_3(\bO)_\C$     \\
\hline
Cayley type (CT) \phantom{\Big)}&& &&&  \\
\hline 
$\su_{r,r}(\C)$ &  $\su_{r,r}(\C)$ &$\R \oplus \fsl_r(\C)$&
$C_r$ & $h_r$ & $\Herm_r(\C)$ \\
$\sp_{2r}(\R)$  & $\sp_{2r}(\R)$ &$\R \oplus \fsl_r(\R)$ &
$C_r$ & $h_r$ & $\Sym_r(\R)$    \\
$\so_{2,d}(\R), d> 2$ &  $\so_{2,d}(\R)$ &
$\R \oplus \so_{1,d-1}(\R)$ 
&   $C_2$ &$ h_2$ & $\R^{1,d-1}$  \\
$\so^*(4r)$ &  $\so^*(4r)$ &
$\R \oplus \fsl_r(\H)$& $C_r$ & $h_r$ & $\Herm_r(\H)$   \\
 $\fe_{7(-25)}$ & $\fe_{7(-25)}$ & 
$\R \oplus \fe_{6(-26)}$ &  $C_3$ & $h_3$ & $\Herm_3(\bO)$   \\
\hline
Split type (ST) \phantom{\Big (} &&&&&  \\
\hline 
$\fsl_{2r}(\R)$ &$\su_{r,r}(\C)$ &  $\so_{r,r}(\R)$ &
$A_{2r-1}$ & $h_r$ & $M_r(\R)$  \\
$\so_{2r,2r}(\R)$ & $\so^*(4r)$ & $\so_{2r}(\C)$ &
$D_{2r}$ & $h_{2r-1}, h_{2r}$ & $\Skew_{2r}(\R)$  \\
$\fe_7(\R)$ & $\fe_{7(-25)}$ &$\fsl_4(\H)$ &  $E_7$ & $h_7$ &
$\Herm_3(\bO_{\rm split})$   \\
  $\so_{p+1,q+1}(\R)$ &$\so_{2,p+q}(\R)$&
  $\so_{1,p}(\R) \oplus \so_{1,q}(\R)$
& $ B_{p+1}\, (p<q)$  & $h_1$ &$\R^{p,q}$    \\
$p,q > 1$ && & $D_{p+1}\, (p = q)$  & &    \\
\hline 
Nonsplit type (NST) \phantom{\Big)} && &&&  \\
\hline 
 $\fsl_{2s}(\H)$ &$\su_{2s,2s}(\C)$ & $\fu_{s,s}(\H)$ 
& $A_{2s-1}$ & $h_s$ & $M_s(\H)$   \\
 $\fu_{s,s}(\H)$ & $\sp_{4s}(\R)$ &$\sp_{2s}(\C)$ 
& $C_{s}$ & $h_s$ & $\Aherm_s(\H)$  \\
$\so_{1,d+1}(\R)$  & $\so_{2,d}(\R)$ & $\so_{1,d}(\R)$
& $A_1$ & $h_1$ & $\R^d$ \\
\hline
\end{tabular} \\[2mm] 
{\rm Table 1: Irreducible ncc symmetric Lie algebras 
$(\g,\tau)$ with $\cE(\g) \cap \fh \not=\eset$} \\

\section{Some calculations in $\fsl_2(\R)$} 
\mlabel{subsec:sl2b}

Arguments are often reduced to relatively simple $\fsl_2(\R)$ calculations. We therefore
collect the basic notations and calculations here in one place for reference.
For $\g  = \fsl_2(\R)$, we fix the Cartan involution 
$\theta(x) = - x^\top$, so that 
\[ \fk = \so_2(\R) \quad \mbox{ and }\quad 
\fp = \{ x \in \fsl_2(\R) \: x^\top = x \}.\]
The basis elements 
\begin{equation}
  \label{eq:sl2ra}
h^0 := \frac{1}{2} \pmat{1 & 0 \\ 0 & -1},\quad 
e^0 =  \pmat{0 & 1 \\ 0 & 0}\quad \mbox{ and }\quad 
f^0 =  \pmat{0 & 0 \\ 1 & 0}
\end{equation}
and 
\begin{equation}
  \label{eq:sl2rb}
h^1 = \frac{1}{2} \pmat{0 & 1 \\ 1 & 0} = \frac{1}{2}(e^0 + f^0),\quad 
e^1 =  \frac{1}{2}\pmat{-1 & 1 \\ -1 & 1}\quad \mbox{ and }\quad 
f^1= \frac{1}{2} \pmat{-1 & -1 \\ 1 & 1}
\end{equation}
satisfy
\[  [h^j, e^j] = e^j, \quad [h^j, f^j] = -f^j, \quad 
[e^j, f^j] = 2 h^j \quad \mbox{ and }\quad \theta(e^j) = -f^j.\] 

For the involution 
\[ \tau\pmat{a & b \\ c & d} = \pmat{a & -b \\ -c & d} \]
we have
\[ \fh = \g^\tau = \R h^0,  \quad 
\fq = \g^{-\tau} = \R h^1 + \R (e^0 - f^0)
\quad\text{and} \quad  C = [0,\infty) e^0 + [0,\infty) f^0  \] 
is a hyperbolic $\Inn(\fh)$-invariant cone in $\fq$, containing 
$h^1$ as a causal Euler element. Further 
\[ C_+ = [0,\infty) e^0 =\left\{\begin{pmatrix} 0 & x\\ 0 & 0\end{pmatrix}
\: x\ge 0\right\}\quad \mbox{ and }\quad 
C_- = - [0,\infty) f^0 =\left\{\begin{pmatrix} 0 & 0\\ -y & 0\end{pmatrix}\:
y\ge0\right\} \] 
lead to 
\[ C_+ + C_- = [0,\infty) e^0 - [0,\infty) f^0 =\left\{\begin{pmatrix} 0 & x \\ -y & 0\end{pmatrix}\: x,y\ge 0\right\},\] 
so that 
\begin{equation}
  \label{eq:coneslice}
 (C_+ + C_-)^\pi \cap \R(e^0 - f^0) 
= \Big\{ t (e^0 - f^0) \: 0 < t < \frac{\pi}{2}\Big\}.
\end{equation}

The subspace $\ft_\fq := \R (e^0 - f^0) = \so_2(\R)$ of $\fq$ 
is maximal elliptic. For 
\[ x_0 := \frac{\pi}{4} (e^0 - f^0) = \frac{\pi}{4} \pmat{0 & 1 \\ -1 & 0} \in \ft_\fq \]
and 
\[ g_0 := \exp(x_0) = \pmat{\cos\big(\frac{\pi}{4}\big) & \sin\big(\frac{\pi}{4}\big)\\
-\sin\big(\frac{\pi}{4}\big) & \cos\big(\frac{\pi}{4}\big)} 
= \frac{1}{\sqrt 2} \pmat{1 & 1 \\ -1 & 1},\]
we then have 
\begin{equation}
  \label{eq:firstconj2}
 \Ad(g_0) h^1 = h^0 \quad \mbox{ and } \quad \Ad(g_0) h^0 = - h^1,
\end{equation}
More generally, we have for $t \in \R$ 
\begin{equation}
  \label{eq:turn}
 e^{t\ad(e^0 - f^0)} h^1 = \cos(2t) h^1 + \sin(2t) h^0 
 \end{equation}
because $
[e^0 - f^0, h^1] = 2 h^0,  \quad 
[e^0 - f^0, h^0] = -2 h^1.
$ 

Now fix the Euler element
\[ h = h^0 = \frac{1}{2} \pmat{1 & 0 \\ 0 & -1} \quad \mbox{ and } \quad 
\tau_h = e^{\pi i \ad h},\] 
and similarly drop the ${}^0$ in other places. Then
we have 
\[ \g_1(h) = \R e, \quad e = \pmat{0 & 1 \\ 0 & 0} \quad \mbox{ and } \qquad 
\g_{-1}(h) = \R f, \quad f = \pmat{0 & 0 \\ 1 & 0}.\] 
Recall that 
\begin{equation}
  \label{eq:hef-rel}
 [h,e] = e, \qquad [h,f] = -f \quad \mbox{ and } \quad 
[e,f] = 2h.
\end{equation}
We have 
\[ (C_\g^\circ)^{-\tau_h} = \underbrace{\R_+ e}_{C_+^\circ}  - 
\underbrace{\R_+ f}_{C_-^\circ}
\quad \mbox{ with } \quad 
C_+^\circ = \R_+ e, \quad 
C_-^\circ = \R_+ f.\] 

For $y = \lambda e + \mu f \in \g^{-\tau_h}$, we have 
\[ [y,h] = -\lambda e + \mu f \quad \mbox{ and } \quad 
[y,[y,h]] = [\lambda e + \mu f, -\lambda e + \mu f] 
= \lambda \mu 2[e,f] = 4 \lambda \mu h.\] 
We define the function $S \: \C \to \C$ by 
$S(z) = \sum_{k = 0}^\infty \frac{(-1)^k}{(2k+1)!} z^k$, so that 
$\sin(z) = z S(z^2).$ 
Then 
\begin{align*}
\sin(\ad y)h 
&= S((\ad y)^2) [y,h] 
= S(4 \lambda \mu) \cdot (-\lambda e+ \mu f).
\end{align*}
This element is contained in $-C_\g^\circ$ 
if $\lambda, \mu > 0$ 
and $\sin(2\sqrt{\lambda\mu}) > 0$. This is satisfied for 
$0 < 4\mu\lambda < \pi^2$. Then $y$ is hyperbolic with 
$\Spec(\ad y) = \{ \pm 2 \sqrt{\lambda\mu}\} \subeq (-\pi, \pi)$.

\begin{ex} \mlabel{ex:b.1} 
(cf.\ Example~\ref{rem:Gh}) 
 Let $G$ be a connected Lie group with Lie algebra $\g$ 
and $h \in \g$ an Euler element. 
In general $G^h$ may be much larger than $G^{\tau_h}$. 
This can be seen for $G=\tilde\SL_2(\R)$ (the simply connected 
covering group of $\SL_2(\R)$) and the Euler element 
$h = \diag(1/2,-1/2)$, 
 where $\diag$ stands for ``diagonal matrix''.
The group $G^h$ contains $Z(G) \cong \Z$ 
and $G^h_e = \exp (\R h)$, 
which shows that $\pi_0(G^h) \cong \Z$.  
The involution $\tau_h$  acts on $Z(G)\subeq \exp(\so_2(\R))$ by inversion 
(cf.~\cite[Ex.~2.10(d)]{MN21}). 

For $G = \SL_2(\R)$ and the same Euler element~$h$, 
the centralizer $G^h$ is the non-connected subgroup of diagonal matrices, 
which coincides with the fixed point group $G^{\tau_h}$ of the corresponding 
involution 
\[ \tau_h\pmat{a & b \\ c & d} = \pmat{a & -b \\ -c & d}.\] 

On the other hand, if we consider the centerfree group 
$G=\SO_{1,2}(\R)_e \cong \PSL_2(\R)$, then 
$G^h=G^h_e$, but the subgroup $G^{\tau_h}$ has $2$ connected components: 
$G^{\tau_h}= G^h\cup \theta G^h$, where $\theta$ is 
a Cartan involution with $\theta (h)=-h$. 
\end{ex}

\section{Polar maps}
\mlabel{app:polmaps}

In this section we discuss
polar maps associated to an involution on a symmetric space, resp., 
to a pair of commuting involutions on a Lie group.
Key properties  are collected in Lemma \ref{lem:2.7}.
These results are used in particular in Subsection~\ref{subsec:4.1a} 
to obtain a polar decomposition of the crown domain 
of the Riemannian symmetric space $M^r = G/K$ and 
in the characterization of this domain as a subset of the tube domain of 
the cone $C$ (Theorem~\ref{thm:crownchar-gen}).

\subsection{Some spectral theory} 

Let $V$ be a finite dimensional real vector space and 
$A \in \End(V)$. 

\begin{lem} \mlabel{lem:c.1} $\ker\big(\frac{\sinh(A)}{A}\big) 
= \bigoplus_{n \not=0} \ker(A^2 + n^2\pi^2 \1)$. 
\end{lem}

\begin{prf} We may w.l.o.g.\ assume that $V$ is complex. 
Then $B := \sinh(A)/A$ is invertible on all generalized 
eigenspace corresponding to eigenvalues 
$\lambda \not= \pi n i$, $n\in \Z\setminus \{0\}$. 
We may therefore assume that $V$ has only 
one eigenvalue $\lambda = n \pi i$, $n \not=0$. 
Then $A$ is invertible, so that 
\[ \ker(B) = \ker(\sinh(A)) = \ker(e^A- e^{-A}) 
= \ker(e^{2A} - \1).\] 
Writing $A = A_s + A_n$ for the Jordan decomposition of $A$, it 
follows that 
\[ e^{2A} = e^{2 A_s} e^{2 A_n} = e^{2 n \pi i} e^{2 A_n} = e^{2A_n}.\] 
As $\ker(e^{2A_n}-\1) = \ker(A_n)$ follows from 
$2A_n = \log(e^{2A_n})$ as a polynomial in $e^{2A_n}-1$, we see that 
$\ker(B) = \ker(A_n)$ is the $\lambda$-eigenspace of $A$. 
\end{prf}

With similar arguments, or by replacing $A$ by  
$A - \frac{\pi i}{2}\1$, we get: 
\begin{lem} \mlabel{lem:c.2}  $\ker(\cosh(A)) = \ker(e^{-2A} + \1) 
= \bigoplus_{n \in \N}\ker\Big(A^2 + \big(n + \frac{1}{2}\big)^2 \pi^2 \1\Big)$. 
\end{lem}

\subsection{Fine points on polar maps} 

\mlabel{subsec:4.1} 

In this subsection, we consider two commuting involutions $\sigma$ and $\tau$ 
on a connected, not necessarily reductive, 
Lie group $G$ and an open $\sigma$-invariant subgroup $H \subeq G^\tau$.
We shall study the polar map 
\begin{equation}
  \label{eq:polmap1}
 \Phi \: G^\sigma \times_{H^\sigma} \fq^{-\sigma} \to M, 
\quad [g,x] \mapsto g.\Exp_{eH}(x) = g \exp(x)H\in G/H
\end{equation}
and its applications. 

As $H$ is invariant under $\tau$ and $\sigma$, 
both define commuting involutions on $M$ 
and their fixed point 
manifolds intersect transversally in $eH$. 
The map 
\[ \Psi \: G^\sigma\times_{H^\sigma} \fq^{-\sigma}\to N(G^\sigma/H^\sigma), \quad \
[g,x] \mapsto g.x\]
is a diffeomorphism onto the {\it normal bundle} 
$N(G^\sigma/H^\sigma)$ of 
the subspace $G^\sigma.eH \cong G^\sigma/H^\sigma$ 
and $\Phi = \Exp \circ \Psi$, where 
$\Exp \: T(M) \to M$ is the exponential map.

First, we determine the regular points of $\Phi$. 
As $\Phi$ is $G^\sigma$-equivariant, it suffices to 
determine for which points $[e,x]$ the 
tangent map $T_{[e,x]}(\Phi)$ is injective, hence bijective 
for dimensional reasons. 
In the following calculation, we shall use 
the formula 
\begin{equation}
  \label{eq:pmap1}
 T_x(\Exp_{eH})y = \exp(x).\frac{\sinh(\ad x)}{\ad x} y 
\quad \mbox{ for } \quad x,y \in \fq
\end{equation}
for the differential of $\Exp$ 
(\cite[Lemma~4.6]{DN93}), where 
$\fq \to T_{\Exp_{eH}(x)}, v \mapsto \exp x.v$, is the linear isomorphism 
induced by the action of $\exp x \in G$ on $M$. 
For $a \in \g^\sigma , x,b \in \fq^{-\sigma}$, we obtain 
\begin{align}\label{eq:tangmap}
 T_{[e,x]}(\Phi)(a,b) 
&= a.\Exp(x) + T_{x}(\Exp_{eH})(b) \notag \\ 
&= \exp(x).\Big(p_{\fq}(e^{-\ad x} a) + \frac{\sinh(\ad x)}{\ad x} b\Big)
\end{align}

Note that $e^{-\ad x}a = \cosh (\ad x) a - \sinh (\ad x) a$. 
If  $a\in \fh^{\sigma}$ then $p_\fq(e^{-\ad x}a) = - \sinh(\ad x)a$, and if 
$a \in \fq^\sigma$, then $p_\fq(e^{-\ad x}a) = \cosh(\ad x)a$. 
Writing $a = a_\fh + a_\fq$ with $a_\fh \in \fh^\sigma$ 
and $a_\fq \in \fq^\sigma$, we thus  obtain 
\begin{equation}\label{eq:tangmap2}T_{[e,x]}(\Phi)(a,b) 
= \exp(x).\Big(\underbrace{\cosh(\ad x)a_\fq}_{\in \fq^\sigma} 
+ \underbrace{\frac{\sinh(\ad x)}{\ad x} b
- \sinh(\ad x)a_\fh}_{\in \fq^{-\sigma}}
\Big).
\end{equation}

The following lemma provides a characterization of the regular points. 

\begin{lem} \mlabel{lem:2.3} 
For $x \in \fq$, the following assertions hold: 
  \begin{itemize}
 \item[\rm(a)] $\Exp_{eH}$ is regular in $x$ if and only if the map 
$\frac{\sinh(\ad x)}{\ad x} \: \fq \to \fq$ 
is invertible, which is equivalent to 
\begin{equation}
  \label{eq:2.3.a}
\Spec(\ad x\res_{\fq_L}) \cap  \Z\pi i \subeq \{0\}, 
\quad \mbox{ where } \quad \fq_L := \fq + [\fq,\fq].
\end{equation}
  \item[\rm(b)] If $\Exp_{eH}\res_{\fq^{-\sigma}}$ is regular in 
$x \in \fq^{-\sigma}$, 
then the polar map $\Phi$ in \eqref{eq:pmap1} is regular in 
$[g,x]$ if and only if, in addition, 
$\cosh(\ad x) \: \fq^\sigma \to \fq^\sigma$ 
is invertible, which is equivalent to 
\begin{equation}
  \label{eq:2.3.b}
\Spec(\ad x\res_{\fq_L}) \cap  \Big(\frac{\pi}{2} + \Z\pi\Big) i 
= \eset.
\end{equation}
  \end{itemize}
\end{lem}

\begin{prf} (a) follows from the spectral theoretic description 
of the kernel of $\frac{\sinh(\ad x)}{\ad x}\big|_{\fq}$ as the intersection of 
$\fq$ with the sum of the eigenspaces of $\ad x$ in $\g_\C$ for the 
eigenvalues $\lambda \in \pi i \Z \setminus \{0\}$ 
(Lemma~\ref{lem:c.1}). 

\nin (b) Suppose that the restriction of $\Exp_{eH}$ to $\fq^{-\sigma}$ 
is regular, i.e., that 
$\frac{\sinh(\ad x)}{\ad x} \: \fq^{-\sigma} \to \fq^{-\sigma}$ 
is invertible. 
Then 
\eqref{eq:tangmap2} shows that $\Phi$ is regular in $[e,x]$ 
if and only if 
$\cosh(\ad x) \: \fq^\sigma \to \fq^\sigma$ 
is invertible, and this is equivalent to the condition on 
$\Spec(\ad x\res_{\fq_L})$ stated in~(b).
\end{prf}

The following lemma contains a wealth of information 
on singular points of the polar map~$\Phi$. 

\begin{lem} \mlabel{lem:2.7} 
Let $\Omega \subeq \fq^{-\sigma}$ be an open 
$H^\sigma$-invariant subset consisting of $\Exp$-regular elliptic elements, 
and consider the polar map 
\[ \Phi \: G^\sigma_e \times_{H^\sigma} \Omega \to M, \quad [g,y] 
\mapsto g.\Exp_{eH}(y).\] 
Let $x \in \Omega$, write 
$m := \Exp_{eH}(x) \in M = G/H$, $\cO_m := G^\sigma_e.m$ for its orbit, and put 
\[ \sigma_x := e^{-2 \ad x}, \qquad \zeta_x := e^{-\ad x} \in \Aut(\g).\] 
Then the following assertions hold: 
\begin{itemize}
\item[\rm(a)] $\g^{-\sigma_x} := \{ y \in \g \: \sigma_x(y) = - y \} 
= \ker(\cosh(\ad x))$. 
\item[\rm(b)] $\fq^{\sigma, -\sigma_x}$ complements the subspace 
$\exp(-x).\im(T_{[e,x]}(\Phi)) = \cosh(\ad x)\fq^\sigma \oplus \fq^{-\sigma}$  
in $\fq$. 
\item[\rm(c)] The eigenspaces $\g^{\pm\sigma_x}$ are $\tau$-invariant, and on 
the Lie subalgebra $\g^{\sigma_x^2} = \g^{\sigma_x} \oplus \g^{-\sigma_x}$, 
the involution $\tau$ commutes with $\sigma_x$. 
\item[\rm(d)] The eigenspaces $\g^{\pm\sigma_x}$ are $\zeta_x$-invariant, 
on $\g^{\sigma_x}$ the automorphisms $\zeta_x$ and $\tau$ commute, 
and on on $\g^{-\sigma_x}$ the complex structure $\zeta_x$ and $\tau$ 
anticommute. In particular, we have 
\[ \zeta_x(\fh^{\sigma_x}) = \fh^{\sigma_x}, \quad 
\zeta_x(\fq^{\sigma_x}) = \fq^{\sigma_x}, \quad 
\zeta_x(\fh^{-\sigma_x}) = \fq^{-\sigma_x}, \quad 
\zeta_x(\fq^{-\sigma_x}) = \fh^{-\sigma_x}.\] 
\item[\rm(e)] The eigenspaces $\g^{\pm\sigma_x}$ are $\sigma$-invariant, and on 
the Lie subalgebra $\g^{\sigma_x^2}$, 
the involution $\sigma$ commutes with $\sigma_x$. 
On $\g^{\sigma_x}$, the automorphisms $\zeta_x$ and $\sigma$ commute, 
and on on $\g^{-\sigma_x}$ the complex structure $\zeta_x$ 
anticommute with  $\sigma$. In particular, we have 
\[ \zeta_x(\g^{\sigma,\sigma_x}) = \g^{\sigma,\sigma_x}, \quad 
\zeta_x(\g^{-\sigma,\sigma_x}) = \fg^{-\sigma, \sigma_x}, \quad 
\zeta_x(\g^{\pm\sigma, -\sigma_x}) = \g^{\mp\sigma, -\sigma_x}.\] 
\item[\rm(f)] The stabilizer Lie algebra of $m$ in $\g$ is 
\[ \g_m = \g^{\tau\sigma_x}= \{ y \in \g \: \sigma_x(y) = \tau(y)\} 
 \quad \mbox{ and } \quad 
\g_m^{\sigma_x^2} =  \fh^{\sigma_x} \oplus \fq^{-\sigma_x}.\]
The stabilizer Lie algebra in $\g^\sigma$ is 
\begin{equation}
  \label{eq:gmsigma}
 \g_m^\sigma =\fh^{\sigma,\sigma_x} \oplus \fq^{\sigma, -\sigma_x}.  
\end{equation}
The stabilizer group $G_m$ acts on $T_m(M) = \exp x.\fq$ 
by $g.(\exp x.y) = \exp x.\big(\Ad(\zeta_x^G(g))y\big)$. 
\item[\rm(g)] The tangent space of the orbit $\cO_m$ is 
\[ T_m(\cO_m) = \exp(x).\big(\cosh(\ad x)\fq^\sigma + [x,\fh^\sigma]\big).\] 
\item[\rm(h)] Let $\fq_x := \fq^{\sigma,-\sigma_x} + \fq^{-\sigma,\sigma_x}$ 
and $\fh_x := [\fq_x,\fq_x]$.  
If the group  $\Inn_\g(\fh_x)$ acts as a relatively compact group on $\fq_x$ and 
$\Inn_\g(\fh_x)\fq^{-\sigma,\sigma_x} = \fq_x$, then 
$m$ is an interior point of $\im(\Phi)$. 
\end{itemize}
\end{lem}

\begin{prf} (a) follows directly from $2\cosh(\ad x) = e^{\ad x} + e^{-\ad x}$. 

\nin (b) With \eqref{eq:tangmap} 
we see that  the image of $T_{[e,x]}(\Phi)$ is the subspace 
\[ \exp(x).\Big(\underbrace{\cosh(\ad x)\fq^{\sigma}}_{\subeq \fq^{\sigma}}
+ 
\underbrace{\frac{\sinh(\ad x)}{\ad x} \fq^{-\sigma}}_{\subeq 
\fq^{-\sigma}}\Big).\]
As $\ad x$ is semisimple, a complement of 
$\exp(-x).\im(T_{[e,x]}(\Phi))$ in $\fq$ is 
\begin{equation}
  \label{eq:kercosh}
\ker\big(\cosh(\ad x)\res_{\fq^{\sigma}}\big)  
= \fq^{\sigma, -\sigma_x}.
\end{equation}

\nin (c) In view of $\tau \sigma_x^2 \tau = \sigma_x^{-2}$, 
the fixed point space 
$\g^{\sigma_x^2}$ is $\tau$-invariant. On this subspace $\sigma_x= \sigma_x^{-1} 
= \tau \sigma_x \tau$, so that $\tau$ preserves the two eigenspaces
$\g^{\pm \sigma_x}$ of $\sigma_x$ on $\g^{\sigma_x^2}$. 

\nin (d) As $\zeta_x$ commutes with $\sigma_x$, 
the eigenspaces $\g^{\pm\sigma_x}$ are $\zeta_x$-invariant. 
Further, (c) implies that, on $\g^{\sigma_x}$, we have 
$\tau \zeta_x \tau = \zeta_x^{-1} = \zeta_x$. 

\nin (e) is shown with similar arguments as (c) and (d). 

\nin (f) In $M = G/H$, the point $m = \Exp_{eH}(\exp x) = \exp x H$ 
is obtained by acting 
with $\exp x$ on the base point $eH$. Therefore its stabilizer group is 
$G_m = \exp x H \exp(-x)$ with the Lie algebra 
\[ \g_m = e^{\ad x} \fh = \Fix(e^{\ad x} \tau e^{-\ad x}) 
= \Fix(\tau e^{-2\ad x}) = \Fix(\tau \sigma_x).\]
Now the $\tau$-invariance of $\g^{\sigma_x^2}$ implies that 
\[ \g^{\sigma_x^2}_m 
= \g_m \cap \g^{\sigma_x^2} 
= \g^{\sigma_x,\tau} \oplus \g^{-\sigma_x,-\tau}
= \fh^{\sigma_x} \oplus \fq^{-\sigma_x}.\] 

To verify \eqref{eq:gmsigma}, let $y = y_\fh+ y_\fq \in \g^\sigma$ with 
$y_\fh \in\fh^\sigma$ and $y_\fq \in \fq^\sigma$. 
Then the corresponding vector field $Y_M$ on $M$ satisfies 
\begin{align} 
  \label{eq:ymm}
 Y_M(m) &= y.m 
= \exp(x).p_{\fq}(e^{-\ad x}y) 
= \exp(x).\Big(\frac{1}{2}(e^{-\ad x}y - \tau(e^{-\ad x}y)\Big) \notag\\
&= \exp(x).\Big(\frac{1}{2}(e^{-\ad x}y - e^{\ad x}\tau(y))\Big)
= \exp(x).\big(\cosh(\ad x)y_\fq - \sinh(\ad x)y_\fh\big).
\end{align} 
Therefore $Y_M(m) = 0$ is equivalent to 
\[ 0 = \underbrace{\cosh(\ad x)y_\fq}_{\in \fq^{\sigma}} - 
\underbrace{\sinh(\ad x) y_\fh}_{\in \fq^{-\sigma}}.\] 
Thus both summands have to vanish, which is equivalent to 
\[ e^{-2\ad x} y_\fq = - y_\fq \quad \mbox{ and }\quad 
e^{-2\ad x} y_\fh =  y_\fh.\] 
This implies \eqref{eq:gmsigma}.

To complete the proof of (f), we note that, 
for $v \in \fq$ and $g \in G_m$, we have 
\begin{equation}
  \label{eq:expxtrans}
 g.(\exp(x).y)
= \exp(x).\big(\Ad(\zeta^G_x(g))y\big),
\end{equation}
where  $\zeta^G_x(G_m) = H$ acts on $T_{eH}(M) \cong \fq$ by 
the adjoint representation. 

\nin (g) From \eqref{eq:ymm} it follows that 
\begin{align*}
\exp(-x).T_m(\cO_m) 
&= \cosh(\ad x) \fq^\sigma +  \sinh(\ad x) \fh^\sigma 
= \cosh(\ad x) \fq^\sigma +  \frac{\sinh(\ad x)}{\ad x} [x,\fh^\sigma]\\
&{\buildrel ! \over =} \cosh(\ad x) \fq^\sigma + [x,\fh^\sigma]. 
\end{align*}
Here the last equality follows from 
$(\ad x)^2 \fh^\sigma \subeq \fh^\sigma$ and the 
$\Exp$-regularity of $x$, which, by Lemma~\ref{lem:2.3}, 
is equivalent to the invertibility 
of $\frac{\sinh(\ad x}{\ad x}$ on $\fq$. 

\nin (h) By \eqref{eq:kercosh}, a natural complement of 
\[ \exp(-x).T_m(\cO_m) = \cosh(\ad x) \fq^\sigma + [x,\fh^\sigma] \] 
in $\fq$ is the subspace 
\[ \fq_x := 
\fq^{\sigma,-\sigma_x} \oplus (\fq^{-\sigma} \cap \ker(\ad x)) 
= \fq^{\sigma,-\sigma_x} \oplus \fq^{-\sigma, \sigma_x},\] 
where the last equality follows from 
\[ \fq^{\sigma_x} = \bigoplus_{n \in \Z}  \ker( (\ad x)^2 + \pi^2 n^2 \id_\fq) 
\quad \mbox{ and }\quad 
\Spec(\ad x) \cap \Z \pi i \subeq \{0\}.\] 

As $x$ is $\Exp_{eH}$-regular and 
\[ \Exp_{eH}(x + y) = \exp(x).\Exp_m(y) \quad 
\mbox{ for }\quad y \in \fq^{-\sigma} \cap \ker(\ad x) 
= \fq^{-\sigma,\sigma_x},\]
the subset 
\[ \Omega := \{ y \in \fq_x \: \Exp_m(\exp(x).y) \in \im(\Phi) \} \] 
contains a $0$-neighborhood $U_0$ in $\fq^{-\sigma,\sigma_x}$. 
Now our assumption 
$\Inn_\g(\fh_x) \fq^{-\sigma,\sigma_x} =\fq_x$ and the relative 
compactness of the group $\Inn_\g(\fh_x)$ on $\fq_x$ imply that 
$U_1 := \Inn_\g(\fh_x) U_0$ is a $0$-neighborhood in $\fq_x$. 

Finally, we obtain from \eqref{eq:gmsigma} 
\[ \zeta_x^{-1}(\fh_x) \subeq 
 \zeta_x^{-1}(\fh^{\sigma,\sigma_x} + \fh^{-\sigma,-\sigma_x}) 
\ {\buildrel {(d),(e)} \over \subeq}\  \fh^{\sigma, \sigma_x} + \fq^{\sigma, -\sigma_x} 
\ {\buildrel \eqref{eq:gmsigma} \over =}\  \g_m^\sigma,\] 
so that 
\[ \Exp_m(\exp(x).U_1) 
= \Exp_m(\exp(x).\Inn_\g(\fh_x) U_0) 
= (G^\sigma_e)_m.\Exp_m(\exp(x).U_0) 
\subeq \im(\Phi).\] 
This means that $\Omega$ is a $0$-neighborhood in $\fq_x$, 
and hence that $\im(\Phi)$ is a neighborhood of $m$ 
because the map 
\[ G^\sigma_e \times \fq_x \to M, \quad 
(g,y) \mapsto g.\Exp_m(\exp(x).y) \] 
has surjective differential in $(e,0)$. 
\end{prf}

\begin{rem} If $\rho_i(\ad x) < \pi/2$, then 
the polar map $\Phi$ in the preceding lemma 
is regular in $[e,x]$ (Lemma~\ref{lem:2.3}). In this 
case $\ker(\cosh(\ad x)) = \g^{-\sigma_x}$ is trivial, so that 
\[ \g_m^\sigma = \fh^{\sigma,\sigma_x} = \fh^\sigma \cap \ker(\ad x) 
= \fz_{\fh^\sigma}(x)\]
follows from Lemma~\ref{lem:2.7}(f).
\end{rem}

The next lemma is useful to verify condition (h) in the preceding lemma. 

\begin{lem} \mlabel{lem:2.4} 
Suppose that $(\g,\tau)$ is a reductive symmetric Lie algebra 
such that $\fq$ consist of elliptic elements,  $\sigma\tau= \tau\sigma$, 
and that $\fa \subeq \fq^{-\sigma}$ 
is a maximal abelian subspace in $\fq$. Then 
\[ \fq^\sigma =[\fh^{-\sigma}, \fq^{-\sigma}], 
\qquad \Inn_\g([\fq,\fq])\fq^{-\sigma} = \fq,\] 
and the group $\Inn_\g([\fq,\fq])$ is compact. 
\end{lem}

\begin{prf} Let $\fa \subset \fq$ be maximal abelian in $\fq$. Then $\fa $ contains $\z (\fq)$. 
We then note that $\fz_\g(\fa) = \fa \oplus \fz_\fh(\fa)$ 
is $\tau$- and $\sigma$-invariant. Further 
$\g = \fz_\g(\fa) \oplus [\fa,\g]$ implies that 
$\fq = \fa \oplus [\fa,\fh].$ As 
\[ [\fa,\fh] 
= [\fa, \fh^\sigma \oplus \fh^{-\sigma}] 
\subeq \fq^{-\sigma} \oplus \fq^{\sigma},\] 
this leads to 
\[ \fq^{\sigma} = [\fa, \fh^{-\sigma}] 
\subeq [\fq^{-\sigma}, \fh^{-\sigma}] \subeq \fq^\sigma.\]  

The second assertion follows from $\Inn_\g([\fq,\fq])\fa = \fq$, and the third 
from the fact that the reductive Lie algebra $\fq + [\fq,\fq]$ 
(it is an ideal of $\g$) is compact. 
\end{prf}

\subsection{Fibers of the polar map} 
\mlabel{subsec:1.3}

For the polar map we have to analyze the relation 
\[ \Exp(x) = g.\Exp(y).\] 
Applying the quadratic representation yields 
\[ \exp(2x) = g \exp(2y) g^\sharp = \exp(2 \Ad(g)y) gg^\sharp.\] 
For $x,y \in \fq^{-\sigma}$ and $g \in G^\sigma$, we also have 
$gg^\sharp \in G^\sigma$, so that 
\[ \exp(4x) 
= \exp(2x) \sigma(\exp(2x))^{-1}
= \exp(4 \Ad(g)y).\] 

If $x$ and $y$ are sufficiently small 
(imaginary spectral radius $< \frac{\pi}{4}$), we thus obtain 
$\Ad(g)y = x$, and thus $gg^\sharp = e$, i.e., 
$\tau(g) = g$.

See \cite[Cor.~7.35]{HN93} for similar arguments.

\subsection{Fibers of $\Exp$} 
\mlabel{subsec:1.2}

Suppose that $x, y \in \fq$ have the same exponential image 
$\Exp(x) = \Exp(y)$ in $M$. We further assume that 
$\Spec(\ad x) \cap i \pi \Z \subeq \{0\}$, so that 
$\Exp$ is regular in $x$. Then we obtain in $G$ the identity 
\[ \exp(2x) = Q(\Exp x) = Q(\Exp y) = \exp(2y),\] 
and since $\Spec(\ad(2x)) \cap 2 \pi i \Z \subeq \{0\}$, 
$\exp$ is regular in $x$. Therefore \cite[Lemma~9.2.31]{HN12} implies that 
\[ [x,y]= 0 \quad \mbox{ and } \quad \exp(2x - 2y) = e.\]
We conclude that $\exp(x-y) = \exp(y-x)$, which leads to 
\[ \Exp(y-x) = \tau_M(\Exp(x-y)) = \exp(y-x)H = \exp(x-y)H = \Exp(x-y),\] 
so that $\Exp(y-x) \in M^\tau$, 
and $\Exp(\R (x-y)) \subeq M$ is a closed geodesic. 

We also conclude that $x-y$ is elliptic 
with $\Spec(\ad(x-y))\subeq \pi i \Z$. 

\begin{lem}
If $\rho_i(\ad x), \rho_i(\ad y) < \pi,$ 
then $\exp(x) = \exp(y)$ implies 
$x-y \in \fz(\g)$. 

If, in addition, $G$ is simply connected 
or $\g$ is semisimple, then $x = y$. 
\end{lem}

\begin{prf} The preceding discussion implies that 
$[x,y] = 0$. Now $\rho_i(\ad (x-y)) < 2\pi$ leads to $\ad(x-y) = 0$, 
i.e., to $x-y \in \fz(\g)$. 
\end{prf}

\section{Quadrics as symmetric spaces} 
\mlabel{app:1}

In this appendix we discuss an important example
of a non-compactly causal symmetric spaces:
$d$-dimensional de Sitter space $\dS^d$, realized as a hyperboloid
in Minkowski space.

\subsection{Quadrics as symmetric spaces} 
\mlabel{app:a.1}

\begin{definition} \mlabel{def:ss} (\cite{Lo69}) (a) Let $M$ be a 
smooth manifold and 
\[ \mu \: M \times M \to M, \quad (x,y) \mapsto x \cdot y =: s_x(y) \] 
be a smooth map with the following properties: 
each $s_x$ is an involution for which $x$ is an  isolated fixed point and 
\begin{equation}
  \label{eq:symspcond}
 s_x(y \cdot z) = s_x(y)\cdot s_x(z) \quad \mbox{ for all } \quad x,y \in M.
\end{equation}
Then we call $(M,\mu)$ a {\it symmetric space}. 

 (b) A morphism of symmetric spaces $M$ and $N$ is a smooth
map $\varphi : M\to N$ such that $\varphi (x\cdot y)= \varphi (x)\cdot 
\varphi (y)$ for $x, y \in M$.

 (c) A {\it geodesic} of a symmetric space is a 
morphism $\gamma \: \R \to M$ of symmetric spaces, i.e., 
\[ \gamma(2x-y) = \gamma(x) \cdot \gamma(y) \quad \mbox{ for } \quad 
x,y \in \R.\] 
Any geodesic is uniquely determined by $\gamma'(0) \in T_{\gamma(0)}(M)$, 
and, conversely, every $v \in T_p(M)$ generates a unique 
geodesic $\gamma_v$ with $\gamma_v(0) = p$ and $\gamma_v'(0) = v$. 
Accordingly, geodesics are encoded in the {\it exponential functions} 
\[ \Exp_p \:  T_p(M) \to M, \quad \Exp_p(v) := \gamma_v(1).\] 
We then have 
\[ \gamma_v(t) = \Exp_p(tv) \quad \mbox{ for } \quad t \in \R.\]
\end{definition}

\begin{ex}
Let $(V,\beta)$ be a finite dimensional real vector space, 
endowed with a non-degenerate symmetric bilinear form~$\beta$. 
Then every anisotropic element $x \in V$ defines an 
involution 
\[ s_x(y) := -y + 2 \frac{\beta(x,y)}{\beta(x,x)} x \]
fixing $x$ and satisfying 
$\Fix(-s_x) = x^\bot = \{ y \in V\: \beta(x,y) = 0\}$. 

For $c \in \R^\times$, let 
\[ Q_c := Q_c(V,\beta) := \{ v \in V \: \beta(v,v) = c\} \] 
denote the corresponding quadric in $(V,\beta)$. 
Then $(Q_c, \mu)$ with 
$\mu(x,y) = s_x(y)$ 
is a symmetric space (Definition~\ref{def:ss}) 
and $\dim Q_c = \dim V -1$. 
Note that dilation by $r \in \R^\times$ is an isomorphism of symmetric spaces 
from $Q_c$ to $Q_{r^2c}$. 

For $p \in Q_c$, the tangent space is $T_p(Q_c) = p^\bot 
= \{ v \in V \: \beta(p,v) = 0\}$. 
To describe the exponential function of the symmetric space $Q_c$, 
we use the entire functions $C, S \: \C \to \C$ defined by 
\begin{equation}
  \label{eq:CandS}
 C(z) := \sum_{k = 0}^\infty \frac{(-1)^k}{(2k)!} z^{k} \quad \mbox{ and } \quad 
 S(z) := \sum_{k = 0}^\infty \frac{(-1)^k}{(2k+1)!} z^{k}
\end{equation}
which satisfy 
\begin{equation}
  \label{eq:csrela}
 \cos z = C(z^2) \quad \mbox{ and } \quad \sin z = z S(z^2) 
\quad \mbox{ for } \quad z \in \C
  \end{equation} 
and 
\begin{equation}
  \label{eq:csrelb}
\cosh z = C(-z^2) \quad \mbox{ and } \quad \sinh z = z S(-z^2) 
\quad \mbox{ for } \quad z \in \C.
\end{equation}
Note that 
\begin{equation}
  \label{eq:csrel1}
1 = C(z)^2 + z S(z)^2 \quad \mbox{ for }\quad z \in \C
\end{equation}
follows from $1 = \cos^2 z + \sin^2 z$ and the surjectivity of the 
square map on $\C$. 
\end{ex}

\begin{prop}
  \mlabel{prop:expfct}
The exponential function of the symmetric space $Q_c$ is given by 
\begin{equation}
 \label{eq:exp-rel}
\Exp_p(v) = C\Big(\frac{\beta(v,v)}{\beta(p,p)}\Big) p 
+ S\Big(\frac{\beta(v,v)}{\beta(p,p)}\Big) v \quad \mbox{ for }  \quad 
p \in Q_c, v \in T_p(Q_c) = p^\bot. 
\end{equation}
\end{prop}

\begin{prf}
To verify this claim, abbreviate 
$\eps := \frac{\beta(v,v)}{\beta(p,p)}$. 
By \eqref{eq:csrel1}, on the right hand 
side of \eqref{eq:exp-rel}, 
the quadratic form $\beta(\cdot,\cdot)$ has the value 
\[ C(\eps)^2 \beta(p,p)
+ S(\eps)^2 \beta(v,v)
= \big(C(\eps)^2 + S(\eps)^2 \eps\big) \beta(p,p) = \beta(p,p) = c.\] 
Therefore the right hand side of \eqref{eq:exp-rel} is contained in $Q_c$. 

It remains to show that, for $\eps \in \{-1,0,1\}$, the curve 
\[ \gamma_v(t) := C(t^2\eps) p + S(t^2\eps) t v, \] 
which satisfies $\gamma_v'(0) = v$, actually is a geodesic, i.e., 
\[ \gamma_v(2t-s) 
= \gamma_v(t) \cdot \gamma_v(s)\quad \mbox{ for } \quad t,s \in \R.\] 
We consider three  cases: \\
\nin {\bf $\eps = 0$:} Then $\gamma_v(t) = p + tv$ and 
\[ \gamma_v(t) \cdot \gamma_v(s) 
= - \gamma_v(s) + \frac{2}{c} \beta(\gamma_v(t),\gamma_v(s)) \gamma_v(t) 
= - (p + sv) + 2(p  + tv) 
= p + (2t-s)v.\] 
\nin {\bf $\eps = 1$:} Then 
\[ \gamma_v(t) = \cos(t) p + \sin(t) v\] 
by \eqref{eq:csrela}, 
and 
\begin{align*}
\gamma_v(t) \cdot \gamma_v(s) 
&= - \gamma_v(s) + \frac{2}{c} \beta(\gamma_v(t),\gamma_v(s)) \gamma_v(t) \\
&= - (\cos(s) p + \sin(s) v) + 
2(\cos(t)\cos(s) + \sin(t) \sin(s)) (\cos(t)p + \sin(t)v) \\
&= (-\cos(s) + 2\cos(t)^2\cos(s)  + 2\sin(t)\sin(s)\cos(t)) p \\
&\quad + (-\sin(s) + 2\sin(t)\cos(t)\cos(s)  +2 \sin(t)^2\sin(s)) v \\ 
&= (\cos(s) - 2\sin(t)^2\cos(s)  + 2\sin(t)\sin(s)\cos(t)) p \\
&\quad + (-\sin(s) + 2\sin(t)\cos(t)\cos(s)  +2 \sin(t)^2\sin(s)) v \\ 
&= \cos(2t-s) p + \sin(2t-s) v = \gamma_v(2t-s).
\end{align*}
\nin {\bf $\eps = -1$:} Then 
\[ \gamma_v(t) = \cosh(t) p + \sinh(t) v\] 
by \eqref{eq:csrelb}, 
and 
\begin{align*}
\gamma_v(t) \cdot \gamma_v(s) 
&= - \gamma_v(s) + \frac{2}{c} \beta(\gamma_v(t),\gamma_v(s)) \gamma_v(t) \\
&= - (\cosh(s) p + \sinh(s) v) + 
2(\cosh(t)\cosh(s) + \sinh(t) \sinh(s)) (\cosh(t)p + \sinh(t)v) \\
&= (-\cosh(s) + 2\cosh(t)^2\cosh(s)  + 2\sinh(t)\sinh(s)\cosh(t)) p \\
&\quad + (-\sinh(s) + 2\sinh(t)\cosh(t)\cosh(s)  +2 \sinh(t)^2\sinh(s)) v \\ 
&= (\cosh(s) +2\sinh(t)^2\cosh(s)  + 2\sinh(t)\sinh(s)\cosh(t)) p \\
&\quad + (-\sinh(s) + 2\sinh(t)\cosh(t)\cosh(s)  +2 \sinh(t)^2\sinh(s)) v \\ 
&= \cosh(2t-s) p + \sinh(2t-s) v = \gamma_v(2t-s).
\qedhere\end{align*}
\end{prf}

\subsection{Minkowski space} 
\mlabel{subsec:a.2}

On $V := \R^{1,d}$, we consider the Lorentzian form 
\[ [x,y] := x_0 y_0 - \bx \by,\] 
the {\it future light cone} 
\[ V_+ := \{ x = (x_0, \bx) \in V \: x_0 > 0, x_0^2 - \bx^2 > 0\}, \] 
and the tube domain 
\[ \cT_V = V + i V_+ \subeq V_\C.\]  
The standard boost vector field $X_h(v) = hv$ is defined by 
$h \in \so_{1,d}(\R)$, given by 
\begin{equation}
  \label{eq:h}
hx = (x_1, x_0, 0,\ldots, 0).
\end{equation}
It generates the flow 
\[ \alpha_t(x) = e^{th} x 
= (\cosh t \cdot x_0 + \sinh t \cdot x_1, 
\cosh t \cdot x_1 + \sinh t \cdot x_0, x_2, \ldots, x_d)\] 
and defines the involution 
\[\tau_h(x) := e^{\pi i h} x= (-x_0, -x_1, x_2, \ldots, x_d), \] 
that we extend to an antilinear involution 
$\oline\tau_h$ on $V_\C$. It also defines a 
Wick rotation 
\[ \kappa_h(x) = e^{-\frac{\pi i}{2} h} x
= (-i x_1, -i x_0, x_2, \ldots, x_d) \] 
satisfying $\kappa_h^2 = \tau_h$.

\begin{lem} \mlabel{lem:minko}
The following subsets of $V$ are equal: 
  \begin{itemize}
  \item[\rm(a)] The {\it standard right wedge} 
$W_R := \{ x \in V \: x_1 > |x_0| \}$. 
\item[\rm(b)] The positivity domain 
$W_V^+(h) := \{ v \in V \: X_h(v) \in V_+ \}$ of $X_h(v) = hv$.
  \item[\rm(c)] The KMS domain of $\alpha$: 
$\{ x \in V \: (\forall t \in (0,\pi))\ \alpha_{it}(x) \in \cT_V\}$.   
\item[\rm(d)] The Wick rotation $\kappa_h(\cT_V^{\oline\tau_h})$ 
of the fixed point set $\cT_V^{\oline\tau_h} 
= \cT_V \cap (V^{\tau_h} + i V^{-\tau_h})$ of the antiholomorphic 
involution $\oline\tau_h$ on $\cT_V$. 
  \end{itemize}
\end{lem}

\begin{prf} The equality of $W_R$ and $W_V^+(h)$ follows  immediately from 
\eqref{eq:h}. For $0 < t < \pi$, we have 
\[ \Im(\alpha_{it}(x)) 
= \Im(\cos t \cdot x_0 + \sin t \cdot i x_1, 
\cos t \cdot x_1 + \sin t \cdot i x_0, x_2, \ldots, x_d) 
= \sin t\cdot (x_1, x_0, 0,\ldots,0).\] 
This implies the equality of the sets in (a) and (c). 
Finally, we observe that 
\[ \cT_V^{\oline\tau_h} = (V + i V_+)^{\oline\tau_h} 
= V^{\tau_h}  + i V_+^{-\tau_h} 
=\{ (ix_0, i x_1, x_2, \ldots, x_d) \in V \: 
x_0 > |x_1|\}\] 
implies that 
$\kappa_h(\cT_V^{\oline\tau_h}) = W_R$. 
\end{prf}

\begin{rem} For the closed light cone $C := \oline{V_+}$ and the 
$h$-eigenspaces 
\[ V_{\pm 1}(h) = \R (\be_1 \pm \be_0),
\quad \mbox{ we put } \quad 
C_\pm := \pm C \cap V_{\pm 1}(h) 
= [0,\infty) (\be_1 \pm \be_0),\] 
so that we obtain the following description of the standard right wedge 
\[ W_R = W_V^+(h) 
= V_0(h) + C_+^\circ + C_-^\circ 
= V_0(h) + \R_+ (\be_1 + \be_0) + \R_+ (\be_1 - \be_0).\] 
\end{rem}

\subsection{De Sitter space $\dS^d, d \geq 2$} 
\mlabel{subsec:a.3}

We write $G := \SO_{1,d}(\R)^\up$ for the connected Lorentz group 
acting on Minkowski space $(V,[\cdot,\cdot])$ and consider the 
{\it de Sitter space} 
\[ M := \dS^d := \{ x \in V \: [x,x] = -1\} = G.\be_1.\]
The purpose of this appendix is derive our main results
  for the example $\dS^d$ by direct calculations without the elaborate
  structure theory developed for the general case.  
By Proposition~\ref{prop:expfct}, the 
exponential function of the symmetric space $\dS^d$ is given by 
\[ \Exp_p(v) = C(-[v,v]) p + S(-[v,v]) v,\] 
so that spacelike vectors $v$ generate closed geodesics. 

De Sitter space inherits the structure of a 
non-compactly causal symmetric space from the embedding into 
Minkowski space $(V,V_+)$. In particular, the positive cone 
in a point $x \in \dS^d$ is given by 
\[ V_+(x) := V_+ \cap T_x(\dS^d) = V_+ \cap x^\bot.\] 

We keep the notation $h$, $\alpha_t$ and $X_h$ from our discussion 
of Minkowski space in Subsection~\ref{subsec:a.2}. 
As $\dS^d$ is $\alpha$-invariant, 
the vector field $X_h$ is tangential to $\dS^d$, 
and Lemma~\ref{lem:minko} immediately implies that 
its {\it positivity domain} is given by 
\[ W_{\dS^d}^+(h)
:= \{ x \in \dS^d \:  X_h(x) \in V_+(x) \} 
= W_R \cap \dS^d = \{ x \in \dS^d \: x_1 > |x_0|\}.\] 
The centralizer of $h$ in $G$ is the subgroup 
\begin{equation}
  \label{eq:gh-lorentz}
G^h = \exp(\R h) \SO_{d-1}(\R) \cong 
\SO_{1,1}(\R)^\up \times \SO_{d-1}(\R) 
\end{equation}
(cf.\ \cite[Lemma~4.12]{NO17}).

\begin{lem} \mlabel{lem:a.4} 
For 
\[ V_+^\pi(\be_1) := \{ x \in T_{\be_1}(\dS^d) \cong \be_1^\bot 
\: x_0 > 0, 0 < [x,x] < \pi^2 \}, \]
we have 
\begin{equation}
  \label{eq:desit-tube}
 (V + i V_+) \cap \dS^d_\C 
= \{ z \in V_\C \: [z,z] = -1, \Im  z \in V_+\} 
=  G.\Exp_{\be_1}(iV_+^\pi(\be_1)).
\end{equation}
\end{lem}

\begin{prf} First we observe that both sides of \eqref{eq:desit-tube} 
are $G$-invariant. 

\nin ``$\supeq$'':  We have to show that any $x \in V_+^\pi(\be_1)$ satisfies 
$\Exp_{\be_1}(ix) \in V + i V_+$. Since the orbit of $x$ under 
$G_{\be_1}$ contains an element in $\R \be_0$, we may assume that 
$x = x_0 \be_0$ with $0 < x_0 < \pi$. Then 
\[ \Exp_{\be_1}(ix) = \Exp_{\be_1}(ix_0 \be_0) 
= C(x_0^2) \be_1 + S(x_0^2)  x_0 i \be_0 
= \cos(x_0) \be_1 + \sin(x_0) i \be_0  \in V + i V_+ \] 
follows from $\sin(x_0) > 0$. 

\nin ``$\subeq$'': If $z = x + i y\in \dS^d_\C \cap (V + i V_+)$, then 
$[z,z] = -1$ and $y \in V_+$. Acting with $G$, we may thus assume that 
$y = y_0 \be_0$ with $y_0 > 0$. Then $x \in y^\bot = \be_0^\bot$ follows from 
$\Im [z,z] = 0$, so that we 
may further assume that $x = x_1 \be_1$ for some $x_1 \geq 0$. 
Now $z = i y_0 \be_0 + x_1 \be_1 \in \dS^d_\C$ implies that 
\[ -1 = [z,z] = - y_0^2 - x_1^2.\] 
Hence there exists a $t \in (0,\pi)$ with 
$y_0 = \sin t$ and $x_1 = \cos t$. This leads to 
\[ \Exp_{\be_1}(i t \be_0) = 
\cos(t) \be_1 + \sin(t) i \be_0 
= x_1 \be_1 + y_0 i \be_0 = z.
\qedhere\] 
\end{prf}

\begin{defn} \mlabel{def:desittube}
The complex manifold 
\begin{equation}
  \label{eq:desit-tubes}
\cT_M = \cT_{\dS^d} := G.\Exp_{\be_1}(i V_+^\pi(\be_1)) = \dS^d_\C \cap 
(V + i V_+) 
\end{equation}
is called the {\it tube domain of $\dS^d$}. If coincides with 
$G.\Exp_m(i V_+^\pi(m))$ for every $m \in \dS^d$. 
\end{defn}

Starting from the relation $\cT_{\dS^d} = G.\Exp_{\be_2}(i V_+^\pi(\be_2))$, 
the invariance of $\be_2$ under $\tau_h = e^{\pi i h}$ 
permits us to obtain a nice description of the fixed point set of
$\oline\tau_h$ on $\cT_{\dS^d}$. 

\begin{lem} \mlabel{lem:a.9}
$\cT_{\dS^d}^{\oline\tau_h} 
= G^h.\Exp_{\be_2}(i V_+^\pi(\be_2)^{-\tau_h}).$
\end{lem}

\begin{prf} We clearly have 
\[ \cT_{\dS^d}^{\oline\tau_h} 
\supeq G^{\tau_h}.\Exp_{\be_2}(i V_+^\pi(\be_2)^{-\tau_h}).\] 
We also note that 
\[ V_+^\pi(\be_2)^{-\tau_h} 
= V_+^\pi(\be_2) \cap (\R \be_0 + \R \be_1) 
= \{ x_0 \be_0 + x_1 \be_1 \:  x_0 > 0, 0 < x_0^2 - x_1^2 < \pi^2\}.\] 
Next we observe that 
\[ \cT_{\dS^d}^{\oline\tau_h}
= (V + i V_+)^{\oline \tau_h} \cap \dS^d_\C 
= (V^{\tau_h} + i V_+^{-\tau_h}) \cap \dS^d_\C
= \{ (i x_0, i x_1, x_2, \ldots, x_d) \in \dS^d_\C 
\: x_0 > |x_1|\}.\] 
Hence the elements $z \in \cT_{\dS^d}^{\oline\tau_h}$ are of the form 
\[ z  = (ix_0, ix_1, x_2, \ldots, x_d), \quad x_0 > |x_1|, \ 
x_0^2 - x_1^2 + x_2^2 + \cdots +  x_d^2 = 1.\] 
Therefore the orbit of $z$ under $G^h = \exp(\R h) \SO_{d-1}(\R)$ 
(see \eqref{eq:gh-lorentz}) 
contains an element of the form $y = (iy_0, 0, y_2,0,\ldots,0)$ 
with $y_0 > 0$ and $y_0^2 + y_2^2 = 1$. 
Hence there exists a $t \in (0,\pi)$ with 
$y_0 = \sin(t)$ and $y_2 = \cos(t)$. Then 
\[ \Exp_{\be_2} (it\be_0) = \cos(t) \be_2 + \sin(t) i \be_0 
= y_2 \be_2 + i y_0 \be_0 = y.\] 
This implies that 
\[ \cT_{\dS^d}^{\oline\tau_h} 
\subeq G^h.\Exp_{\be_2}(i (0,\pi) \be_0) = G^h.\Exp_{\be_2}(i V_+^\pi(\be_2)^{-\tau_h}).
\qedhere\] 
\end{prf}

\begin{prop} \mlabel{prop:desit}
The following subsets of de Sitter space $M =\dS^d$ are equal: 
  \begin{itemize}
  \item[\rm(a)] $W_R \cap \dS^d = \{ x \in \dS^d \: x_1 > |x_0|\}$. 
\item[\rm(b)] $W_M^+(h) := \{ x \in \dS^d  \: X_h(x) \in V_+(x) \}$ 
(the positivity domain of $X_h$). 
\item[\rm(c)] 
$W_M^{\rm KMS}(h) = \{ x \in \dS^d \: (\forall t \in (0,\pi))
\ \alpha_{it}(x) \in \cT_M\}$ 
(the KMS domain of $\alpha$). 
 \item[\rm(d)] $\kappa_h(\cT_M^{\oline\tau_h})$. 
 \item[\rm(e)] $W_M(h) = G^h.\Exp_{\be_2}((C_+^\circ + C_-^\circ)^\pi)$. 
  \end{itemize}
\end{prop}

\begin{prf} The equality of the sets under (a) and (b) 
follows from Lemma~\ref{lem:minko}. 

As $\alpha_{it}(m) \in M_\C$ for every $t \in \R$ and $m \in M$, 
Lemma~\ref{lem:a.4} shows that the set under (c) coincides 
with $W_V^+(h) \cap M = W_M ^+(h)$. 

From the proof of Lemma~\ref{lem:a.9}, we recall that 
\[ \cT_{\dS^d}^{\oline\tau_h}
= \{ (i x_0, i x_1, x_2, \ldots, x_d) \in \dS^d_\C 
\: x_0 > |x_1|\}.\] 
Now 
$\kappa_h(x) = (-i x_1, -i x_0, x_2, \ldots, x_d)$ implies
\[ \kappa_h(\cT_{\dS^d}^{\oline\tau_h})
= \{ (x_0, x_1, x_2, \ldots, x_d) \in \dS^d 
\: x_1 > |x_0|\} = W_R \cap \dS^d = W_M(h).\] 
Combining Lemma~\ref{lem:a.9} with (d), we finally obtain 
\begin{align*}
W_M^+(h) 
&= \kappa_h(\cT_{\dS^d}^{\oline\tau_h}) 
= G^h.\Exp_{\be_2}( \kappa_h(iV_+^\pi(\be_2)^{-\tau_h}))
= G^h.\Exp_{\be_2}((C_+^\circ + C_-^\circ)^\pi)
\end{align*}
for $C_\pm^\circ = \R_+ (\be_1 \pm \be_0).$
\end{prf}

\begin{rem}
For the $\alpha$-fixed base point $\be_2 \in \dS^d$ the cone  
\[ W_R(\be_2) := W_R \cap T_{\be_2}(\dS^d) 
= C_+^\circ + C_-^\circ + T_{\be_2}(\dS^d)^{\tau_h}, 
\quad \mbox{ where } \quad  C_\pm^\circ = \R_+ (\be_1 \pm \be_0)\] 
is an infinitesimal version of the wedge domain in $M = \dS^d$. 
\end{rem}
\medskip

\noindent
{\bf Data availability:} Data sharing is not applicable to this article as no 
dataset were used or generated in the article.


\begin{thebibliography}{aaaaaaaa}

\bibitem[AG90]{AG90} Akhiezer, D. N., and S. G. Gindikin, {\it On Stein 
extensions of real symmetric spaces}, Math. Ann. {\bf 286} (1990), 1--12

\bibitem[AW63]{AW63} Araki, H., and E. J. Woods, {\it 
Representations of the canonical commutation relations describing 
a nonrelativistic infinite free Bose gas}, 
J. Math. Phys. {\bf 4} (1963), 637--662 

\bibitem[BJL02]{BJL02} Baumg\"artel, H., Jurke, M., and F. Lledo, 
{\it Twisted duality of the CAR-algebra}, J.~Math. Physics {\bf 43:8} (2002), 
4158--4179

\bibitem[Bo09]{Bo09} Borchers, H.-J., {\it On the net of von Neumann
  algebras assiciated with a wedge and wedge-causal manifold},
  Preprint, 2009; available at http://www.theorie.physik.uni-goettingen.de/
forschung/qft/publications/2009

\bibitem[BB99]{BB99} Borchers, H.-J., and D.~Buchholz, {\it Global
  properties of vacuum states in de Sitter space},  Annales
Poincare Phys. Theor. {\bf A70} (1999), 23--40 

\bibitem[BR87]{BR87} Bratteli, O., and D.~W.~Robinson, ``Operator Algebras and Quantum Statistical
Mechanics I,'' 2nd ed.,
Texts and Monographs in Physics, Springer-Verlag, 1987 

\bibitem[BM96]{BM96} Bros, J., and U. Moschella, {\it Two-point functions and quantum 
fields in de Sitter universe}, Rev. Math. Phys. {\bf 8} (1996), 327--391

\bibitem[BMS01]{BMS01} Buchholz, D., Mund, J., and S.J.~Summers, {\it
  Transplantation of local nets and geometric modular 
  action on Robertson-Walker Space-Times},
  Fields Inst. Commun. {\bf 30} (2001), 65--81 
  
\bibitem[BS04]{BS04} Buchholz D., and S.J.~Summers, {\it Stable
  quantum systems in anti-de Sitter space: Causality, inde-
  pendence and spectral properties},  J. Math. Phys. {\bf 45} (2004),
  4810--4831 
  
\bibitem[CLRR22]{CLRR22} Ciolli, F., Longo, R., Ranallo, A., and G.~Ruzzi, {\it 
Relative entropy and curved spacetimes}, 
J. Geom. Phys. {\bf 172} (2022), Paper No. 104416, 16 pp

\bibitem[DLM11]{DLM11} Dappiaggi, C., Lechner, G., and
  E. Morfa-Morales, {\it Deformations of quantum field theories on
    spacetimes with Killing vector fields},
  Comm. Math. Phys. {\bf 305:1} (2011), 99--130

\bibitem[DN93]{DN93} D\"orr, N., and K.-H. Neeb, {\it On 
wedges in Lie triple systems and ordered 
symmetric spaces}, Geometriae Dedicata {\bf 46} (1993), 1--34  

\bibitem[FK94]{FK94} Faraut, J., and A. Koranyi, ``Analysis on Symmetric
  Cones,'' Oxford Math.\ Monographs, Oxford University Press, 1994 

\bibitem[GK02]{GK02} Gindikin, S., and B. Kr\"otz, {\it Complex crowns 
of Riemannian symmetric spaces and non-compactly causal symmetric spaces}, 
Trans. Amer. Math. Soc. {\bf 354:8} (2002), 3299--3327

\bibitem[Ha96]{Ha96} Haag, R., ``Local Quantum Physics. Fields, Particles, Algebras,'' 
Second edition, Texts and Monographs in Physics,  Springer-Verlag, Berlin, 1996


\bibitem[Hel78]{Hel78} Helgason, S., ``Differential Geometry, Lie Groups, 
and Symmetric Spa\-ces,'' Acad. Press, London, 1978


\bibitem[HN93]{HN93} Hilgert, J., and K.-H. Neeb, 
``Lie Semigroups and Their Applications,'' Lecture Notes in 
Math.\ {\bf 1552}, Springer Verlag, Berlin,
Heidelberg, New York, 1993 

\bibitem[HN12]{HN12} Hilgert, J., and K.-H. Neeb, 
``Structure and Geometry of Lie Groups,'' Springer, 2012

\bibitem[H{\'O}97]{HO97} Hilgert, J., and G. {\'O}lafsson, {\it Causal Symmetric Spaces, Geometry and Harmonic
Analysis}, Perspectives in Mathematics {\bf 18}, Academic Press, 1997


\bibitem[Kr07]{Kr07} Kr\"otz, B., {\it
    Corner view on the crown domain},
  Japan. J.  Math. {\bf 2} (2007), 303--311

\bibitem[KN96]{KN96} Kr\"otz, B., and K.-H. Neeb, {\it On 
hyperbolic cones and mixed symmetric spaces}, 
J. Lie Theory {\bf 6:1} (1996), 69--146

\bibitem[KS05]{KS05} Kr\"otz, B and R. J. Stanton, 
{\it Holomorphic extensions of representations. II. 
Geometry and harmonic analysis}, Geom. Funct. Anal. {\bf 15:1} (2005), 190--245 

\bibitem[LR08]{LR08} Lauridsen-Ribeiro, P., ``Structural and Dynamical
  Aspects of the AdS/CFT Correspondence: a
  Rigorous Approach,'' PhD thesis, Sao Paulo, 2007,
available at http://arxiv.org/abs/0712.0401v3 [math-ph], 2008

\bibitem[Le15]{Le15} Lechner, G., {\it Algebraic Constructive Quantum 
Field Theory: Integrable Models and Deformation Techniques}, 
in ``Advances in Algebraic Quantum Field Theory,'' Eds. R.~Brunetti et al; 
Math. Phys. Stud., Springer, 2015, 397--449, arXiv:math.ph:1503.03822
 
\bibitem[Lo69]{Lo69} Loos, O., ``Symmetric Spaces I: General Theory,'' 
W. A. Benjamin, Inc., New York, Amsterdam, 1969

\bibitem[Ma97]{Ma97} Matsuki, T., {\it Double coset decompositions 
of reductive Lie groups arising from two involutions}, 
J. Algebra {\bf 197:1} (1997), 49--91

\bibitem[MN21]{MN21} Morinelli, V., and K.-H. Neeb, {\it 
Covariant homogeneous nets of standard subspaces}, 
Comm. Math. Phys. {\bf 386} (2021), 305--358; 
arXiv:math-ph.2010.07128 

\bibitem[MN\'O22a]{MNO22a} Morinelli, V., K.-H. Neeb, and G.\ \'Olafsson, 
  {\it From Euler elements and $3$-gradings to 
    non-compactly causal symmetric spaces}, J. Lie Theory
  {\bf 33:1} (2023), to appear; arXiv:2207.14034
  
\bibitem[MN\'O22b]{MNO22b} Morinelli, V., K.-H. Neeb, and G. \'Olafsson, {\it 
Modular geodesics and wedge domains  
in general non-compactly causal symmetric spaces}, 
in preparation 

\bibitem[Ne99]{Ne99} Neeb, K.-H., 
{\it On the complex geometry of invariant domains 
in complexified symmetric spaces}, 
Annales de l'Inst.\ Fourier {\bf 49:1} (1999), 177--225

\bibitem[Ne00]{Ne00} Neeb, K.-H., ``Holomorphy and Convexity in Lie Theory'', 
Expositions in Mathematics {\bf 28}, de Gruyter Verlag, 2000

\bibitem[Ne10]{Ne10} Neeb, K.-H.,  {\it 
On differentiable vectors for representations of infinite dimensional 
Lie groups}, J. Funct. Anal. {\bf 259} (2010), 2814--2855


\bibitem[N\'O17]{NO17} Neeb, K.-H., and G.\, \'Olafsson,  {\it 
Antiunitary representations and modular theory}, 
in ``50th Sophus Lie Seminar'', Eds. K. Grabowska et al, 
Banach Center Publications {\bf 113}; pp.~291--362; 
arXiv:math-RT:1704.01336 

\bibitem[N\'O20]{NO20} Neeb, K.-H., and G.\, \'Olafsson, 
{\it Reflection positivity on spheres}, 
Analysis and Mathematical Physics  {\bf 10:1} (2020), Art. 9, 59 pp

\bibitem[N\'O21]{NO21} Neeb, K.-H., and G.\, \'Olafsson, 
  {\it Wedge domains in compactly causal symmetric spaces},
Int. Math. Res. Notices, to appear; arXiv:2107.13288 

\bibitem[N\'O\O21]{NOO21} Neeb, K.-H., G.\, \'Olafsson, and B. \O{}rsted, 
{\it Standard subspaces of Hilbert spaces of holomorphic 
functions on tube domains}, Communications in Math. Phys. 
{\bf 386} (2021), 1437--1487; arXiv:2007.14797 


\bibitem[Ne{\'O}00]{NeO00}  Neumann, A., G.~\'Olafsson,
  {\it Minimal and maximal semigroups related to causal
symmetric spaces}, Semigroup Forum \textbf{61} (2000), 57--85

\bibitem[Oeh20]{Oeh20} Oeh, D., {\it  
Classification of 3-graded causal subalgebras of real simple Lie algebras}, 
arXiv:math.RT:2001.03125 

\bibitem[\'O91]{Ol91} \'Olafsson, G., {\it Symmetric spaces of hermitian 
type}, Diff. Geom. and its Applications {\bf 1} (1991), 195--233

\bibitem[\'O\O91]{OO91} \'Olafsson, G., and   {\O}rsted, B., {\it The holomorphic discrete series of an affine
symmetric space and representations with reproducing kernels}, Trans.
  Amer. Math. Soc. {\bf 326} (1991), 385--405
  
\bibitem[Ol82]{Ol82} G. I. Olshanski, {\it 
Invariant orderings in simple Lie groups. Solution of a problem of E.B.~Vinberg}, 
Funktsional. Anal. i Prilozhen. {\bf 16:4} (1982),  80--81

\bibitem[Sa80]{Sa80} Satake, I., ``Algebraic Structures of Symmetric Domains,'' 
Publ.\ Math. Soc. Japan {\bf 14}, Princeton Univ. Press, 1980

\bibitem[Schr97]{Schr97} Schroer, B., {\it Wigner representation theory of the Poincar\'e group, 
localization, statistics and the S-matrix},  Nuclear Phys. {\bf B 499-3} (1997), 519--546

\bibitem[Si74]{Si74} Simon, B., ``The $P(\Phi)_2$ Euclidean (Quantum) Field Theory'', 
Princeton Univ. Press, 1974 

\bibitem[TW97]{TW97} Thomas, L.J., and E.H.~Wichmann, {\it On the
  causal structure of Minkowski spacetime}, J. Math.
  Phys. {\bf 38} (1997), 5044--5086 

\bibitem[Wo72]{Wo72} Wolf, J., 
{\it The fine structure of Hermitian symmetric spaces},
in ``Symmetric spaces,'' W.~Boothby and G.~Weiss eds., Marcel Dekker, 1972


\end{thebibliography}
\end{document}